\documentclass[10pt,a4paper,draft]{article}

\usepackage[utf8]{inputenc}
\usepackage[T1]{fontenc}
\usepackage[english]{babel}
\usepackage{lmodern}
\usepackage[normalem]{ulem} 
\DeclareFontFamily{OT1}{mathc}{}
\DeclareFontShape{OT1}{mathc}{m}{it}{<-> mathc10}{}
\DeclareMathAlphabet{\mathcalx}{OT1}{mathc}{m}{it}

\usepackage{suffix} 

\usepackage[a4paper, left=20mm, right=20mm, top=2cm, bottom=2cm]{geometry}
\usepackage{setspace}
\addto\extrasenglish{%
}
\newcommand{\appendixnumberline}[1]{Appendix #1\quad}
\usepackage{amsmath,amsthm,mathdots,amssymb,bm,bbm,mathrsfs,mathtools,stmaryrd}
\usepackage{nicefrac}
\usepackage[final]{graphicx} 
\usepackage{subcaption}
\usepackage{tikz,pgfplots}\pgfplotsset{compat=1.16}
\usepackage{booktabs}

\usepackage{thm-restate}
\newtheorem{theorem}{Theorem}
\newtheorem{corollary}{Corollary}
\newtheorem{lemma}{Lemma}

\newtheorem{fact}{Fact}[section]

\newtheorem{rmk}{Remark}
\newtheorem{defn}{Definition}

\usepackage{comment}

\usepackage[colorlinks=true, urlcolor=violet, linkcolor=blue, citecolor=red, hyperindex=true, linktocpage=true, pagebackref=true, draft=false]{hyperref}
\usepackage{mleftright}\mleftright 

\usepackage{enumitem}
\usepackage{microtype,color,graphicx}
\usepackage{tikz-cd}
\usepackage{xcolor}
\usepackage{multirow}
\usepackage{makecell}
\usepackage{tablefootnote}

\usepackage[capitalise,nameinlink]{cleveref} 
\makeatletter
\AddToHook{cmd/appendix/before}{\def\cref@section@alias{appendix}}
\makeatother

\crefname{defn}{Definition}{Definitions}
\Crefname{defn}{Definition}{Definitions}
\crefname{hyp}{Hypothesis}{Hypotheses}
\Crefname{hyp}{Hypothesis}{Hypotheses}
\crefname{conj}{Conjecture}{Conjectures}
\Crefname{conj}{Conjecture}{Conjectures}
\crefname{quest}{Question}{Questions}
\Crefname{quest}{Question}{Questions}
\crefname{fact}{Fact}{Facts}
\Crefname{fact}{Fact}{Facts}
\crefname{eg}{Example}{Examples}
\Crefname{eg}{Example}{Examples}
\crefname{claim}{Claim}{Claims}
\Crefname{claim}{Claim}{Claims}
\crefname{obs}{Observation}{Observations}
\Crefname{obs}{Observation}{Observations}
\crefname{rmk}{Remark}{Remarks}
\Crefname{rmk}{Remark}{Remarks}

\newcommand{\CB}{\mathcal{B}}
\newcommand{\CC}{\mathcal{C}}
\newcommand{\BC}{\mathbb{C}}
\newcommand{\CD}{\mathcal{D}}
\newcommand{\CE}{\mathcal{E}}

\newcommand{\CG}{\mathcal{G}}
\newcommand{\CH}{\mathcal{H}}
\newcommand{\CI}{\mathcal{I}}

\newcommand{\CK}{\mathcal{K}}
\newcommand{\CL}{\mathcal{L}}
\newcommand{\CM}{\mathcal{M}}

\newcommand{\CP}{\mathcal{P}}

\newcommand{\BR}{\mathbb{R}}
\newcommand{\BS}{\mathbb{S}}
\newcommand{\CS}{\mathcal{S}}
\newcommand{\CT}{\mathcal{T}}
\newcommand{\CU}{\mathcal{U}}

\newcommand{\BZ}{\mathbb{Z}}

\newcommand{\vA}{\bm{A}}

\newcommand{\vB}{\bm{B}}

\newcommand{\vC}{\bm{C}}
\newcommand{\vD}{\bm{D}}

\newcommand{\ve}{\bm{e}}

\newcommand{\vg}{\bm{g}}

\newcommand{\vH}{\bm{H}}
\newcommand{\vI}{\bm{I}}
\newcommand{\vJ}{\bm{J}}
\newcommand{\vK}{\bm{K}}

\newcommand{\vL}{\bm{L}}

\newcommand{\vM}{\bm{M}}
\newcommand{\vn}{\bm{n}}

\newcommand{\vO}{\bm{O}}
\newcommand{\vP}{\bm{P}}
\newcommand{\vPi}{\bm{\Pi}}

\newcommand{\vR}{\bm{R}}
\newcommand{\vS}{\bm{S}}
\newcommand{\vT}{\bm{T}}

\newcommand{\vU}{\bm{U}}

\newcommand{\vW}{\bm{W}}
\newcommand{\vX}{\bm{X}}
\newcommand{\vY}{\bm{Y}}
\newcommand{\vZ}{\bm{Z}}
\newcommand{\vsigma}{\bm{ \sigma}}

\newcommand{\vrho}{\bm{ \rho}}

\newcommand{\vertiii}[1]{{\left\vert\kern-0.25ex\left\vert\kern-0.25ex\left\vert #1 \right\vert\kern-0.25ex\right\vert\kern-0.25ex\right\vert}}

\newcommand{\labs}[1]{\left\vert {#1} \right\vert}

\newcommand{\e}{\mathrm{e}}

\newcommand{\ri}{\mathrm{i}}
\newcommand{\rd}{\mathrm{d}}

\newcommand{\nrm}[1]{\left\| #1 \right\|}
\newcommand{\ipc}[2]{\left\langle#1,#2\right\rangle}

\newcommand{\abs}[1]{\left\vert #1 \right\vert}
\WithSuffix\newcommand\abs*[1]{\vert #1 \vert}

\DeclarePairedDelimiterX{\braket}[1]{\langle}{\rangle}{#1}
\DeclarePairedDelimiterX\ketbra[2]{| }{|}{#1 \delimsize\rangle\!\delimsize\langle #2}	
\DeclarePairedDelimiterX\dotp[2]{\langle}{\rangle}{#1, #2}
\newcommand{\bigO}[1]{\mathcal{O}\left( #1 \right)}

\newcommand{\bigOm}[1]{\Omega\left( #1 \right)}
\newcommand{\bigTh}[1]{\Theta\left( #1 \right)}
\DeclareMathAlphabet{\dutchcal}{U}{dutchcal}{m}{n}
\SetMathAlphabet{\dutchcal}{bold}{U}{dutchcal}{b}{n}
\DeclareMathAlphabet{\dutchbcal}{U}{dutchcal}{b}{n}
\newcommand{\comm}[1]{\mathrm{comm}\left(#1\right)}

\newcommand{\Tr}{\operatorname{Tr}}
\newcommand{\ad}{\mathcalx{ad}}

\newcommand{\gap}{\operatorname{gap}}
\newcommand{\Var}{\operatorname{Var}}
\newcommand{\SU}{\mathrm{SU}(2)}
\newcommand{\Hn}{\CH_n}
\newcommand{\Pn}{P_n}
\newcommand{\Rn}{R_n}

\newcommand{\EC}{\CE_{\mathrm{C}}}
\newcommand{\ED}{\CE_{\mathrm{D}}}
\newcommand{\EDM}{\CE_{\mathrm{DM}}}
\newcommand{\EF}{\CE_{\mathrm{F}}}
\newcommand{\Stot}{\bm{S}_{\mathrm{tot}}}
\newcommand{\nr}[1]{\lVert #1\rVert_{\vrho}}
\newcommand{\ip}[2]{\langle #1,#2\rangle_{\vrho}}
\newcommand{\op}[1]{\lVert #1\rVert_{\mathrm{op}}}

\newcommand{\LD}{\CL_{\mathrm{D}}}
\newcommand{\LDM}{\CL_{\mathrm{DM}}}
\newcommand{\LC}{\CL_{\mathrm{C}}}

\newcommand{\LFs}{\CL_{\mathrm{F},\sigma}}

\newcommand{\DD}{\CD_{\mathrm{D}}}

\newcommand{\DC}{\CD_{\mathrm{C}}}

\newcommand{\DFs}{\CD_{\mathrm{F},\sigma}}

\newcommand{\TD}{\CT_{\mathrm{D}}}

\newcommand{\TC}{\CT_{\mathrm{C}}}

\newcommand{\TFs}{\CT_{\mathrm{F},\sigma}}

\newcommand{\Cad}{\CC_{\mathrm{ad}}}

\newcommand{\rmd}{\mathrm{d}}
\newcommand{\rmT}{\mathrm{T}}
\newcommand{\rmx}{\mathrm{x}}
\newcommand{\rmy}{\mathrm{y}}
\newcommand{\rmz}{\mathrm{z}}

\usetikzlibrary{quantikz2}

\begin{document}

\title{The complete phase diagram of detailed-balanced Gibbs samplers\\ for the mean-field Heisenberg model
}
	\author{
		Levente Raikovich\thanks{Alfréd Rényi Institute of Mathematics, Budapest, Hungary. 
			{\tt raikovich.levente@renyi.hu}}
		\and 
		József Mák\thanks{Wigner Institute of Physics, Budapest, Hungary. 
			{\tt mak.jozsef@wigner.hu}}
		\and
        Csaba Czabán\thanks{Alfréd Rényi Institute of Mathematics, Budapest, Hungary. 
			{\tt czaban.csaba@renyi.hu}}
		\and 
		András Gilyén\thanks{Alfréd Rényi Institute of Mathematics, Budapest, Hungary. 
		{\tt gilyen.andras@renyi.hu}}
	}
\maketitle

\begin{abstract}
We study the spectral gap of recently discovered families of KMS-detailed-balanced Gibbs-sampler Lindbladians for the mean-field Heisenberg model on $n$ spins. We show that the spectral gaps of these natural families of Lindbladians admit the same phase diagram: the gap is constant above the critical temperature (i.e., in the disordered phase), scales as $1/\sqrt{n}$ at the critical temperature, and as $1/n$ below the critical temperature (i.e., in the ordered phase). Our result builds on and extends the recent work of Basso et al.~(2026) which characterized this phase diagram for the GNS-detailed-balanced Davies generator, but only partially resolved the gap scaling at the critical temperature. 
\end{abstract}

\setcounter{tocdepth}{2}
\tableofcontents

\section{Introduction}

Markov chain Monte Carlo methods such as the Metropolis algorithm~\cite{metropolis1953StateCalculationsByComputers} and Glauber dynamics~\cite{glauber1963TimeDepIsing} are workhorses of classical statistical physics. A recent class of analogous algorithms prepare quantum Gibbs states~\cite{temme2009QuantumMetropolis,rall2023ThermalStatePrepRounding,chen2025EfficientQThermSim,chen2023ExactQGibbsSampler,jiang2024QMetropolisWeakMeas,ding2024GibbsSamplingViaKMS,gilyen2024QGeneralGlauberMetropolis}; most recent proposals implement a Lindbladian whose fixed point is the Gibbs state. Essentially, they generalize the Davies generator~\cite{davies1974MasterEquationI,davies1976MasterEquationII}, a physical Lindbladian describing a system weakly coupled to a thermal bath. Since the Davies generator requires resolving the Bohr frequencies exactly, efficiently implementable KMS-detailed-balanced alternatives have been developed: the operator Fourier transform (OFT) construction~\cite{chen2025EfficientQThermSim,chen2023ExactQGibbsSampler} resolves the transferred energy only up to a Gaussian uncertainty, while the coherent reweighting scheme~\cite{gilyen2024QGeneralGlauberMetropolis,ding2024GibbsSamplingViaKMS} filters the jump operators without estimating the energy at all. These Gibbs samplers can be viewed as quantum generalizations of the classical Glauber and Metropolis algorithms~\cite{gilyen2024QGeneralGlauberMetropolis}.

Rigorous mixing-time bounds for quantum Gibbs samplers are known for commuting Hamiltonians~\cite{temme2013LowerBoundDaviesGap} and at high temperatures~\cite{rouze2024EffThermalizationGibbsSamp,rouze2024OptimalQAlgorithmGibbs}, but these results avoid thermal phase transitions. A tractable non-commuting model with a genuine phase transition is the mean-field Heisenberg ferromagnet, whose Gibbs states break the $\SU$ symmetry below $\beta=2$~\cite{basso2026SpecGapDaviesMeanFieldHeisenberg}. In a recent paper Basso \textit{et al.}~\cite{basso2026SpecGapDaviesMeanFieldHeisenberg} showed that the gap of the Davies generator with single-site Pauli jumps is $\Theta(1)$ for $\beta<2$ and $\Theta(n^{-1})$ for $\beta>2$, leaving only $\Omega(n^{-1})\leq\gap\leq\bigO{n^{-1/2}}$ at $\beta=2$.

We show that this dynamical phase diagram is universal: for the coherent Lindbladian with any normalized KMS filter, the OFT Lindbladian with any fixed width $\sigma>0$, and Davies generators with arbitrary KMS weights, the gap is $\Theta(1)$ for $\beta<2$, $\Theta(n^{-1/2})$ at $\beta=2$, and $\Theta(n^{-1})$ for $\beta\in(2,\infty]$ (\Cref{thm:allT}); at $\beta=\infty$ the mixing time is moreover $\Theta(n\log(1/\varepsilon))$. This resolves the critical case also for the Davies generator, and shows that the quantum model mirrors the classical mean-field models with continuous symmetry. Following~\cite{basso2026SpecGapDaviesMeanFieldHeisenberg}, we decompose the observables according to the $\SU\times S_n$ symmetry. On the operators commuting with $\vH_n$ all considered Dirichlet forms are equivalent to the Davies one, while on the complements we compare the coherent Lindbladian with the interchange process, using Aldous' spectral gap theorem~\cite{caputo2010AldousSpectralGapConj}, and with the collective Casimir operator. Finally, we transfer all bounds to the OFT Lindbladians by writing their Dirichlet form as a mixture of coherent ones, a decomposition valid for arbitrary Hamiltonians.

For the readers' convenience we summarize the notation and abbreviations used in this paper in \Cref{sec:recap_notation}.

Let $n\ge2$; the Hamiltonian $\vH_n$ of the mean-field Heisenberg model on the $n$-qubit Hibert space $\CH_n=(\mathbb C^2)^{\otimes n}$ is
\begin{equation}
    \label{eq:mean-field-heisenberg-model}
    \vH_n=-\frac1n\sum_{\alpha\in\{\rmx,\rmy,\rmz\}}(\Stot^\alpha)^2, \quad\text{ where }\quad
    \Stot^\alpha=\frac12\sum_{i=1}^n\vsigma^{(i,\alpha)}, 
\end{equation}
and $\vsigma^{(i,\alpha)}$ denotes the Pauli $\alpha\in\{\rmx,\rmy,\rmz\}$ operator on the $i$-th qubit.

\subsection{Our main result about KMS-detailed-balanced Gibbs Samplers}
We study dynamical means of preparing the Gibbs state for inverse temperature $\beta\in[0,\infty]$ \begin{equation}
    \label{eq:Gibbs-state}
    \vrho=\vrho_{n,\beta}=\frac{e^{-\beta \vH_n}}{\Tr(e^{-\beta \vH_n})} = \frac{1}{Z_{n,\beta}} e^{-\beta \vH_n};
\end{equation} we denote and define the partition function as $Z_{n,\beta} = \Tr(e^{-\beta \vH_n})$. The $\beta = \infty$ (zero-temperature) state is understood as the maximally mixed state on the ground space of $\vH_n$.

We will use various superoperators $\CM\colon \CB(\CH_n)\to\CB(\CH_n)$ to describe and analyse the evolution of the quantum states throughout the processes. A key tool for ensuring convergence towards the Gibbs state is KMS detailed balance, which is defined as self-adjointness of $\CM$ with respect to the KMS inner product
\[
\ip{\vX}{\vY}=\ipc{\vrho^{1/2}\vX\vrho^{1/2}}{\vY}=\Tr(\vrho^{1/2}\vX^\dagger\vrho^{1/2}\vY),\qquad
\nr{\vX}=\|\vrho^{1/4}\vX\vrho^{1/4}\|_2,
\]
where we also used the Hilber-Schmidt inner product 
\[
\ipc{\vX}{\vY}=\Tr(\vX^\dagger\vY),\qquad
\|\vX\|_2^2=\Tr(\vX^\dagger \vX).
\]
We define $\CM^\dagger$ as the adjoint of a superoperator $\CM$ with respect the Hilbert-Schmidt inner product.

The natural KMS-detailed-balanced Lindbladians\footnote{Note that here we work with Lindbladians in the Heisenberg picture, similarly to~\cite{basso2026SpecGapDaviesMeanFieldHeisenberg}, therefore the Lindbladians in~\cite{gilyen2024QGeneralGlauberMetropolis,chen2025EfficientQThermSim,ding2024GibbsSamplingViaKMS} correspond to the adjoint of the $\CL$ defined in \autoref{eq:KMS-Lindbladian-general-form}.} \begin{equation}
    \label{eq:KMS-Lindbladian-general-form}
    \CL[\cdot]=\CT'[\cdot]-\frac12\{\CT'[\vI],\cdot\}
\end{equation}
that we study originate from the completely positive single-site jump process
\begin{equation*}
\CT[\cdot]=\sum_{i\in[n]}\sum_{\alpha\in\{\rmx,\rmy,\rmz\}}\vsigma^{(i,\alpha)}[\cdot]\vsigma^{(i,\alpha)},
\end{equation*}
where 
$$\CT'=\BS^{(f)}\llbracket\CT\rrbracket$$
is a KMS-detailed-balanced completely positive superoperator constructied via a specific, carefully chosen linear super-superoperator $\BS^{(f)}\llbracket\cdot\rrbracket$ corresponding to some continuous function $f\colon \BR \to \BC$
satisfying
\begin{equation}
    \label{eq:f-KMS}
    \begin{aligned}
        &f(\nu) = e^{-\frac{\nu}{2}}f^*(-\nu) \quad \text{(KMS condition)}\\
        &|f| \leq 1, \quad c_\star:=\inf_{\nu\leq 0}|f(\nu)|^2>0 \quad \text{(normalization).}
    \end{aligned}    
\end{equation}

Here we study three different types of Lindbladians of the form \eqref{eq:KMS-Lindbladian-general-form}. We highlight for the reader that here we give the Lindbladians in the Heisenberg picture, as opposed to the usual Schrödigner picture literature, e.g., \cite{gilyen2024QGeneralGlauberMetropolis,chen2025EfficientQThermSim}. Stated differently, here we work with the Hilbert--Schmidt adjoints of the Lindbladians described in the previous papers. Another key observation is that due to the rich symmetries of $\vH_n$ and the jump operators the below Lindbladians need no coherent correction term to achieve KMS detailed balance~\cite{chen2023ExactQGibbsSampler,gilyen2024QGeneralGlauberMetropolis}, see \autoref{apx:coherent-term-zero} for a detailed argument.  \begin{itemize}
    \item \textbf{The Davies generator} at invese temeperature $\beta$ is constructed so that for a single jump operator
\begin{align}\label{eq:davies-generator-transition}
    \TD^{(i,\alpha)}:=  \BS^{(f)}_{\mathrm{D}}\llbracket\vsigma^{(i,\alpha)}[\cdot]\vsigma^{(i,\alpha)}\rrbracket=\sum_{\nu\in B(\vH_n)}(f(\beta\nu)\vsigma^{(i,\alpha)}_\nu)^\dagger[\cdot](f(\beta\nu)\vsigma^{(i,\alpha)}_\nu),
\end{align}
where we use the decomposition according to the set of Bohr frequencies $B(\vH_n):=\mathrm{spec}(\vH_n)-\mathrm{spec}(\vH_n)$: \begin{equation}
    \label{eq:Bohr-decomposition}
    \begin{aligned}
        \vA_\nu &:=\sum_{E'-E=\nu}\ketbra{\psi_{E'}}{\psi_{E'}}\vA \ketbra{\psi_{E}}{\psi_{E}},\text{ corresponding to the eigendecomposition }\\
        \vH_n &= \sum_E E\ketbra{\psi_E}{\psi_E}.
    \end{aligned}
\end{equation}
Thus, the Davies Lindbladian is defined as 
\begin{align}\label{eq:LindbladTermsDef}
    \LD[\cdot]:=\sum_{\substack{i\\\alpha\in\{\rmx,\rmy,\rmz\}}}\LD^{(i,\alpha)}[\cdot]
    \quad\text{ where }\quad
    \LD^{(i,\alpha)}[\cdot]:=\TD^{(i,\alpha)}[\cdot]-\frac12\{\TD^{(i,\alpha)}[\vI],\cdot\}.
\end{align} In the special case of Metropolis weights
\begin{align}\label{eq:Metropolis}
f_M(\nu):=\left\{\begin{array}{cc}
     1 & \text{for } \nu\leq 0,  \\
     e^{-\frac{\nu}{2}} & \text{for } \nu\geq 0.
\end{array}\right.
\end{align}
we denote the corresponding Davies-Metropolis generator as $\CL_{\mathrm{DM}}$, which is studied in~\cite{basso2026SpecGapDaviesMeanFieldHeisenberg}.\footnote{Our $\CL_{\mathrm{DM}}$ is a constant multiple of the Lindbladian studied in~\cite{basso2026SpecGapDaviesMeanFieldHeisenberg}.}

\item \textbf{The ``coherent'' Lindbladian generator} \cite{gilyen2024QGeneralGlauberMetropolis,ding2024GibbsSamplingViaKMS} at inverse temperature $\beta$ is constructed similarly so that for a single jump operator
\begin{align}
    \label{eq:coherent-generator-transition}
    \TC^{(i,\alpha)}:=\BS^{(f)}_{\mathrm{C}}\llbracket\vsigma^{(i,\alpha)}[\cdot]\vsigma^{(i,\alpha)}\rrbracket=\Bigg(\sum_{\nu\in B(\vH_n)}f(\beta\nu)\vsigma^{(i,\alpha)}_\nu\Bigg)^{\!\!\dagger}[\cdot]\Bigg(\sum_{\nu\in B(\vH_n)}f(\beta\nu)\vsigma^{(i,\alpha)}_\nu\Bigg),
\end{align}
and the ``coherent'' Lindbladians $\LC$ and $\LC^{(i,\alpha)}$ are defined analogously to \autoref{eq:LindbladTermsDef}.

\item \textbf{The operator Fourier transform (OFT)} construction \cite{chen2025EfficientQThermSim,gilyen2024QGeneralGlauberMetropolis} uses transformed jumps \begin{align}\label{eq:oft-S-transform}
        \TFs^{(i,\alpha)}:=\BS^{(f)}_{\mathrm{F},\sigma}\llbracket\vsigma^{(i,\alpha)}[\cdot]\vsigma^{(i,\alpha)}\rrbracket
        &= \int_\BR \left(f\Big(\beta\omega+\frac{\sigma^2}{2}\Big)\widehat{\vsigma^{(i,\alpha)}}(\omega)\right)^{\!\dagger} [\cdot] \left(f\Big(\beta\omega+\frac{\sigma^2}{2}\Big)\widehat{\vsigma^{(i,\alpha)}}(\omega)\right) \rmd \omega, \text{ where }\\
        \widehat\vK(\omega) =& \frac{1}{\sqrt{2 \pi}} \int_\BR e^{\ri \vH_n t} \vK e^{-\ri \vH_n t} g_{\sigma}(t) e^{-\ri \omega t} \rmd t = \sum_{\nu \in B(\vH_n)} \hat{g}_{\sigma}(\omega - \nu) \vK_\nu,\,\sigma>0,\text{ and}\label{eq:OFT}\\
        g_{\sigma}(t) =
        \sqrt{\frac{\sigma}{\beta} \sqrt{\frac2\pi}}
        e^{-\frac{\sigma^2 t^2}{\beta^2}} &\overset{\mathrm{FT}}{\longleftrightarrow}
        \hat{g}_{\sigma}(\omega) =
        \sqrt{\frac{\beta}{\sigma \sqrt{2 \pi}}}
        e^{-\frac{\beta^2\omega^2}{4 \sigma^2}}=\frac{1}{\sqrt{2 \pi}} \int_\BR g_{\sigma}(t) e^{-\ri \omega t}\rmd t,\nonumber
\end{align} 
where $\frac{\sigma}{\beta}$ (not to be confused with the Pauli operators $\vsigma$) is the ``standard deviation'' of the operator Fourier transform in \autoref{eq:oft-S-transform}, expressing how accurately the label $\omega$ of $\widehat\vK(\omega)$ approximates the Bohr frequency components $\vK_\nu$ for $\nu\approx \omega$.
The OFT Lindbladians $\LFs$ and $\LFs^{(i,\alpha)}$ are defined analogously to \autoref{eq:LindbladTermsDef}.

One can relate this construction to the Davies generator by observing that 
\begin{align}\label{eq:DaviesAsLimit}
    \lim_{\sigma\rightarrow 0}\LFs^{(i,\alpha)}=\LD^{(i,\alpha)},
\end{align}
so we can think about this family of Lindbladians as a generalisation or extension of the Davies generator.
\end{itemize}

In the rest of the paper when we write $\CL$ we mean that the statement holds with respect to each of the above constructions $\LC,\LD,$ and $\LFs$. The KMS-detailed-balance property of the Lindbladians $\LC$ is proven in~\cite{gilyen2024QGeneralGlauberMetropolis,ding2024GibbsSamplingViaKMS}, while that of $\LFs$ in~\cite{chen2025EfficientQThermSim,chen2023ExactQGibbsSampler}, while that of Davies follows for example from~\eqref{eq:DaviesAsLimit}.

An important consequence of the KMS-detailed-balance property is that $\vrho$ is a fixed point of the Lindbladian $\CL^\dagger$ and that the corresponding discriminant 
\begin{align}\label{eq:DiscriminantDef}
    \CD[\cdot]=\vrho^{\frac{1}{4}}\CL[\vrho^{-\frac{1}{4}} \cdot \vrho^{-\frac{1}{4}}]\vrho^{\frac{1}{4}},
\end{align}
is self-adjoint with respect to the Hilbert-Schmidt inner product. We define the discriminants $\DC,\DD,\DFs$ and $\DC^{(i,\alpha)},\DD^{(i,\alpha)},\DFs^{(i,\alpha)}$ analogously to \eqref{eq:DiscriminantDef} with respect to the Lindbladians $\LC,\LD,\LFs$ and $\LC^{(i,\alpha)},\LD^{(i,\alpha)},\LFs^{(i,\alpha)}$.

Now we are ready to state our main result that the spectral gaps of the above families of KMS-detailed-balnced Lindbaldians admit the following common phase diagram:
\begin{restatable}[Main theorem]{theorem}{allT}\label{thm:allT}
	For fixed $\beta > 0$ the Lindbladians $\CL\in\{\LC,\LD,\LFs\}$ have the gap phase diagram
	\[
	\gap(-\CL)=\left\{
	\begin{array}{lr}
		\Theta(1),&\beta\in(0,2),\\
		\Theta(n^{-1/2}),&\beta=2,\\
		\Theta(n^{-1}),&\beta\in(2,\infty].
	\end{array}\right.
	\]
\end{restatable} The $\beta=0$ infinite-temperature point is easily incorporated into the phase diagram of the Davies--Metropolis and the coherent generators, and the gap there is also $\gamma(-\LC) = \bigTh(1)$, see \cite{basso2026SpecGapDaviesMeanFieldHeisenberg} and \Cref{eq:beta-zero-gap}. The operator Fourier transform construction, however, is badly behaved at the $\beta=0$ point ($\LFs \equiv 0$), and we hence exclude that point from the general statement here. We do note, however, that taking a limit $\beta \to 0,~\sigma \to 0$ so that $\beta / \sigma = \mathrm{const.}$ is possible, and would also yield a constant gap.

\subsection{Proof sketch}
Here we provide a proof sketch of our main result by breaking it down to a sequence of four key propositions that we prove in their own sections in the rest of the paper. Assuming these statements, we show how the proof of our main result could be assembled. 

The proof heavily relies on the observation~\cite{basso2026SpecGapDaviesMeanFieldHeisenberg} that the considered Lindbaldians inherit the symmetries of the mean-field Heisenberg model. For the symmetry $(\vg,\pi)\in \SU\times S_n$ let us define the conjugation 
\begin{equation}
    \Delta_{\vg,\pi}[\cdot]:= \vg^{\otimes n}\vU_{\pi}[\cdot](\vg^{\otimes n}\vU_{\pi})^\dagger,
\end{equation}
where $\vU_\pi$ is the unitary which permutes the sites as
\[
\vU_\pi
\ket{\psi_1}\otimes\cdots\otimes \ket{\psi_n}
=
\ket{\psi_{\pi^{-1}(1)}}\otimes\cdots\otimes \ket{\psi_{\pi^{-1}(n)}}.
\]
Then, as shown by~\cite[Theorem 3.2, Remark B.3]{basso2026SpecGapDaviesMeanFieldHeisenberg}
we have
\begin{equation}\label{eq:intertwining}
    \CL\circ \Delta_{\vg,\pi} =\Delta_{\vg,\pi}  \circ \CL
\end{equation}
for the Lindbladians we consider here $\CL\in\{\LC,\LD,\LFs\}$. In fact \cite[Theorem 3.2, Remark B.3]{basso2026SpecGapDaviesMeanFieldHeisenberg} shows this for $\LDM$, but it is not difficult to extend the theorem to all Lindbladians considered in this work, see \Cref{cor:covariance}.

Our approach is to break down the spectral analysis to various invariant subspaces of the operators $\CB(\CH_n)$ which may be defined as the \emph{commutant}, a unital C$^*$ sub-algebra, defined as
\begin{equation}\label{eq:CommDef}
    \begin{aligned}
        \comm{O} &:= \{ \vC\in \CB(\CH_n): \vC \vO=\vO \vC \quad \text{ for all } \vO\in O\} \quad \text{or}\\
        \comm{\vO} &:= \{ \vC\in \CB(\CH_n): \vC \vO=\vO \vC \},
    \end{aligned}
\end{equation}
with respect to certain (collections of) self-adjoint operators $O\subseteq \CB(\CH_n)$ related to the symmetries of $\vH_n$.

We will separately consider the $\SU$ and permutation symmetry sectors defined below. 
\begin{defn}[Collective-spin-rotation-invariant algebra]
\label{def:Rn}
The collective-spin-rotation-invariant algebra is
\[
\Rn
:=\comm{\{\Stot^x,\Stot^y,\Stot^z\}}.
\]
That is, $\Rn$ consists of the operators commuting with every collective spin rotation.
\end{defn}

An important observation is that all operators in $\Rn$ commute with $\vH_n$ as we show below:
\begin{align}\label{eq:RnInCHn}
   \Rn 
   =\comm{\{\Stot^x,\Stot^y,\Stot^z\}}
   \subseteq \comm{\{(\Stot^x)^2+(\Stot^y)^2+(\Stot^z)^2\}}
   =\comm{\{\vH_n\}}.
\end{align}

\begin{defn}[Permutation-invariant algebra]
\label{def:Pn}
The permutation-invariant algebra is
\[
\Pn
:=\comm{\{\vU_{(ij)}\colon 1\le i< j\le n\}}.
\]
That is, $\Pn$ consists of the operators commuting with every site permutation.
\end{defn}

Once again, we show below that all operators in $\Pn$ commute with $\vH_n$:
\begin{align}
   \Pn &=\comm{\{\vU_{(ij)}\colon 1\le i< j\le n\}}\nonumber\\
   &=\comm{\left\{\frac{1}{2}\left(\vI+\vsigma^{(i,\rmx)}\vsigma^{(j,\rmx)}+\vsigma^{(i,\rmy)}\vsigma^{(j,\rmy)}+\vsigma^{(i,\rmz)}\vsigma^{(j,\rmz)}\right)\colon 1\le i< j\le n\right\}}\nonumber\\
   &\subseteq \comm{\left\{\sum_{i<j}\left(\vsigma^{(i,\rmx)}\vsigma^{(j,\rmx)}+\vsigma^{(i,\rmy)}\vsigma^{(j,\rmy)}+\vsigma^{(i,\rmz)}\vsigma^{(j,\rmz)}\right)\right\}}\nonumber\\
   &=\comm{\{\vH_n\}},\label{eq:PnInCHn}
\end{align}
where we used that 
\begin{align}\label{eq:UijPauliDec}
   \vU_{(ij)}=\frac{1}{2}\left(\vI+\vsigma^{(i,\rmx)}\vsigma^{(j,\rmx)}+\vsigma^{(i,\rmy)}\vsigma^{(j,\rmy)}+\vsigma^{(i,\rmz)}\vsigma^{(j,\rmz)}\right)
\end{align} is the $i \leftrightarrow j$ swap operator.

In particular, \autoref{eq:intertwining} then implies that $\Pn$ is invariant under 
the Lindbladians $\CL$ we consider here
\begin{align}\label{eq:PInv}
    &\vX &\CL^\dagger[&\vX]&\CL^\dagger[&\vX]\nonumber\\
    \vX\in \Pn \quad \Longrightarrow \quad \forall \pi \in S_n\colon\quad\qquad &\,\,\rotatebox{90}{$=$} \qquad\Longrightarrow &&\,\,\rotatebox{90}{$=$} \qquad\underset{\eqref{eq:intertwining}}{\overset{by}{\Longrightarrow}} &&\,\,\rotatebox{90}{$=$}\qquad\Longrightarrow && \CL^\dagger[\vX]\in \Pn,\\
    \Delta_{\bm{e},\pi}[&\vX] &\CL^\dagger[\Delta_{\bm{e},\pi}[&\vX]]&\Delta_{\bm{e},\pi}[\CL^\dagger[&\vX]]\nonumber
\end{align}
analogously to~\cite[Observation~3.2, Eq.~(3.12)]{basso2026SpecGapDaviesMeanFieldHeisenberg},
and similarly $\Rn$ is also invariant  
\begin{align}\label{eq:RInv}
    &\vX &\CL^\dagger[&\vX]&\CL^\dagger[&\vX]\nonumber\\
    \vX\in \Rn \quad \Longrightarrow \quad \forall \vg \in \SU\colon\!\!\quad\qquad &\,\,\rotatebox{90}{$=$} \qquad\Longrightarrow\!\! &&\,\,\rotatebox{90}{$=$} \qquad\underset{\eqref{eq:intertwining}}{\overset{by}{\Longrightarrow}} &&\,\,\rotatebox{90}{$=$}\qquad\Longrightarrow && \CL^\dagger[\vX]\in \Rn.\\
    \Delta_{\vg,\mathrm{id}}[&\vX] &\CL^\dagger[\Delta_{\vg,\mathrm{id}}[&\vX]]&\Delta_{\vg,\mathrm{id}}[\CL^\dagger[&\vX]]\nonumber
\end{align}

\begin{corollary}\label{cor:reductions}
    $\!$Both $\Pn,\Rn$ and their Hilbert-Schmidt orthogonal complements $\Pn^\perp\!,\Rn^\perp\!$ are invariant
    under $\CL,\CD,\CL^\dagger\!\!$.
\end{corollary}
\begin{proof}
    For $\CL^\dagger\!$ and $\Pn,\Rn$ this is proven by \eqref{eq:PInv}-\eqref{eq:RInv}.
    Due to the $\SU$ and $S_n$ symmetry of $\vH_n$ we have $\vH_n\in \Rn\cap \Pn$ and consequently also $\exp(c \vH_n)\in \Rn\cap \Pn$ and $\vrho^c\in \Rn\cap \Pn$ for all $c\in\BC$, thus, e.g., $\Pn=\vrho^c\Pn\vrho^c$.
    Due to KMS detailed balance
    $\CD[\cdot]=\vrho^{\frac14}\CL[\vrho^{-\frac14}\cdot\vrho^{-\frac14}]\vrho^{\frac14}$ and $\CL^\dagger[\cdot]=\vrho^{\frac12}\CL[\vrho^{-\frac12}\cdot\vrho^{-\frac12}]\vrho^{\frac12}$, 
    implying their $\Pn,\Rn$ invariance. Self-adjointness of $\CD$ also implies the invariance of\footnote{As $\vrho^{\pm\frac12}$ is in the sub-algebras $\Pn,\Rn$ the orthogonal complements coincide for the KMS and Hilbert-Schmidt inner products.\label{footnote:kms-and-hilbert-coincide}} $\Pn^\perp,\Rn^\perp$ under $\CD$ and thus also under $\CL,\CL^\dagger$. 
\end{proof}

Let us denote by $\gamma(\CL)$ the spectral gap of $-\CL$, and by $\gamma(\CL,O)$ the spectral gap of $-\CL$ on the invariant subspace $O \in \{\Pn,\Rn,\Pn^\perp,\Rn^\perp\}$ if $\sqrt{\vrho} \in O \iff \vI \in O$, and otherwise the
smallest eigenvalue of $-\CL$ on $O$. 
By the above corollary we have both
\begin{align}
    \gamma(\CL)&=\min\{\gamma(\CL,\Pn),\gamma(\CL,\Pn^{\perp})\}, \text{ and }\label{eq:PDec}\\
    \gamma(\CL)&=\min\{\gamma(\CL,\Rn),\gamma(\CL,\Rn^{\perp})\}\label{eq:RDec}.
\end{align}
Our first key observation is that the magnitude of the Dirichlet forms 
corresponding to the Davies and ``coherent'' constructions
\begin{equation}
    \label{eq:Dirichlet-def}
    \CE(\vX) := \ip{\vX}{-\CL[\vX]}
\end{equation}
match on the subspace of operators commuting with $\vH_n$.
\begin{restatable}[Gap equivalence for commuting operators]{proposition}{commutantComparison}
\label{prop:commutant-comparison}
For every $n\ge2$ and $\beta\ge0$ we have
\[
    c_\star\EDM(\vX)
    \le \EC(\vX)
    \le \EDM(\vX) \qquad\text{ for all } \vX\in\comm{\{\vH_n\}}.
\]
Consequently $\gamma(\LC,\Rn)=\Theta(\gamma(\LDM,\Rn))$ and $\gamma(\LC,\Pn)=\Theta(\gamma(\LDM,\Pn))$.
\end{restatable}

We prove the gap bound in the low-temperature regime by studying the dynamics separately on global-spin-rotation-invariant operators ($\Rn$) and on their complement ($\Rn^\perp$).

\begin{restatable}[Low-temperature bound]{proposition}{SUtwoComplement}
\label{prop:R-complement}
For $\beta\ge0$ we have $\gamma(\LC,\Rn^{\perp}) = \Omega\left(\frac1n\right)$.
\end{restatable}

We prove the gap bound in the high-temperature regime by studying the dynamics separately on site-permutation-invariant operators and on their complement. 

\begin{restatable}[High-temperature bound]{proposition}{CperpHighT}\label{prop:P-complement}
	For all $\beta\in[0,2]$ we have $\gamma(\LC,\Pn^{\perp})=\Omega(1)$.
\end{restatable}

Finally, we close the remaining gap in the analysis of~\cite{basso2026SpecGapDaviesMeanFieldHeisenberg} for the critical temperature.
\begin{restatable}[Critical-temperature bound]{proposition}{criticalBound}
\label{prop:critical-bound}
For $\beta=2$ we have $\gamma(\LDM)=\Omega(n^{-1/2})$, and $\ipc{\Stot^z}{\Stot^z}_{\vrho}=\Theta(n^{\frac32})$.
\end{restatable}

The above four propositions enable a quick proof of our main result as follows.

\begin{proof}[Proof of our main result: \autoref{thm:allT}]
	The main result of \cite[Theorem~1.1]{basso2026SpecGapDaviesMeanFieldHeisenberg} combined with our \autoref{prop:critical-bound} directly yields
	\begin{equation}\label{eq:Davies-gaps}
		\gamma(\LDM)=
		\begin{cases}
			\Theta(1),&0\le\beta<2,\\
            \Theta(n^{-1/2}),&\beta=2,\\
			\Theta(n^{-1}),&\beta>2.
		\end{cases}
	\end{equation}

    For the other KMS-detailed-balanced Lindbladians we proceed as follows:
    
\paragraph{Gap lower bounds.}
	For $0\le\beta\le 2$ using the permutation-symmetric decomposition we get
    \begin{align}
        \gamma(\CL) &=\min\{\gamma(\CL,\Pn),\gamma(\CL,\Pn^{\perp})\}\tag{by \eqref{eq:PDec}}\\
                    &=\Omega(\min\{\gamma(\LDM,\Pn),1\})\tag{by \autoref{prop:commutant-comparison}-\ref{prop:P-complement}}\\
                    &=		
                    \begin{cases}
            			\Omega(1),&0\le\beta<2,\\
                        \Omega(n^{-1/2}),&\beta=2.
            		\end{cases}\tag{by \eqref{eq:Davies-gaps}}
    \end{align}
    Similarly, for $\beta>2$ using the rotation-symmetric decomposition we get
    \begin{align}
        \gamma(\CL) &=\min\{\gamma(\CL,\Rn),\gamma(\CL,\Rn^{\perp})\}\tag{by \eqref{eq:RDec}}\\
                    &=\Omega(\min\{\gamma(\LDM,\Rn),1/n\})\tag{by \autoref{prop:commutant-comparison}-\ref{prop:R-complement}}\\
                    &=\Omega(1/n).\tag{by \eqref{eq:Davies-gaps}}
    \end{align}
    The lower bound for the OFT version follows from our gap transfer result~\Cref{cor:gt-sectors} proved in \Cref{sec:gap-transfer}.
	
	\paragraph{Gap upper bounds.}
	Let $\vM_z=\Stot^z$.  Since $[\vM_z,\vH_n]=0$ and the magnetic eigenvalues are
	symmetric in every spin block,
	\begin{equation}\label{eq:magnetization-variance}
		\Tr(\vrho \vM_z)=0,
		\quad\text{ i.e., }\quad
        \ipc{\vI}{\vM_z}_{\vrho}=0.
	\end{equation}
    We use the characterisation of the spectral gap as the extrema of the Rayleigh quotient 
    \begin{align}\label{eq:gap-rayleigh-quotient}
    \gap(-\CD)=\gap(-\CL)=
    \inf_{\substack{\vX\ne0\\ \Tr(\vrho \vX)=0}}
    \frac{\ip{\vX}{-\CL[\vX]}}{\nr{\vX}^2}
    =\inf_{\Var_{\vrho}(\vX)>0}\frac{\CE(\vX)}{\Var_{\vrho}(\vX)},
    \end{align}
    where
    \begin{align}\label{eq:VarDef}
        \Var_{\vrho}(\vX):=\nr{\vX-\Tr(\vrho \vX)\vI}^{2}.
    \end{align}
    The spectral gap is therefore
    \begin{align}
        \gamma(\LC)&\leq\frac{\EC(\vM_z)}{\Var_{\vrho}(\vM_z)}\nonumber\\
        &\leq\frac{\EDM(\vM_z)}{\Var_{\vrho}(\vM_z)}\tag{by \autoref{prop:commutant-comparison}}\\
        &\leq\frac{4n}{\Var_{\vrho}(\vM_z)}.\tag{by \cite[Lemma~3.22]{basso2026SpecGapDaviesMeanFieldHeisenberg}}\\
        &=
        \left\{
        \renewcommand{\arraystretch}{1.2}
        \begin{array}{@{}l@{\quad}r@{}l@{\qquad}l@{}}
            \bigO{1},        & 0\le{} & \beta<2, & (\text{by \cite[Lemma~3.23]{basso2026SpecGapDaviesMeanFieldHeisenberg}})\\
            \bigO{n^{-1/2}}, &        & \beta=2, & (\text{by \autoref{prop:critical-bound}})\\
            \bigO{n^{-1}},   &        & \beta>2. & (\text{by \cite[Lemma~3.23]{basso2026SpecGapDaviesMeanFieldHeisenberg}})
        \end{array}
        \right.
    \end{align}
    The $\beta = \infty$ case of the coherent generator is covered separately in \Cref{sec:zero-temperature}, whereas the OFT construction is covered in \Cref{sec:gap-transfer}.
\end{proof}

\begin{restatable}[Gap of a Davies generator with general KMS-symmetric \texorpdfstring{$f$}{f}]{rmk}{DaviesGeneralf}\label{prop:davies-general-f}

The gap $\gamma(\LD)$ of a Davies generator with a general KMS-symmetric weight function \eqref{eq:f-KMS} is of the same order as the gap $\gamma(\LDM)$ of the Davies generator with Metropolis weights, i.e., \begin{equation*}
    \gamma(\LD) = \Theta(\gamma(\LDM)).
\end{equation*}
\end{restatable} As we will see, the discussion in \Cref{sec:davies-vs-coherent-commuting-operators} almost immediately implies this result. Hence, although above we wrote the propositions and theorem for the Metropolis--Davies generator $\LDM$, they actually hold for generic Davies generators $\LD$.

\section{Preliminaries}

\subsection{Schur--Weyl decomposition}

Let $s$ and $m=-s,-s+1,\ldots,s$ label the eigenvalues of $\Stot^2 = \sum_\alpha (\Stot^\alpha)^2$ and $\Stot^\rmz$, respectively, and $k$ an integer index labeling degeneracies. It is well known from elementary quantum theory \cite[Chapter 3]{Sakurai2020ModernQuantumMechanics} that the states $\ket{s,m,k}$ form an orthonormal basis for the space $\Hn = (\BC^2)^{\otimes n}$ of a set of $n$ spin-1/2s. Collective spin rotations generated by the $\Stot^\alpha$ may mix states with different $m$ but not with different $k$ or $s$, whereas site permutations may mix states with different $k$ but not $m$ or $s$. This can be compactly expressed by the Schur--Weyl decomposition (\cite[Theorem~2.2]{basso2026SpecGapDaviesMeanFieldHeisenberg}),
\begin{align}
    \Hn & \cong\bigoplus_{s\in\CS_n}V_s\otimes W_s,
    \qquad
    \CS_n=
    \begin{cases}
        \{0,1,\ldots,n/2\},&n\text{ even},\\
        \{1/2,3/2,\ldots,n/2\},&n\text{ odd}.
    \end{cases}\label{eq:schur-weyl}\\
    \dim V_s &= 2s + 1\nonumber\\
    m_{n,s} &:= \dim W_s =\binom n{n/2-s}-\binom n{n/2-s-1}
    =\frac{2s+1}{n/2+s+1}\binom n{n/2-s}\label{eq:multiplicity}
\end{align}
Here $V_s$ is the irreducible spin-$s$ representation of $\SU$, and $W_s$ is the
corresponding irreducible representation of $S_n$, with partition
$(n/2+s,n/2-s)$. To see the formula for $m_{n,s}$, consider that the $\Stot^\rmz=s$ eigenspace has dimension $\binom n{n/2-s}$.
Each spin-$s'$ copy contributes exactly one vector to it when $s' \ge s$, so this
dimension equals $\sum_{s' \ge s}m_{n,s'}$. Subtract the same identity at $s+1$.

On $V_s\otimes W_s$ the quadratic Casimir $\Stot^2$ is $s(s+1)\vI$, so, by letting $\vPi_s$ be the projection onto the spin-$s$ block,
\begin{equation}
    \label{eq:Heisenberg-spectrum}
    \Stot^2\vPi_s=s(s+1)\vPi_s,
    \qquad
    \vH_n\vPi_s = E_s \vPi_s = -\frac{s(s+1)}n\vPi_s,
\end{equation} yielding the spectrum $E_s$ of $\vH_n$, and that $\vH_n$ is scalar on each spin-$s$ block. 

The group actions of collective rotations commute with those of permutations, and the representations belonging to distinct $s$-blocks
are inequivalent for each group. Thus a commuting operator has no matrix elements
between different $s$ blocks, and Schur's lemma makes it scalar on the irreducible
factor on which that group acts. Consequently, we may write $\Rn$ and $\Pn$ from \Cref{def:Rn,def:Pn} as
\begin{equation*}
    \Pn
    =\bigoplus_s\CB(V_s)\otimes \vI_{W_s},\qquad
    \Rn
    =\bigoplus_s \vI_{V_s}\otimes\CB(W_s).
\end{equation*}
From the previous argument we also have
\begin{equation*}
    \Pn\cap\Rn=\operatorname{span}\{\vPi_s:s\in\CS_n\}
\end{equation*}

and that \begin{equation}
    \label{eq:commutant-H_n}
    \Pn,\Rn \subseteq \comm{\vH_n},
\end{equation} as shown earlier. Moreover, the full commutant is also block diagonal, $\comm{\vH_n} = \bigoplus_s\CB(V_s\otimes W_s)$, 
and it need not coincide with either symmetry algebra.

\subsection{Bohr decompositions}

In this section we recall the useful facts about Bohr decompositions, as defined in \autoref{eq:Bohr-decomposition}, that \begin{equation}\label{eq:bohr-identities}
	e^{\ri t\vH_n} \vA_\omega e^{-\ri t\vH_n}=e^{\ri t\omega}\vA_\omega,
	\qquad
	\vrho^a \vA_\omega \vrho^{-a}=e^{-a\beta\omega}\vA_\omega, \qquad (\vA_\omega)^\dagger = (\vA^\dagger)_{-\omega}.
\end{equation} These immediately follow from \Cref{eq:Bohr-decomposition,eq:Gibbs-state}.

We also recall a result from \cite[Appendix C]{basso2026SpecGapDaviesMeanFieldHeisenberg}: \begin{fact}
    \label{fact:pauli-selection-rule}
    For every $i,\alpha$ we have \begin{equation*}
        \vPi_{s'}\vsigma^{(i,\alpha)}\vPi_s=0 \text{ if } |s'-s|>1.
    \end{equation*} From this we also have that every Bohr frequency component of a local Pauli operator satisfies
    \begin{equation*}
        |\omega|\le1.
    \end{equation*}
\end{fact} That is, single-site Paulis can only change both the total spin quantum number $s$ and the energy $E_s$ by at most 1. The second part of the statement follows from the first because the eigenvalues of $\vH_n$ depend only on $s$, see \Cref{eq:Heisenberg-spectrum}, and the possible energy transfers are
\[
E_s-E_s=0,\qquad E_{s+1}-E_s=-\frac{2(s+1)}n,
\qquad E_{s-1}-E_s=\frac{2s}n.
\]
For an existing $s\to s+1$ transition one has $s+1\le n/2$, and for an existing
$s\to s-1$ transition one has $s\le n/2$.

\subsection{Coherent Lindblad generator}

Here we introduce further symbols which we will use to handle the ``coherent reweighting'' scheme in \Cref{eq:coherent-generator-transition} and the associated Lindbladian generator.

For a function $h$ on the finite Bohr spectrum, we have \begin{equation*}
    h(\ad_{\vH_n})(\vA)=\sum_\omega h(\omega)\vA(\omega);
\end{equation*} here $\ad_{\vX}[\vY] = [\vX,\vY]$ denotes the commutator. With $f(\nu)$ KMS (see \Cref{eq:f-KMS}) and \begin{equation}\label{eq:fsymm-KMS}
    \widetilde{f}(\nu) = e^{\nu / 4} f(\nu), \qquad \widetilde{f}(-\nu) \overset{\text{Eq.~\eqref{eq:f-KMS}}}= \widetilde{f}^*(\nu)
\end{equation} it is helpful to introduce \begin{equation}
    \label{eq:B-definition}
    \begin{gathered}
        \vL^{(i,\alpha)} = f(\beta\,\ad_{\vH_n})(\vsigma^{(i,\alpha)}), \qquad \text{and} \qquad \vB^{(i,\alpha)} = \vrho^{-1/4}\vL^{(i,\alpha)}\vrho^{1/4} =  \widetilde{f}(\beta\,\ad_{\vH_n})(\vsigma^{(i,\alpha)}); \qquad \text{also,}\\
        (\vB^{(i,\alpha)})^\dagger = \vB^{(i,\alpha)} \qquad
        \vB^{(i,\alpha)}(t) = e^{\ri t\vH_n}\vB^{(i,\alpha)} e^{-\ri t\vH_n}.
    \end{gathered}
\end{equation} These are nothing but the jump operators in the coherent Lindbladian, as discussed near \Cref{eq:coherent-generator-transition}, and its corresponding discriminant \eqref{eq:DiscriminantDef}. Explicitly, for $\vX \in \CB(\Hn)$, \begin{equation}\label{eq:coherent-lindblad}
	\begin{aligned}
	    \LC[\vX]
    	&=\sum_{i,\alpha}(\vL^{(i,\alpha)})^\dagger \vX \vL^{(i,\alpha)}
    	-\frac12\left\{(\vL^{(i,\alpha)})^\dagger \vL^{(i,\alpha)},\vX\right\},\qquad
    	\EC(\vX):=\ip{\vX}{-\LC(\vX)}.\\
        \DC[\vX] &= \sum_{i,\alpha}\vB^{(i,\alpha)} \vX \vB^{(i,\alpha)}
    	-\frac12\{(\vL^{(i,\alpha)})^\dagger \vL^{(i,\alpha)},\vX\}
	\end{aligned}
\end{equation} for the present model $\vH_n$ and all single-site Pauli jumps, we gave the Lindbladian in the Heisenberg picture, and we also made explicit the definition of the Dirichlet form $\EC(\vX)$. In passing, we note note that a possible choice for the functions $f(\nu)$ and $\widetilde{f}(\nu)$ is the Glauber weighting scheme \cite{gilyen2024QGeneralGlauberMetropolis}, $f(\nu) = (1 + e^{\nu / 2})^{-1}$ and $\widetilde{f}(\nu) = f(\nu) e^{\nu/4} = [2\cosh(\nu/4)]^{-1}$. Moreover, $\LC(\vI)=0$ trivially, and $\Tr(\vrho\LC(\vX)) = \Tr(\LC^\dagger[\vrho]\vX) = \Tr(\mathbf{0}\vX) = 0$. The first identity is equivalent to trace preservation since $\LC$ is in the Heisenberg picture, whereas for the second identity we need that the HS-adjoint $\LC^\dagger$ is the Schrödinger-picture Lindbladian, whose fixed point is $\vrho$ by its KMS symmetry.

We will make heavy use of the following Lemma, derived in Ref. \cite[Appendix~A.2, Theorem~3, equation~(A15)]{leng2026AccelQGibbsWithoutQWalks}:
\begin{lemma}\label{lem:form}
	For every observable $\vX$,
	\begin{equation}\label{eq:coh-form}
		\EC(\vX)=\sum_{i,\alpha}\int_{\mathbb R}
		\nr{[\vB^{(i,\alpha)}(t),\vX]}^2\,\rmd\mu_\beta(t),
	\end{equation}
	where $d\mu_\beta(t)=dt/[\beta\cosh(2\pi t/\beta)]$ for $\beta>0$.
\end{lemma} For the integration measure we also have \begin{equation}\label{eq:fourier}
    \int_{\mathbb R}e^{\ri t\omega}\,\rmd\mu_\beta(t) = \frac1{2\cosh(\beta\omega/4)}%
	\qquad \mu_\beta(\mathbb R)=\frac12.
\end{equation}

In the $\beta = 0$ (infinite temperature) case, we may write \begin{equation*}
    \mu_0=\frac12\delta_0
\end{equation*} as a weak limit. With this, \Cref{lem:form,eq:fourier} automatically hold, and we also have $\vL^{(i,\alpha)}=\vB^{(i,\alpha)}
= f(0) \vsigma^{(i,\alpha)}$ \begin{equation*}
	\EC(\vX)=\frac{f(0)^2}2 \sum_{i,\alpha}\nr{[\vsigma_i^\alpha,\vX]}^2.
\end{equation*} Note that $f(0) = f^*(0) \in \BR$ by \Cref{eq:f-KMS}, so no absolute value is needed. Also, because at $\beta=0$ the jump operators $\vL^{(i,\alpha)} = f(0) \vsigma^{(i,\alpha)}$ are scaled single-site Paulis, every Pauli string $\vP$ of weight $k$ is an eigenvector of $\LC$ with eigenvalue $-4f(0)^2k$; the weight $k$ is the number of sites where the string acts non-trivially. Hence, \begin{equation}\label{eq:beta-zero-gap}
	\gap(-\LC)=4f(0)^2\qquad(\beta=0).
\end{equation}

\section{Comparison of the coherent and Davies generators for operators commuting with the Hamiltonian}\label{sec:davies-vs-coherent-commuting-operators}

For observables commuting with $\vH_n$, the coherent and Davies Dirichlet forms are sums of the same squared commutator norms, with different scalar weights. We first establish this decomposition, then compare the weights.

With the all single-site Pauli operators as before, the Heisenberg-picture Davies generator with Metropolis weights is
\[
\LDM(\vX)=\sum_{i,\alpha,\omega}f^2_{\mathrm M}(\beta \omega)
\left((\vsigma^{(i,\alpha)}_\omega)^\dagger \vX \vsigma^{(i,\alpha)}_\omega
-\frac12\{(\vsigma^{(i,\alpha)}_\omega)^\dagger \vsigma^{(i,\alpha)}_\omega,\vX\}\right).
\]
Its KMS Dirichlet form (\cite[Lemma~2.6, equation~(2.14)]{basso2026SpecGapDaviesMeanFieldHeisenberg}) is \begin{equation}\label{eq:davies-bohr-form}
    \EDM(\vX):=\ip{\vX}{-\LDM(\vX)}
    =\frac12\sum_{i,\alpha,\omega}
    \abs*{\widetilde{f}_{\mathrm M}(\beta \omega)}^2
    \nr{[\vsigma^{(i,\alpha)}_\omega,\vX]}^2
    =\frac12\sum_{i,\alpha,\omega}
    e^{-\beta\abs*{\omega}/2}
    \nr{[\vsigma^{(i,\alpha)}_\omega,\vX]}^2.
\end{equation}

\begin{lemma}[Bohr decomposition on the commutant of $\vH_n$]\label{lem:bohr-diagonalization}
For every $n\ge2$, $\beta\ge0$, and $\vX \in \comm{\vH_n}$,
\begin{equation}\label{eq:coh-bohr-form}
    \EC(\vX)=\frac12\sum_{i,\alpha,\omega}
    \abs*{\widetilde{f}(\beta \omega)}^2
    \nr{[\vsigma^{(i,\alpha)}_\omega,\vX]}^2.
\end{equation}
\end{lemma}

\begin{proof}
If $[\vH_n,\vX]=0$, then
\[
[\vH_n,[\vsigma^{(i,\alpha)}_\omega,\vX]]
=[[\vH_n,\vsigma^{(i,\alpha)}_\omega],\vX]
+[\vsigma^{(i,\alpha)}_\omega,[\vH_n,\vX]]
=\omega[\vsigma^{(i,\alpha)}_\omega,\vX].
\]
Thus the commutator remains in the $\omega$-Bohr subspace. What we need next is that distinct Bohr subspaces are KMS-orthogonal: consider that if
$[\vH_n,\vY_\omega]=\omega \vY_\omega$ and
$[\vH_n,\vZ_\nu]=\nu \vZ_\nu$, then
\[
\ip{\vY_\omega}{\vZ_\nu}
=\sum_{E,E'}\sqrt{p_Ep_{E'}}\,
\Tr\bigl((\vP_E\vY_\omega \vP_{E'})^\dagger \vP_E\vZ_\nu \vP_{E'}\bigr) \propto \delta_{\omega,\nu},
\qquad p_E=Z_{n,\beta}^{-1}e^{-\beta E}.
\]
The proportionality to $\delta_{\omega,\nu}$ is because $\vP_E\vY_\omega \vP_{E'}$ is zero unless $E-E' = \omega$, and $\vP_E\vZ_\nu \vP_{E'}$ is zero unless $E-E' = \nu$, hence the product may only be non-zero if $\omega=\nu$.
Hence
\begin{align*}
	\nr{[\vB^{(i,\alpha)}(t),\vX]}^2
	&=\left\|\sum_\omega \widetilde{f}(\beta \omega) e^{it\omega}
	[\vsigma^{(i,\alpha)}_\omega,\vX]\right\|_{\vrho}^2\\
	&=\sum_\omega \abs*{\widetilde{f}(\beta \omega)}^2
	\nr{[\vsigma^{(i,\alpha)}_\omega,\vX]}^2.
\end{align*}
Plugging this back into \eqref{eq:coh-form} and using
$\mu_\beta(\mathbb R)=1/2$ from \eqref{eq:fourier}
proves \eqref{eq:coh-bohr-form}.
\end{proof}

We are now in a position to prove \Cref{prop:commutant-comparison} from the Introduction by bounding the ratio of the two weights in \Cref{eq:coh-bohr-form,eq:davies-bohr-form}
\commutantComparison*
\begin{proof}
We have $\widetilde{f}(\beta \omega) = e^{\beta \omega / 4} f(\beta \omega) = e^{-\beta \omega / 4} f^*(-\beta \omega) = \widetilde{f}^*(-\beta \omega)$ as per \Cref{eq:f-KMS,eq:fsymm-KMS}. Taking absolute values and squaring, we also have $\abs*{\widetilde{f}(\beta \omega)}^2 = e^{-\beta \abs*{\omega} / 2} \abs*{f(-\beta \abs*{\omega})}^2$, an even function; from this \begin{equation*}
    c_\star e^{-\beta \abs*{\omega} / 2} \leq \abs*{\widetilde{f}(\beta \omega)}^2 \leq e^{-\beta \abs*{\omega} / 2}
\end{equation*} because of the normalization in \Cref{eq:f-KMS}.

Comparing \eqref{eq:coh-bohr-form} with the Davies form \eqref{eq:davies-bohr-form} term by term gives
\begin{equation}\label{eq:davies-comparison}
	c_\star \EDM(\vX) \le \EC(\vX) \le \EDM(\vX)
	\qquad\text{if }\vX \in \comm{\vH_n}.
\end{equation}
In particular, this holds on $\Pn$ and $\Rn$ by \Cref{eq:RnInCHn,eq:PnInCHn}.
\end{proof}

We can almost immediately see now how bound the gap of the Davies generator with general KMS-symmetric weights. \DaviesGeneralf*
\begin{proof}
    Replacing the weights in \Cref{eq:davies-bohr-form}: \begin{equation*}
        \ED(\vX):=\ip{\vX}{-\LD(\vX)}
        =\frac12\sum_{i,\alpha,\omega}
        \abs*{\widetilde{f}(\beta \omega)}^2
        \nr{[\vsigma^{(i,\alpha)}_\omega,\vX]}^2
    \end{equation*} This is of the same form as \Cref{eq:coh-bohr-form} but now $\vX \in \CB(\Hn)$ is not restricted to $\comm{\Hn}$. Hence the term-by-term comparison in the proof of \Cref{prop:commutant-comparison} also applies here with the same constants, giving \begin{equation}
	c_\star \EDM(\vX) \le \ED(\vX) \le \EDM(\vX)
	\qquad\forall \vX \in \CB(\Hn) \qedhere.
    \end{equation}
\end{proof}

\begin{rmk}[Role of the commutation assumption]
The commutation assumption is essential for the diagonalization of the coherent Dirichlet form.  If
$\vX=\sum_\nu \vX_\nu$ is arbitrary, then
\[
[\vH_n,[\vsigma^{(i,\alpha)}_\omega,\vX_\nu]]
=(\omega+\nu)[\vsigma^{(i,\alpha)}_\omega,\vX_\nu],
\]
so the cross terms in \Cref{eq:coh-bohr-form} need not vanish. When $[\vH_n,\vX]=0$, only $\nu=0$ is present, which gives the termwise comparison in \eqref{eq:davies-comparison}.
\end{rmk}

\section{Gap bound on the orthogonal complement of the permutation-invariant sector via \texorpdfstring{$S_n$}{Sn} symmetry}
\label{sec:high-temp}
In this section we prove the high-temperature bound on the complement of the permutation-invariant sector.

\CperpHighT*

By \eqref{eq:PDec}, together with \Cref{prop:commutant-comparison}, this is the only missing input for the gap lower bound in the regime $0\leq\beta\le2$. By~\Cref{cor:reductions}, $\Pn^\perp$ is invariant under $\LC$. This Hilbert--Schmidt complement coincides with the KMS-orthogonal complement of $\Pn$ (see~\Cref{footnote:kms-and-hilbert-coincide}). Since $\LC$ is KMS-self-adjoint and $\ip{\vX}{-\LC[\vX]}=\EC(\vX)\ge0$ for all $\vX \in \CB(\Hn)$ by \Cref{lem:form}, the restriction of $-\LC$ to the invariant subspace $\Pn^\perp$ is KMS-self-adjoint and positive semidefinite.
As $\vI\in\Pn$, which is the unique zero-mode of $\LC$, we have that $-\LC$ is positive definite on the restriction to $\Pn^\perp$, and $\gamma(\LC,\Pn^\perp)$ is simply its smallest eigenvalue. Moreover, since $\vI\in\Pn$, we also have
$\Tr(\vrho\vX)=\ip{\vI}{\vX}=0~\forall \vX\in\Pn^\perp$ and hence $\Var_{\vrho}(\vX)=\nr{\vX}^2$. The variational characterization of $\gamma(\LC,\Pn^\perp)$ then reads,
\[
	\gamma(\LC,\Pn^\perp)=\min_{\substack{\vX\in\Pn^\perp \\ \vX \neq \mathbf{0}}}\frac{\EC(\vX)}{\nr{\vX}^2}.
\]
To prove \Cref{prop:P-complement} it therefore suffices to find a constant $c>0$, independent of $n$, such that
\begin{equation}\label{eq:Cperp-form-bound}
	\EC(\vX)\ \ge\ c\,\nr{\vX}^2\qquad\text{for all }\vX\in\Pn^\perp .
\end{equation}
The following theorem details exactly this bound with a KMS-filter $f$ whose symmetrized filter $\widetilde{f}$ has no zeros on $[-\beta, \beta]$; this is guaranteed by the normalization condition in \Cref{eq:f-KMS}.
\begin{theorem}[Gap bound on the orthogonal complement of the permutation-invariant sector]\label{thm:high-temp}
	Consider the inverse temperature $\beta \geq 0$. Let $f$ be a KMS filter and let $\widetilde{f}$ be its symmetrized counterpart as per \Cref{eq:f-KMS,eq:fsymm-KMS}. Define
	\begin{equation}\label{eq:band-constants}
		m_f\coloneqq\inf_{\labs{\nu}\le\beta}\labs{\widetilde f(\nu)},\qquad
		N_f\coloneqq\sup_{\labs{\nu}\le\beta}\labs{f(\nu)} .
	\end{equation}
	Then for every $n\ge2$ and every $\vX\in\Pn^\perp$,
	\begin{equation}\label{eq:high-temp-main}
		\EC(\vX)\ \ge\ \gamma_f\,\nr{\vX}^2,
		\qquad
		\gamma_f\coloneqq\frac{n}{n-1}\cdot\frac{m_f^4}{54\,N_f^2}\ \ge\ \frac{m_f^4}{54\,N_f^2}.
	\end{equation}
	Consequently, $\gamma(\LC,\Pn^\perp)\ge\gamma_f$.
\end{theorem}

In order to prove~\Cref{thm:high-temp} we will compare $\EC$ to the \emph{interchange Lindbladian} which is defined as
\begin{equation*}
    -\CL_{\mathrm{IP}}[\cdot]\coloneqq \sum_{1 \le i < j \le n}\CI[\cdot] - \Delta_{\ve,(ij)}[\cdot], \quad \Delta_{\ve,(ij)}[\cdot] = \vU_{(ij)}[\cdot]\vU_{(ij)}^\dagger,
\end{equation*}
where $\vU_{(ij)}$ is the swap operator from~\Cref{def:Pn}. By Aldous' spectral gap theorem~\cite{caputo2010AldousSpectralGapConj} we will show that $-\CL_{\mathrm{IP}}$ has a spectral gap $n$ on $\Pn^\perp$ (see~\Cref{sec:interchange-lindbladian}). The more precise statement is the following.
\begin{restatable}[Spectral gap of the interchange Lindbladian]{lemma}{lemmIP}\label{lem:IPCperp}
	The map $-\CL_{\mathrm{IP}}$ is KMS-self-adjoint with kernel $\Pn$, and for every $\vX\in\CB(\Hn)$
	\begin{equation}\label{eq:LIPNormForm}
		\ip{\vX}{-\CL_{\mathrm{IP}}[\vX]}=\frac12\sum_{i<j}\nr{\vX-\Delta_{\ve,(ij)}[\vX]}^2 .
	\end{equation}
	The KMS-orthogonal complement $\Pn^\perp$, which coincides with the Hilbert--Schmidt one (see~\Cref{footnote:kms-and-hilbert-coincide}),
	is invariant under $\CL_{\mathrm{IP}}$, and every eigenvalue of $-\CL_{\mathrm{IP}}$ on $\Pn^\perp$ is at least $n$. Equivalently,
	\begin{equation}\label{eq:IP-gap}
		\ip{\vX}{-\CL_{\mathrm{IP}}[\vX]}\ \ge\ n\,\nr{\vX}^2\qquad\text{for all }\vX\in\Pn^\perp .
	\end{equation}
\end{restatable}
Consequently we will continue the proof by showing that for some constant $c$, independent of $n$, we have $-\CD_{C} \succeq -\frac{c}n{}\CD_{\mathrm{IP}}$ and the explicit constants are given by the following theorem
\begin{restatable}[Comparison with the interchange Lindbladian]{lemma}{lemmDom}\label{lem:IPDom}
	Consider the setting of~\Cref{thm:high-temp}. Then for every $n\ge2$ and every $\vX\in\CB(\Hn)$,
	\begin{equation}\label{eq:IPDom}
		\EC(\vX)\ \ge\ \frac{m_f^4}{54\,N_f^2\,(n-1)}\,\ip{\vX}{-\CL_{\mathrm{IP}}[\vX]} .
	\end{equation}
\end{restatable}
\Cref{thm:high-temp} follows directly from \Cref{lem:IPCperp,lem:IPDom}. We prove \Cref{lem:IPCperp} in \Cref{sec:interchange-lindbladian}, and \Cref{lem:IPDom}, together with the short deduction of \Cref{thm:high-temp}, in \Cref{sec:interchange-reduction}.

\subsection{The interchange Lindbladian}
\label{sec:interchange-lindbladian}
In this subsection we will detail the proof of~\Cref{lem:IPCperp}. The key idea is to reduce the spectral analysis of $\CL_{\mathrm{IP}}$ to the spectral analysis of the corresponding interchange process as in~\cite{harrow2008QExpandersFromCayleyGraphs,bardet2021GroupTransferenceLindbladians}, or as in the following lemma, for which we provide a short proof for completeness.
\begin{lemma}\label{lem:spec-laplacian}
    Consider 
    the Cayley graph of the finite group $G$ corresponding to the set of elements $\Sigma\subseteq G$, and let $\vL_{\Sigma}$ denote the corresponding (unnormalized) Laplacian which acts on a function $f\colon G\to\BC$ as 
    \begin{align}\label{eq:CayLapDef}
        (\vL_{\Sigma}f)(g)
        :=\sum_{s\in\Sigma}w_s\bigl(f(g)-f(g\cdot s)\bigr)\quad\text{ for some positive weights }\quad w\in\BR_+^{\Sigma}.
    \end{align}
    Let $\vU(g)$ be a finite-dimensional unitary representation of $G$, then the corresponding Lindbladian
    \begin{align}\label{eq:CayLindDef}
        -\CL_{\Sigma}[\cdot]:=\sum_{s\in\Sigma}w_s\bigl(\CI[\cdot]-\vU(s)[\cdot]\vU(s^{-1})\bigr)
    \end{align}
    satisfies $\text{Spec}(-\CL_{\Sigma})\subseteq \text{Spec}(\vL_{\Sigma})$.
\end{lemma}
\begin{proof}
    Let $\lambda \in \text{Spec}(-\CL_{\Sigma})$ and let $\vY\neq0$ be a corresponding eigenoperator, i.e.,
    \begin{align}\label{eq:CayLindEV}
        -\CL_{\Sigma}[\vY]=\lambda \vY,
    \end{align}
    and associate to $\vY$ the scalar function via the Hilbert-Schmidt inner product as
    \begin{equation}\label{eq:fYDef}
        f_{\vY}(g):=\langle \vY, \vU(g)\vY \vU(g^{-1})\rangle.
    \end{equation}
    This function is nonzero, since for the unit element $e$ of the group $G$ we have
    \[
    f_{\vY}(e)=\langle \vY,\vY\rangle=\Tr[\vY^\dagger\vY]>0.
    \]
    Finally observe that
    \begin{align}
        (\vL_{\Sigma}f_{\vY})(g)
        &=\sum_{s\in\Sigma}w_s(f_{\vY}(g)-f_{\vY}(g\cdot s))\tag{by \eqref{eq:CayLapDef}}\\
        &=\sum_{s\in\Sigma}w_s(
        \langle \vY, \vU(g)\vY \vU(g^{-1})\rangle
        -\langle \vY, \vU(g\cdot s)\vY \vU(s^{-1}\cdot g^{-1})\rangle)\kern-10mm\tag{by \eqref{eq:fYDef}}\\
        &=
        \langle \vY, \vU(g)\Big(\sum_{s\in\Sigma}w_s\big(\vY -\vU(s)\vY \vU(s^{-1})\big) \Big)\vU(g^{-1})\rangle \tag{by linearity}\\
        &=
        \langle \vY, \vU(g)(-\CL_{\Sigma}[\vY])\vU(g^{-1})\rangle \tag{by \eqref{eq:CayLindDef}}\\
        &=\langle \vY, \vU(g)\lambda\vY \vU(g^{-1})\rangle \tag{by \eqref{eq:CayLindEV}}\\&
        =\lambda f_{\vY}(g).\tag*{\qedhere}
    \end{align}
\end{proof}
\begin{proof}[The proof of~\Cref{lem:IPCperp}]
   \emph{Self-adjointness and the quadratic form.} Fix $i<j$ and write $\vU=\vU_{(ij)}$ and $\Delta=\Delta_{\ve,(ij)}$.
	The swap $\vU$ is a self-adjoint unitary. It commutes with the permutation-invariant $\vH_n$, hence with $\vrho$.
	Therefore $\Delta$ is a KMS isometry:
	\[
		\ip{\Delta[\vX]}{\Delta[\vY]}=\Tr\bigl[\vU\vX^\dagger\vU\vrho^{1/2}\vU\vY\vU\vrho^{1/2}\bigr]=\ip{\vX}{\vY}.
	\]
	As $\Delta^2=\CI$, it is also KMS-self-adjoint. Hence $-\CL_{\mathrm{IP}}=\sum_{i<j}(\CI-\Delta_{\ve,(ij)})$ is a real linear combination of
	KMS-self-adjoint maps, and thus KMS-self-adjoint. Moreover,
	\begin{align*}
		\nr{\vX-\Delta[\vX]}^2
		&=\nr{\vX}^2-\ip{\vX}{\Delta[\vX]}-\ip{\Delta[\vX]}{\vX}+\nr{\Delta[\vX]}^2\\
		&=2\nr{\vX}^2-2\ip{\vX}{\Delta[\vX]} \tag{$\Delta$ is a KMS-self-adjoint isometry}\\
		&=2\ip{\vX}{(\CI-\Delta)[\vX]}.
	\end{align*}
	Summing over $i<j$ gives \eqref{eq:LIPNormForm}.
    
    \emph{Kernel and invariance.} By \eqref{eq:LIPNormForm}, $\CL_{\mathrm{IP}}[\vX]=0$ if and only if $\Delta_{\ve,(ij)}[\vX]=\vX$ for all $i<j$.
	Every permutation $\pi\in S_n$ is a product of swaps, so this holds if and only if $\CU_\pi[\vX]=\vX$ for all $\pi$,
	i.e., $\vX\in\Pn$. Since $-\CL_{\mathrm{IP}}$ is KMS-self-adjoint, it has a KMS-orthonormal eigenbasis in which the eigenoperators
	with eigenvalue $0$ span its kernel $\Pn$. The remaining eigenoperators span $\Pn^\perp$, which is therefore invariant, and none of
	them has eigenvalue $0$.
    
    \emph{Gap.} Applying~\Cref{lem:spec-laplacian} with $G=S_n$, the set of all transpositions $\Sigma$, and uniform unit weights, we can reduce our spectral gap analysis to that of the corresponding interchange-process Laplacian:
	\[
	(\vL_{\mathrm{IP}}F)(\pi)
	:=\sum_{i<j}\bigl(F(\pi)-F(\pi(ij))\bigr),
	\qquad F:S_n\to\mathbb C.
	\]
	
	Finally, Aldous' spectral-gap theorem \cite[Theorem~1.1]{caputo2010AldousSpectralGapConj} supplies the required spectral estimate $\gamma\geq n$ by equating the spectral gap of $\vL_{\mathrm{IP}}$ with that of the following (unnormalized) Laplacian of the complete graph $K_n$
	\[
	(\vL_{K_n}f)(i)
	=\sum_{j\neq i}\bigl(f(i)-f(j)\bigr)
	=nf(i)-\sum_{j=1}^n f(j).
	\]
	Let $\vJ$ denote the $n\times n$ all one matrix; using we can conclude
	\begin{align*}
		\gamma=\operatorname{gap}(-\CL_{\mathrm{IP}})
		\ge\operatorname{gap}(\vL_{\mathrm{IP}})
		=\operatorname{gap}(\vL_{K_n})
		=\operatorname{gap}(n\vI-\vJ)
		=n.\tag*{\qedhere}
	\end{align*}
\end{proof}

\subsection{Reducing to the interchange Lindbladian}\label{sec:interchange-reduction}
In this subsection we prove \Cref{lem:IPDom}. Throughout the remaining subsections, $f$ is a general KMS filter in the sense of~\eqref{eq:f-KMS}, with symmetrized filter $\widetilde f(\nu)=e^{\nu/4}f(\nu)$, so that $\widetilde f(-\nu)=\widetilde{f}^*(\nu)$. We write the jumps as
\[
    \vL^{(i,\alpha)}\coloneqq f(\beta\,\ad_{\vH_n})(\vsigma^{(i,\alpha)}),
\]
and their corresponding symmetrized versions
\[
    \vB^{(i,\alpha)}\coloneqq\widetilde f(\beta\,\ad_{\vH_n})(\vsigma^{(i,\alpha)})=\vrho^{-1/4}\vL^{(i,\alpha)}\vrho^{1/4}.
\]

Let $\vPi^\pm_{ij}\coloneqq\frac12(\vI\pm\vU_{(ij)})$ be the projections onto the symmetric and antisymmetric sectors of the swap operator $\vU_{(ij)}$. Since $\Delta_{\ve,(ij)}$ is an involution, every $\vX$ splits uniquely into a swap-even and a swap-odd part, $\vX=\vX^+_{ij}+\vX^-_{ij}$. The odd part decomposes as
\begin{equation}\label{eq:odd-part}
	\vX^-_{ij}\coloneqq\tfrac12\bigl(\vX-\Delta_{\ve,(ij)}[\vX]\bigr)
	=\vPi^+_{ij}\vX\vPi^-_{ij}+\vPi^-_{ij}\vX\vPi^+_{ij},
\end{equation}
and since $\vX-\Delta_{\ve,(ij)}[\vX]=2\vX^-_{ij}$, the identity~\eqref{eq:LIPNormForm} of \Cref{lem:IPCperp} reads
\begin{equation}\label{eq:LIP-odd}
    \ip{\vX}{-\CL_{\mathrm{IP}}[\vX]}=2\sum_{i<j}\nr{\vX^-_{ij}}^2 .
\end{equation}

Recall the sum-of-squares representation from~\Cref{lem:form}
\begin{equation}\label{eq:local-energy}
	\EC(\vX)=\sum_{i=1}^n\CE_i(\vX),
	\qquad
	\CE_i(\vX)\coloneqq\sum_\alpha\int_{\BR}\nr{[\vB^{(i,\alpha)}(t),\vX]}^2\,\rd\mu_\beta(t)\ \geq0 .
\end{equation}
Here $\vB^{(i,\alpha)}(t)=\e^{\ri t\vH_n}\vB^{(i,\alpha)} e^{-\ri t\vH_n}$, and $\mu_\beta$ is the universal measure of \Cref{lem:form},
with $\mu_\beta(\BR)=\frac12$.

For each pair we want to prove $\nr{\vX^-_{ij}}^2\lesssim\CE_i(\vX)+\CE_j(\vX)$, and we do so through an auxiliary operator $\vK_{ij}$ which we will define later
\begin{equation}\label{eq:K-strategy}
	\delta^2\,\nr{\vX^-_{ij}}^2\ \le\ \nr{[\vK_{ij},\vX]}^2\ \le\ c\,\bigl(\CE_i(\vX)+\CE_j(\vX)\bigr).
\end{equation}
The first lower bound holds as soon as $\vK_{ij}$ commutes with $\vU_{(ij)}$ and $\vrho$ and has separated spectra on the two swap sectors, for the following reason.

With respect to $\vI=\vPi^+_{ij}+\vPi^-_{ij}$, an operator commuting with $\vU_{(ij)}$ is block diagonal. Writing
\[
	\vK_{ij}=\begin{pmatrix}\vK_+&0\\0&\vK_-\end{pmatrix},\qquad
	\vX=\begin{pmatrix}\vX_{++}&\vX_{+-}\\\vX_{-+}&\vX_{--}\end{pmatrix},\qquad
	\vX^-_{ij}=\begin{pmatrix}0&\vX_{+-}\\\vX_{-+}&0\end{pmatrix},
\]
we obtain
\[
	[\vK_{ij},\vX]=\begin{pmatrix}[\vK_+,\vX_{++}]&\vK_+\vX_{+-}-\vX_{+-}\vK_-\\\vK_-\vX_{-+}-\vX_{-+}\vK_+&[\vK_-,\vX_{--}]\end{pmatrix}.
\]
Since $\vPi^\pm_{ij}$ commute with $\vrho$, the four blocks of any operator are mutually orthogonal in the KMS inner product. So the diagonal blocks of $[\vK_{ij},\vX]$ can only increase its norm, and $\nr{[\vK_{ij},\vX]}\ge\nr{[\vK_{ij},\vX^-_{ij}]}$. In a common eigenbasis of $\vK_{ij}$, $\vU_{(ij)}$ and $\vrho$, the map $\vX_{+-}\mapsto\vK_+\vX_{+-}-\vX_{+-}\vK_-$ multiplies each matrix element by a difference $\lambda_+-\lambda_-$ of an eigenvalue on the symmetric sector and one on the antisymmetric sector. If all eigenvalues of $\vK_+$ are at least $a$ and all eigenvalues of $\vK_-$ at most  $b<a$, these differences are at least $a-b$, which gives the lower bound in \eqref{eq:K-strategy} with $\delta=a-b$.

The task is therefore to filter an operator with separated sector spectra in such a way that the separation survives with a meaningful $\delta$. In our case the right unfiltered operator is the squared total spin of the pair. With
\[
	\vA^\alpha_{ij,+}\coloneqq\frac{\vsigma^{(i,\alpha)}+\vsigma^{(j,\alpha)}}{\sqrt2},
\]
the identities $(\vsigma^{(i,\alpha)})^2=(\vsigma^{(j,\alpha)})^2=\vI$, $[\vsigma^{(i,\alpha)},\vsigma^{(j,\alpha)}]=0$ and
$\sum_\alpha\vsigma^{(i,\alpha)}\vsigma^{(j,\alpha)}=2\vU_{(ij)}-\vI$ (by the definition of the swap in \Cref{def:Pn}) give
\begin{equation}\label{eq:pair-casimir}
	\sum_\alpha\bigl(\vA^\alpha_{ij,+}\bigr)^2=2\vI+2\vU_{(ij)}=4\,\vPi^+_{ij}.
\end{equation}
Its eigenvalues are $4$ on the symmetric sector and $0$ on the antisymmetric one, which is the singlet of the pair.
Multiplying \eqref{eq:pair-casimir} by $\vPi^-_{ij}$ from both sides and using $\vPi^-_{ij}\vPi^+_{ij}=0$ gives $\sum_\alpha(\vA^\alpha_{ij,+}\vPi^-_{ij})^\dagger(\vA^\alpha_{ij,+}\vPi^-_{ij})=0$. A vanishing sum of positive semidefinite operators has vanishing summands, so
\begin{equation}\label{eq:singlet-annihilation}
	\vA^\alpha_{ij,+}\,\vPi^-_{ij}=0\qquad\text{for every }\alpha .
\end{equation}
Because $\vPi^-_{ij}$ commutes with $\vH_n$, this identity survives filtering; see \Cref{lem:swap-blocks}.
The filter can therefore only distort the eigenvalue $4$ on the symmetric sector, and as we show in~\Cref{lem:swap-blocks}, it never drops below $4m_f^2$.
We therefore set
\begin{equation}\label{eq:K-definition}
	\vB^\alpha_{ij,+}\coloneqq\widetilde f(\beta\,\ad_{\vH_n})\bigl(\vA^\alpha_{ij,+}\bigr)=\frac{\vB^{(i,\alpha)}+\vB^{(j,\alpha)}}{\sqrt2},
	\qquad
	\vK_{ij}\coloneqq\sum_\alpha\bigl(\vB^\alpha_{ij,+}\bigr)^2 .
\end{equation}
Both operators are self-adjoint. Indeed, since the Bohr spectrum is symmetric and using $\widetilde f^*(\nu)=\widetilde f(-\nu)$ and $\vsigma^{(i,\alpha)}(\omega)^\dagger=\vsigma^{(i,\alpha)}(-\omega)$, we get
\[
	(\vB^{(i,\alpha)})^\dagger
	=\sum_{\omega\in B(\vH_n)}\widetilde f^*(\beta\omega)\,\vsigma^{(i,\alpha)}(\omega)^\dagger
	=\sum_{\omega\in B(\vH_n)}\widetilde f(-\beta\omega)\,\vsigma^{(i,\alpha)}(-\omega)
	=\sum_{\omega'\in B(\vH_n)}\widetilde f(\beta\omega')\,\vsigma^{(i,\alpha)}(\omega')
	=\vB^{(i,\alpha)} ,
\]
Hence $\vB^\alpha_{ij,+}$ is self-adjoint, and $\vK_{ij}=\sum_\alpha(\vB^\alpha_{ij,+})^\dagger\vB^\alpha_{ij,+}$ is self-adjoint and positive semidefinite.

The two properties we will utilize in the proofs are:
\begin{restatable}[Lower bound by the swap-odd part]{lemma}{lemmKOdd}\label{lem:K-detects-odd}
	Consider the setting of~\Cref{thm:high-temp}, and the operator $\vK_{ij}\in\CB(\Hn)$ as define in~\eqref{eq:K-definition}. For all sites $i\neq j$ and every $\vX\in\CB(\Hn)$,
	\begin{equation}\label{eq:K-detects-odd}
		\nr{[\vK_{ij},\vX]}^2\ \ge\ 16\,m_f^4\,\nr{\vX^-_{ij}}^2 .
	\end{equation}
\end{restatable}

\begin{restatable}[Upper bound via the single-site Dirichlet forms]{lemma}{lemmKEnergy}\label{lem:K-energy}
	Consider the setting of~\Cref{thm:high-temp}, and the operator $\vK_{ij}\in\CB(\Hn)$ as defined in~\eqref{eq:K-definition}. For all sites $i\neq j$ and every $\vX\in\CB(\Hn)$
	\begin{equation}\label{eq:K-energy}
		\nr{[\vK_{ij},\vX]}^2\ \le\ 432\,N_f^2\,\bigl(\CE_i(\vX)+\CE_j(\vX)\bigr).
	\end{equation}
\end{restatable}

\Cref{lem:K-detects-odd} is proved in \Cref{sec:K-odd} and \Cref{lem:K-energy} in \Cref{sec:K-energy}. Taking them for granted, both the comparison and the theorem follow quickly.
\lemmDom*
\begin{proof}
	Each site occurs in exactly $n-1$ pairs. Hence,
	\begin{align*}
		\EC(\vX) &=\frac1{n-1}\sum_{i<j}\left(\CE_i(\vX)+\CE_j(\vX)\right) \geq \frac1{432\,N_f^2(n-1)}\sum_{i<j}\nr{[\vK_{ij},\vX]}^2 \tag{~\Cref{lem:K-energy}}\\
		&\geq \frac{m_f^4}{27\,N_f^2(n-1)}\sum_{i<j}\nr{\vX^-_{ij}}^2 \tag{\Cref{lem:K-detects-odd}}\\
		&= \frac{m_f^4}{54\,N_f^2(n-1)}\,\ip{\vX}{-\CL_{\mathrm{IP}}[\vX]} . \tag{\Cref{eq:LIP-odd}~\qedhere}
	\end{align*}
\end{proof}

\begin{proof}[Proof of \Cref{thm:high-temp}]
	For $\vX\in\Pn^\perp$, combining \Cref{lem:IPDom} with \eqref{eq:IP-gap} gives
	\[
		\EC(\vX)\ \ge\ \frac{m_f^4}{54\,N_f^2(n-1)}\,\ip{\vX}{-\CL_{\mathrm{IP}}[\vX]}
		\ \ge\ \frac{n\,m_f^4}{54\,N_f^2(n-1)}\,\nr{\vX}^2=\gamma_f\nr{\vX}^2 .
	\]
	The bound $\gamma(\LC,\Pn^\perp)\ge\gamma_f$ then follows from the variational characterization of $\gamma(\LC,\Pn^\perp)$
	given at the beginning of this section.
\end{proof}

\subsection{Symmetries of \texorpdfstring{$\vK_{ij}$}{Kij}}\label{sec:K-symmetry}

The unfiltered pair in~\eqref{eq:pair-casimir} is rotation invariant, swap symmetric and commutes with $\vH_n$. The filter $\widetilde f(\beta\,\ad_{\vH_n})$ preserves every symmetry that commutes with $\vH_n$, so $\vK_{ij}$ inherits all three properties. This makes $\vK_{ij}$ block diagonal with respect to the spin blocks and the two swap sectors, which is what the spectral argument of~\Cref{sec:K-odd} needs.

\begin{lemma}\label{lem:commutation-with-bohr}
	Let $\vX,\vY\in\CB(\Hn)$ satisfy $[\vX,\vH_n]=0=[\vY,\vH_n]$. Then for every function $h\colon B(\vH_n)\to\BC$ and every $\vA\in\CB(\Hn)$,
	\[
		\vX\,h(\ad_{\vH_n})(\vA)\,\vY=h(\ad_{\vH_n})(\vX\vA\vY).
	\]
\end{lemma}
\begin{proof}
	\begin{align*}
        \vX\vA_\omega\vY &= \vX\left(\sum_{E'-E=\omega}\vP_E'\vA\vP_E\right)\vY\\
        &=\sum_{E'-E=\omega}\vP_{E'}(\vX\vA\vY)\vP_E\tag{$[\vX,\vH] = 0\implies[\vX,\vP_E] = 0$}\\
        &= (\vX\vA\vY)(\omega)
    \end{align*}
    Multiplying by $h(\omega)$ and summing over $\omega\in B(\vH_n)$ gives the claim.
\end{proof}

\begin{lemma}[Symmetries of $\vK_{ij}$]\label{lem:K-scalar}
	Let $\epsilon$ be the Levi-Civita symbol. For all sites $i\neq j$ and all $\alpha,\mu\in\{\rmx,\rmy,\rmz\}$:
	\begin{enumerate}[label=(\alph*), ref=\alph*]
		\item\label{item:StotBComm} $[\Stot^\mu,\vB^{(i,\alpha)}]=\ri\sum_{\gamma}\epsilon_{\mu\alpha\gamma}\,\vB^{(i,\gamma)}$, and hence
		$[\Stot^\mu,\vB^\alpha_{ij,+}]=\ri\sum_{\gamma}\epsilon_{\mu\alpha\gamma}\,\vB^\gamma_{ij,+}$;
		\item\label{item:KStotComm} $[\vK_{ij},\Stot^\mu]=0$;
		\item\label{item:BSwapConj} $\vU_{(ij)}\vB^{(i,\alpha)}\vU_{(ij)}=\vB^{(j,\alpha)}$. Hence $\vU_{(ij)}\vB^\alpha_{ij,+}\vU_{(ij)}=\vB^\alpha_{ij,+}$ and $[\vK_{ij},\vU_{(ij)}]=0$.
	\end{enumerate}
\end{lemma}
\begin{proof}
	\ref{item:StotBComm}: $\Stot^\mu=\frac12\sum_k\vsigma_k^\mu$ commutes with $\vH_n$. Applying \Cref{lem:commutation-with-bohr}
	once with $(\vX,\vY)=(\Stot^\mu,\vI)$ and once with $(\vX,\vY)=(\vI,\Stot^\mu)$ gives
	\begin{align*}	    
		[\Stot^\mu,\vB^{(i,\alpha)}] &= [\Stot^\mu,\widetilde f(\beta\,\ad_{\vH_n})(\vsigma^{(i,\alpha)})] = \Stot^\mu\widetilde f(\beta\,\ad_{\vH_n})(\vsigma^{(i,\alpha)}) - \widetilde f(\beta\,\ad_{\vH_n})(\vsigma^{(i,\alpha)})\Stot^\mu\\
		&=\widetilde f(\beta\,\ad_{\vH_n})(\Stot^\mu\vsigma^{(i,\alpha)}) - \widetilde f(\beta\,\ad_{\vH_n})(\vsigma^{(i,\alpha)}\Stot^\mu) = \widetilde f(\beta\,\ad_{\vH_n})\left([\Stot^\mu,\vsigma^{(i,\alpha)}]\right) \tag{\Cref{lem:commutation-with-bohr}}\\
		&=\widetilde f(\beta\,\ad_{\vH_n})\left(\tfrac12[\vsigma_i^\mu,\vsigma^{(i,\alpha)}]\right)=\ri\sum_\gamma\epsilon_{\mu\alpha\gamma}\,\vB_i^\gamma .
	\end{align*}
	The statement for $\vB^\alpha_{ij,+}$ follows by linearity.

	\ref{item:KStotComm}: By the Leibniz rule and \ref{item:StotBComm},
	\[
		\Bigl[\Stot^\mu,\sum_\alpha\vB^\alpha_{ij,+}\vB^\alpha_{ij,+}\Bigr]
		=\ri\sum_{\alpha,\gamma}\epsilon_{\mu\alpha\gamma}\bigl(\vB^\gamma_{ij,+}\vB^\alpha_{ij,+}+\vB^\alpha_{ij,+}\vB^\gamma_{ij,+}\bigr)=0,
	\]
	because the bracket is symmetric under $\alpha\leftrightarrow\gamma$, while $\epsilon_{\mu\alpha\gamma}$ is antisymmetric.

	\ref{item:BSwapConj}: We have $\vU_{(ij)}\vsigma^{(i,\alpha)}\vU_{(ij)}=\vsigma^{(j,\alpha)}$ and $[\vU_{(ij)},\vH_n]=0$.
	Then, \Cref{lem:commutation-with-bohr} gives the first claim as
    \begin{align*}
		\vU_{(ij)}\vB^{(i,\alpha)}\vU_{(ij)}
		&=\vU_{(ij)}\,\widetilde f(\beta\,\ad_{\vH_n})(\vsigma^{(i,\alpha)})\,\vU_{(ij)}
		=\widetilde f(\beta\,\ad_{\vH_n})\bigl(\vU_{(ij)}\vsigma^{(i,\alpha)}\vU_{(ij)}\bigr)\tag{\Cref{lem:commutation-with-bohr}}\\
		&=\widetilde f(\beta\,\ad_{\vH_n})(\vsigma^{(j,\alpha)})=\vB^{(j,\alpha)}.
	\end{align*}
    Similarly $\vU_{(ij)}\vB^{(j,\alpha)}\vU_{(ij)}=\vB^{(i,\alpha)}$. Consequently,
	\[
		\vU_{(ij)}\vB^\alpha_{ij,+}\vU_{(ij)}
		=\frac{\vU_{(ij)}\vB^{(i,\alpha)}\vU_{(ij)}+\vU_{(ij)}\vB^{(j,\alpha)}\vU_{(ij)}}{\sqrt2}
		=\frac{\vB^{(j,\alpha)}+\vB^{(i,\alpha)}}{\sqrt2}
		=\vB^\alpha_{ij,+}.
	\]
	Finally,
	\[
		\vU_{(ij)}\vK_{ij}\vU_{(ij)}
		=\sum_\alpha\bigl(\vU_{(ij)}\vB^\alpha_{ij,+}\vU_{(ij)}\bigr)\bigl(\vU_{(ij)}\vB^\alpha_{ij,+}\vU_{(ij)}\bigr)
		=\sum_\alpha\bigl(\vB^\alpha_{ij,+}\bigr)^2=\vK_{ij},
	\]
    which also immediately implies $[\vK_{ij}, \vU_{(ij)}] = 0$ as
    \begin{equation*}
        \underbrace{\vU_{(ij)}\vU_{(ij)}}_{\vI}\vK_{ij}\vU_{(ij)} = \vU_{(ij)}\vK_{ij} \implies[\vK_{ij}, \vU_{(ij)}] = 0 \tag*{\qedhere}.
    \end{equation*}
\end{proof}

\begin{corollary}[Block structure of $\vK_{ij}$]\label{cor:K-blocks}
	For all sites $i\neq j$:
	\begin{enumerate}[label=(\alph*), ref=\alph*]
		\item\label{item:K-comm-relations} $[\vK_{ij},\vH_n]=0$. Consequently $[\vK_{ij},\vPi_s]=0$ for all $s$, and $[\vK_{ij},\vrho]=0$.
		\item\label{item:K-proj} The projections $\vP_{s,\pm}\coloneqq\vPi_s\vPi^\pm_{ij}$ are mutually orthogonal, sum to $\vI$, and
		\begin{equation}\label{eq:K-block-diag}
			\vK_{ij}=\sum_{s}\sum_{\ell=\pm}\vP_{s,\ell}\,\vK_{ij}\,\vP_{s,\ell}.
		\end{equation}
		\item\label{item:K-basis} There is an orthonormal basis $\CK$ of $\Hn$ such that each $v\in\CK$ lies in a single block $\vP_{s,\pm}\Hn$ and is an eigenvector of $\vK_{ij}$. In particular, every $v\in\CK$ is a common eigenvector of $\vK_{ij}$,
		$\vU_{(ij)}$, $\vH_n$ and $\vrho$.
	\end{enumerate}
\end{corollary}
\begin{proof}
	\ref{item:K-comm-relations}: By \Cref{lem:K-scalar}\,\ref{item:KStotComm} and the Leibniz rule, $[\vK_{ij},\Stot^2]=\sum_\mu\left([\vK_{ij},\Stot^\mu]\Stot^\mu+\Stot^\mu[\vK_{ij},\Stot^\mu]\right)=0$.
	Since $\vH_n=-\Stot^2/n$, and both $\vPi_s$ and $\vrho$ are functions of $\vH_n$, the claim follows.

	\ref{item:K-proj}: $\vU_{(ij)}$ commutes with $\vH_n$ and hence with each $\vPi_s$. So the $\vP_{s,\pm}$ are products of commuting projections. They are mutually orthogonal and sum to $\vI$ because the projectors $\vPi_s$ and $\vPi^\pm_{ij}$ each do.
	By \ref{item:K-comm-relations} and \Cref{lem:K-scalar}\,\ref{item:BSwapConj}, $\vK_{ij}$ commutes with every $\vP_{s,\ell}$.
	Therefore $\vP_{s,\ell}\vK_{ij}\vP_{s',\ell'}=\vK_{ij}\vP_{s,\ell}\vP_{s',\ell'}=0$ unless $(s,\ell)=(s',\ell')$.
	Inserting $\vI=\sum_{s,\ell}\vP_{s,\ell}$ on both sides of $\vK_{ij}$ gives \eqref{eq:K-block-diag}.

    \ref{item:K-basis}: The operators $\vK_{ij}$ (self-adjoint by the discussion after \eqref{eq:K-definition}), $\vU_{(ij)}$ and $\vH_n$
	are self-adjoint and commute pairwise by \ref{item:K-comm-relations} and \Cref{lem:K-scalar}\,\ref{item:BSwapConj}. Hence they have a
	common orthonormal eigenbasis $\CK$. Let $v\in\CK$. Since $\vU_{(ij)}v=\varepsilon v$ with $\varepsilon\in\{\pm1\}$, we have
	$\vPi^\varepsilon_{ij}v=v$. Since $\vH_nv=Ev$ and the energies $E_s$ are pairwise distinct, $E=E_s$ for exactly one $s$, and
	$\vPi_sv=v$ because $\vPi_s$ is the spectral projection of $\vH_n$ for $E_s$. Hence $v=\vPi_s\vPi^\varepsilon_{ij}v\in\vP_{s,\varepsilon}\Hn$.
\end{proof}

\subsection{Lower bound: \texorpdfstring{$\vK_{ij}$}{Kij} separates the swap sectors}\label{sec:K-odd}

Here we make the eigenvalue argument behind \eqref{eq:K-definition} precise. The zero eigenvalue on the antisymmetric sector is protected by \eqref{eq:singlet-annihilation} as $\vPi^-_{ij}$ commutes with $\vH_n$. For the symmetric sector we use that, on a vector of fixed energy, the filter
multiplies each Bohr component of a local operator by $\widetilde f(\beta\omega)$; this is quantified in \Cref{lem:fixed-energy}.

\begin{lemma}[Fixed-energy comparison]\label{lem:fixed-energy}
	Consider the inverse temperature $\beta \geq 0$ and an arbitrary symmetrized KMS-filter $\widetilde{f}$. Let $\vA\in\CB(\Hn)$ satisfy $B(\vA)\subseteq[-1,1]$, and let $\vB\coloneqq\widetilde f(\beta\,\ad_{\vH_n})(\vA)$. For every $s\in\CS_n$ and every $\psi\in\vPi_s\Hn$,
	\begin{equation}\label{eq:fixed-energy}
		\|\vB\psi\|^2=\sum_{\omega\in B(\vA)}\labs{\widetilde f(\beta\omega)}^2\,\|\vA_\omega\psi\|^2.
	\end{equation}
	In particular,
    \begin{equation*}
        m_f\coloneqq\inf_{\labs{\omega}\leq1}\labs{\widetilde f(\beta\omega)}, \qquad \nrm{\vB\psi}^2\geq m_f^2\,\nrm{\vA\psi}^2
    \end{equation*}
\end{lemma}
\begin{proof}
	Since the energies $E_t=-t(t+1)/n$ are pairwise distinct, $\vPi_s$ is the spectral projection of $\vH_n$ for the
	eigenvalue $E_s$, and $\vPi_r\psi=\delta_{rs}\psi$. In the Bohr decomposition only $r=s$ survives:
	\[
		\vA_\omega\psi=\sum_{E_t-E_r=\omega}\vPi_t\vA\vPi_r\psi=\sum_{t\colon E_t=E_s+\omega}\vPi_t\vA\psi .
	\]
	Again by distinctness of the energies, at most one $t$ satisfies $E_t=E_s+\omega$. Write $t(\omega)$ for it, so that
	\[
		\vA_\omega\psi=\vPi_{t(\omega)}\vA\psi,
	\]
	with the convention $\vPi_{t(\omega)}\coloneqq0$ if no such $t$ exists. Distinct frequencies $\omega\neq\omega'$ give
	distinct blocks $t(\omega)\neq t(\omega')$, and therefore
	\[
		\ipc{\vA_\omega\psi}{\vA(\omega')\psi}=\ipc{\vA\psi}{\vPi_{t(\omega)}\vPi_{t(\omega')}\vA\psi}=0
		\qquad(\omega\neq\omega').
	\]
	Hence by Pythagoras
	\begin{align*}
		\nrm{\vB\psi}^2
		&=\nrm{\sum_{\omega\in B(\vA)}\widetilde f(\beta\omega)\,\vA_\omega\psi}^2
		=\sum_{\omega,\omega'\in B(\vA)}\widetilde f^*(\beta\omega)\widetilde f(\beta\omega')\,\ipc{\vA_\omega\psi}{\vA(\omega')\psi}\\
		&=\sum_{\omega\in B(\vA)}\labs{\widetilde f(\beta\omega)}^2\,\nrm{\vA_\omega\psi}^2,
	\end{align*}
	which is \eqref{eq:fixed-energy}. Similarly,
	\[
		\nrm{\vA\psi}^2=\nrm{\sum_{\omega\in B(\vA)}\vA_\omega\psi}^2=\sum_{\omega\in B(\vA)}\nrm{\vA_\omega\psi}^2 .
	\]
	Finally, by the assumptions
	\[
		\nrm{\vB\psi}^2\ \ge\ m_f^2\sum_{\omega\in B(\vA)}\|\vA_\omega\psi\|^2=m_f^2\,\|\vA\psi\|^2 .\qedhere
	\]
\end{proof}
We now apply~\Cref{lem:fixed-energy} to the unfiltered operators $\vA^\alpha_{ij,+}$. Since
\begin{equation*}
	\sum_{\alpha\in\{\rmx,\rmy,\rmz\}}\left(\vA^\alpha_{ij,+}\right)^2 = 4\,\vPi^+_{ij}.
\end{equation*}
the lower bound on the filtered operators $\vB^\alpha_{ij,+}$ translates into a lower bound on the quadratic form of $\vK_{ij}$.
\begin{lemma}[Spectrum of $\vK_{ij}$ on the two sectors]\label{lem:swap-blocks}
	Assume the setting of~\Cref{lem:fixed-energy} and let $\vK_{ij}$ be as in~\eqref{eq:K-definition}. For all sites $i\neq j$, we have $\vK_{ij}\vPi^-_{ij}=0$. Moreover, for every $s$ and every $\psi\in\vPi_s\Hn$ with $\vU_{(ij)}\psi=\psi$,
	\[
		\ipc{\psi}{\vK_{ij}\psi} \geq 4m_f^2\,\nrm{\psi}^2 .
	\]
\end{lemma}
\begin{proof}
	\emph{Antisymmetric sector.} Since $\vPi^-_{ij}$ commutes with $\vH_n$, \Cref{lem:commutation-with-bohr} and~\eqref{eq:singlet-annihilation} give
	\begin{equation*}
	    \vB^\alpha_{ij,+}\vPi^-_{ij}=\widetilde f(\beta\,\ad_{\vH_n})\left(\vA^\alpha_{ij,+}\vPi^-_{ij}\right)=0.
	\end{equation*}
	Hence $\vK_{ij}\vPi^-_{ij}=\sum_\alpha\vB^\alpha_{ij,+}\left(\vB^\alpha_{ij,+}\vPi^-_{ij}\right)=0$.

	\emph{Symmetric sector.} By \Cref{fact:pauli-selection-rule}, $B(\vA^\alpha_{ij,+})\subseteq B(\vsigma^{(i,\alpha)})\cup B(\vsigma^{(j,\alpha)})\subseteq[-1,1]$, so \Cref{lem:fixed-energy} applies to $\vA^\alpha_{ij,+}$. Using that $\vB^\alpha_{ij,+}$ and $\vA^\alpha_{ij,+}$ are self-adjoint and $\vPi^+_{ij}\psi=\psi$,
	\begin{align*}
		\ipc{\psi}{\vK_{ij}\psi} &=\sum_\alpha\ipc{\vB^\alpha_{ij,+}\psi}{\vB^\alpha_{ij,+}\psi} = \sum_\alpha\nrm{\vB^\alpha_{ij,+}\psi}^2\\
		&\geq m_f^2\sum_\alpha\nrm{\vA^\alpha_{ij,+}\psi}^2 \tag{\Cref{lem:fixed-energy}}\\
		&=m_f^2\,\ipc{\psi}{\sum_\alpha\left(\vA^\alpha_{ij,+}\right)^2\psi} =4m_f^2\,\ipc{\psi}{\vPi^+_{ij}\psi} \tag{By~\eqref{eq:pair-casimir}}\\
		&=4m_f^2\nrm{\psi}^2 .\tag*{\qedhere}
	\end{align*}
\end{proof}
The previous lemma concerns vectors with a fixed spin and swap parity. We now translate those bounds directly into bounds on the eigenvalues of $\vK_{ij}$.
\begin{corollary}[Separation of the swap sectors]\label{cor:K-separation}
	Let $\CK$ be the basis of \Cref{cor:K-blocks}\,\ref{item:K-basis}, write $\vK_{ij}v=\lambda(v)v$, and let
	$\CK_\pm\coloneqq\{v\in\CK:\vU_{(ij)}v=\pm v\}$. Then
	\begin{equation}\label{eq:K-eigenval-sector-bound}
		\lambda(v)\ge4m_f^2\quad(v\in\CK_+),
		\qquad
		\lambda(v)=0\quad(v\in\CK_-).
	\end{equation}
\end{corollary}
\begin{proof}
	By \Cref{cor:K-blocks}\,\ref{item:K-basis}, each $v\in\CK$ lies in a single subspace $\vP_{s,\pm}\Hn=\vPi_s\vPi^\pm_{ij}\Hn$. For $v\in\CK_+$, \Cref{lem:swap-blocks} gives $\lambda(v)=\ipc{v}{\vK_{ij}v}\ge4m_f^2$, since $\nrm{v}=1$, while for $v\in\CK_-$, it gives $\vK_{ij}v=\vK_{ij}\vPi^-_{ij}v=0$.
\end{proof}

Thus the spectra of $\vK_{ij}$ on the two swap sectors are separated by at least $4m_f^2$. We now use this separation to prove
\Cref{lem:K-detects-odd}, which we restate for convenience. 
\lemmKOdd*
\begin{proof}
	Let $\CK$ be the common eigenbasis of \Cref{cor:K-blocks}\,\ref{item:K-basis}, with
	$\vK_{ij}v=\lambda(v)v$, $\vU_{(ij)}v=\varepsilon(v)v$ ($\varepsilon(v)\in\{\pm1\}$) and $\vrho v=p(v)v$ ($p(v)>0$).
	Put $X_{vw}\coloneqq\ipc{v}{\vX w}$ for $v,w\in\CK$. By the self-adjoint property of $\vK_{ij}$ and $\vU_{(ij)}$ we can show that
	\[
		\ipc{v}{[\vK_{ij},\vX]w}=\left(\lambda(v)-\lambda(w)\right)X_{vw},
		\qquad
		\ipc{v}{\vX^-_{ij}w}=\tfrac12\left(1-\varepsilon(v)\varepsilon(w)\right)X_{vw}.
	\]
	In this basis $\nr{\vY}^2=\sum_{v,w\in\CK}\sqrt{p(v)p(w)}\,\labs{\ipc{v}{\vY w}}^2$, hence
	\[
		\nr{[\vK_{ij},\vX]}^2=\sum_{v,w\in\CK}\sqrt{p(v)p(w)}\,\labs{\lambda(v)-\lambda(w)}^2\labs{X_{vw}}^2,
		\qquad
		\nr{\vX^-_{ij}}^2=\sum_{\varepsilon(v)\neq\varepsilon(w)}\sqrt{p(v)p(w)}\,\labs{X_{vw}}^2 .
	\]
	If $\varepsilon(v)\neq\varepsilon(w)$, one of $v,w$ lies in $\CK_+$ and the other in $\CK_-$, so $\labs{\lambda(v)-\lambda(w)}\geq4m_f^2$ by \eqref{eq:K-eigenval-sector-bound}. Dropping all terms with $\varepsilon(v)=\varepsilon(w)$ from the first sum gives~\eqref{eq:K-detects-odd}.
\end{proof}

\subsection{Upper bound: the Dirichlet form controls \texorpdfstring{$[\vK_{ij},\vX]$}{[Kij,X]}}\label{sec:K-energy}
In this subsection we prove \Cref{lem:K-energy}, the upper bound in \eqref{eq:K-strategy}. The single-site Dirichlet forms $\CE_i(\vX)$ and $\CE_j(\vX)$ from~\eqref{eq:local-energy} integrate the commutators of $\vX$ with the time-evolved operators $\vB^{(i,\alpha)}(t)$ and $\vB^{(j,\alpha)}(t)$. Since $\vK_{ij}$ is a sum of squares of the operators $\vB^\alpha_{ij,+}=(\vB^{(i,\alpha)}+\vB^{(j,\alpha)})/\sqrt2$, the Leibniz rule writes
$[\vK_{ij},\vX]$ in terms of exactly such commutators, each multiplied by one further filtered operator. The proof therefore splits into two parts. We first pull the additional filtered factors out of the KMS norm, and then we bound each term individually.

For the first ingredient we use the following H\"older-type inequality, in which the price of multiplying by $\vA$ is the spectral norm of the modular conjugate $\vrho^{\pm1/4}\vA\vrho^{\mp1/4}$.
\begin{fact}[KMS H\"older inequality {\cite[Lemma~II.2]{bergamaschi2025FastMixingSpinChainsAllTemps}}]\label{fact:kms-holder}
	Let $\vrho$ be a full-rank state, and for $\vY\in\CB(\CH)$ let $\nr{\vY}^2\coloneqq\Tr\bigl[\vY^\dagger\vrho^{1/2}\vY\vrho^{1/2}\bigr]$.
	Then for all $\vA,\vY\in\CB(\CH)$,
	\[
		\nr{\vA\vY}\le\bigl\|\vrho^{1/4}\vA\vrho^{-1/4}\bigr\|_\infty\nr{\vY},
		\qquad
		\nr{\vY\vA}\le\bigl\|\vrho^{-1/4}\vA\vrho^{1/4}\bigr\|_\infty\nr{\vY}.
	\]
\end{fact}

For the second ingredient, we need these modular conjugates for $\vA=\vB^{(i,\alpha)}(t)$. Conjugating by $\vrho^{\pm1/4}$ multiplies
each Bohr component of $\vB^{(i,\alpha)}$ by $\e^{\mp\beta\omega/4}$, which turns the symmetrized filter $\widetilde f$ back into the
filter $f$. Hence the modular conjugates are Paulis reweighted by $f$, and their norms are controlled by $N_f$. \Cref{fact:pauli-selection-rule} then ensures that only three such reweighted pieces occur.
\begin{lemma}[Modular conjugates of the filtered operators]\label{lemma:rho-conjugated-B-bound}
    Let $\beta\geq0$, let $f$ be a KMS filter with symmetrized filter $\widetilde{f}$, and let
	$N_f\coloneqq\sup_{\labs{\omega}\leq1}\labs{f(\beta\,\omega)}$ as in~\eqref{eq:band-constants}.
	For every site $i$, Pauli direction $\alpha$ and time $t\in\BR$,
	\[
		\max\left\{\nrm{\vrho^{1/4}\vB^{(i,\alpha)}(t)\vrho^{-1/4}}_\infty,\ \nrm{\vrho^{-1/4}\vB^{(i,\alpha)}(t)\vrho^{1/4}}_\infty\right\}\le3N_f .
	\]
	Consequently, the modular conjugates of $\vB^\alpha_{ij,+}(t)=(\vB^{(i,\alpha)}(t)+\vB^{(j,\alpha)}(t))/\sqrt2$ have spectral norm at most $3\sqrt2\,N_f$..
\end{lemma}
\begin{proof}
	\emph{Reduction to $t=0$.} Since $\vrho$ commutes with $\e^{\ri t\vH_n}$, we have
	$\vrho^{c}\vB^{(i,\alpha)}(t)\vrho^{-c}=\e^{\ri t\vH_n}\left(\vrho^{c}\vB^{(i,\alpha)}\vrho^{-c}\right)e^{-\ri t\vH_n}$,
	and the spectral norm is unitarily invariant.

	\emph{Bohr decomposition.} Write $\vsigma=\vsigma^{(i,\alpha)}$, $\vB=\vB^{(i,\alpha)}$ and $c=\pm\frac14$.
	By \eqref{eq:bohr-identities}, $\vrho^{c}\vsigma_\omega\vrho^{-c}=\e^{-c\beta\omega}\vsigma_\omega$. Hence
	\[
		\vrho^{c}\vB\vrho^{-c}=\sum_{\omega\in B(\vH_n)}g_c(\omega)\,\vsigma_\omega,
		\qquad
		g_c(\omega)\coloneqq e^{-c\beta\omega}\widetilde f(\beta\omega)=
		\begin{cases}
			f(\beta\omega), & c=\tfrac14,\\[2pt]
			f^*(-\beta\omega), & c=-\tfrac14 .
		\end{cases}
	\]
	The second case uses $e^{\nu/4}\widetilde f(\nu)=\bigl(e^{\nu/4}\widetilde f(-\nu)\bigr)^*=f^*(-\nu)$. In both cases $\labs{g_c(\omega)}\leq N_f$ for $\labs{\omega}\le1$.

	Expanding $\vsigma_\omega=\sum_{E_{s'}-E_s=\omega}\vPi_{s'}\vsigma\vPi_s$ and using \Cref{fact:pauli-selection-rule} ($s'-s\in\{-1,0,1\}$),
	\begin{equation}\label{eq:rhoconjugatedB}
		\vrho^{c}\vB\vrho^{-c}=\vT_{-1}+\vT_0+\vT_1,
		\qquad
		\vT_k\coloneqq\sum_{s\in\CS_n}d_{s,k}\,\vPi_{s+k}\vsigma\vPi_s,
		\qquad
		d_{s,k}\coloneqq g_c\bigl(E_{s+k}-E_s\bigr).
	\end{equation}
	Here $\vPi_{s+k}\coloneqq0$ whenever $s+k\notin\CS_n$, and $\labs{d_{s,k}}\le N_f$, because the frequency $E_{s+k}-E_s$ lies in the band.

	\emph{Bounding each term.} Fix $k$ and $\psi\in\Hn$. The vectors $\vPi_{s+k}\vsigma\vPi_s\psi$ lie in mutually orthogonal subspaces for different $s$.
	Since $\vsigma$ is unitary,
	\[
		\nrm{\vT_k\psi}^2=\sum_s\labs{d_{s,k}}^2\nrm{\vPi_{s+k}\vsigma\vPi_s\psi}^2
		\leq N_f^2\sum_s\nrm{\vsigma\vPi_s\psi}^2=N_f^2\sum_s\nrm{\vPi_s\psi}^2=N_f^2\nrm{\psi}^2 .
	\]
	Thus $\nrm{\vT_k}_\infty\leq N_f$, and the triangle inequality in \eqref{eq:rhoconjugatedB} gives the bound $3N_f$.
	The statement for $\vB^\alpha_{ij,+}(t)$ follows from the triangle inequality.
\end{proof}

With these ingredients at hand, the upper bound follows in three steps. We first rewrite $\vK_{ij}$ in terms of time-evolved operators,
then bound the commutator at each fixed time. Finally, we integrate over time to recover the single-site Dirichlet forms.
\lemmKEnergy*
\begin{proof}
    \emph{Time invariance.} By \Cref{cor:K-blocks}\,\ref{item:K-comm-relations}, $[\vK_{ij},\vH_n]=0$, so
	$\vK_{ij}(t)\coloneqq\e^{\ri t\vH_n}\vK_{ij}\e^{-\ri t\vH_n}=\vK_{ij}$ for every $t\in\BR$.
	Expressing $\vK_{ij}(t)$ from its definition~\eqref{eq:K-definition}, we also obtain
	\begin{align*}
		\vK_{ij}(t)
		&=\e^{\ri t\vH_n}\Bigl(\sum_{\alpha\in\{\rmx,\rmy,\rmz\}}\vB^\alpha_{ij,+}\vB^\alpha_{ij,+}\Bigr)\e^{-\ri t\vH_n}\\
		&=\sum_{\alpha\in\{\rmx,\rmy,\rmz\}}\bigl(\e^{\ri t\vH_n}\vB^\alpha_{ij,+}\e^{-\ri t\vH_n}\bigr)\bigl(\e^{\ri t\vH_n}\vB^\alpha_{ij,+}\e^{-\ri t\vH_n}\bigr)\\
		&=\sum_{\alpha\in\{\rmx,\rmy,\rmz\}}\vB^\alpha_{ij,+}(t)^2.
	\end{align*}
	Hence $\vK_{ij}=\sum_\alpha\vB^\alpha_{ij,+}(t)^2$ for every $t\in\BR$.

    \emph{Pointwise bound.} Fix $t\in\BR$ and write $\vK^\alpha_{ij}(t)\coloneqq\vB^\alpha_{ij,+}(t)^2$, so that $\vK_{ij}=\sum_\alpha\vK^\alpha_{ij}(t)$.
	Abbreviate $\vB_+\coloneqq\vB^\alpha_{ij,+}(t)$ and set
	\[
		a_\alpha(t)\coloneqq\nr{[\vB^{(i,\alpha)}(t),\vX]},\qquad b_\alpha(t)\coloneqq\nr{[\vB^{(j,\alpha)}(t),\vX]},
		\qquad\text{so that}\qquad
		\nr{[\vB_+,\vX]}\leq\frac{a_\alpha(t)+b_\alpha(t)}{\sqrt2}.
	\]
	Utilizing the Leibniz rule for each $\alpha$ we obtain
	\begin{align}
		\nr{[\vK^\alpha_{ij}(t),\vX]}
		&=\nr{\vB_+[\vB_+,\vX]+[\vB_+,\vX]\vB_+}\notag\\
		&\le\bigl(\|\vrho^{1/4}\vB_+\vrho^{-1/4}\|_\infty+\|\vrho^{-1/4}\vB_+\vrho^{1/4}\|_\infty\bigr)\,\nr{[\vB_+,\vX]}
		\tag{triangle ineq.\ and \Cref{fact:kms-holder}}\\
		&\le2\cdot\frac{3N_f+3N_f}{\sqrt2}\cdot\frac{a_\alpha(t)+b_\alpha(t)}{\sqrt2}
		\tag{\Cref{lemma:rho-conjugated-B-bound}}\\
		&=6N_f\,(a_\alpha(t)+b_\alpha(t)) \label{eq:Kij-alpha-bound}.
	\end{align}
	By the triangle inequality and Cauchy--Schwarz, this further implies
	\begin{align*}
		\nr{[\vK_{ij},\vX]}^2
		&\le\Bigl(\sum_{\alpha\in\{\rmx,\rmy,\rmz\}}\nr{[\vK^\alpha_{ij}(t),\vX]}\Bigr)^2
		\le3\sum_{\alpha\in\{\rmx,\rmy,\rmz\}}\nr{[\vK^\alpha_{ij}(t),\vX]}^2\\
		&\le108\,N_f^2\sum_{\alpha\in\{\rmx,\rmy,\rmz\}}(a_\alpha(t)+b_\alpha(t))^2 \tag{by~\eqref{eq:Kij-alpha-bound}}\\
		&\le216\,N_f^2\sum_{\alpha\in\{\rmx,\rmy,\rmz\}}\left(a_\alpha^2(t)+b_\alpha^2(t)\right). \tag{Cauchy--Schwarz on $(a_\alpha(t)+b_\alpha(t))^2$}
	\end{align*}

	\emph{Integration.} The left-hand side does not depend on $t$, and $\mu_\beta(\BR)=\frac12$. Integrating both sides of the previous bound against $\rd\mu_\beta(t)$ and using \eqref{eq:local-energy}, i.e.,
	$\int_\BR\sum_\alpha\left(a_\alpha(t)^2+b_\alpha(t)^2\right)\,\rd\mu_\beta(t)=\CE_i(\vX)+\CE_j(\vX)$
    \begin{equation*}
        \frac12\nr{[\vK_{ij},\vX]}^2\le216N_f^2\left(\CE_i(\vX)+\CE_j(\vX)\right). \tag*{\qedhere}
    \end{equation*}
\end{proof}

\section{Low temperature gap bound via \texorpdfstring{$\SU$}{SU(2)} symmetry}
In this section we prove the low-temperature bound on the complement of the collective-spin-rotation-invariant sector.

\SUtwoComplement*

By the same argument as in~\Cref{sec:high-temp}, it suffices to prove
\begin{equation*}
    \EC(\vX) \geq \frac{c}{n}\nr{\vX}^2, \qquad \vX\in\Rn^\perp,
\end{equation*}
with $c$ independent of $n$.
The argument has two steps: the coherent Dirichlet form \eqref{eq:coh-form} is lower-bounded by collective spin commutators, and the latter is lower-bounded by the norm on $\Rn^\perp$.

The following two lemmas hold for every $n\ge2$ and $\beta\ge0$.
\begin{restatable}[Collective comparison]{lemma}{lemmCollective}\label{lem:collective-comparison}
	For every $\vX\in\CB(\Hn)$,
	\begin{equation}\label{eq:collective}
		\EC(\vX)\ \ge\ \frac{2f(0)^2}{n}\sum_{\alpha\in\{\rmx,\rmy,\rmz\}}\nr{[\Stot^\alpha,\vX]}^2 .
	\end{equation}
\end{restatable}
\begin{restatable}[Casimir spectral bound]{lemma}{lemmCasimir}\label{lem:collective-casimir}
    Define the Casimir operator of the collective conjugation action by
    \[
        \mathcal C_{\mathrm{ad}}[\cdot]
        :=\sum_{\alpha=x,y,z}[\Stot^\alpha,[\Stot^\alpha,\cdot]].
    \]
	The superoperator $\mathcal C_{\mathrm{ad}}$ is KMS-self-adjoint and positive semidefinite, with
	\begin{equation}\label{eq:casimir-form}
		\ip{\vX}{\mathcal C_{\mathrm{ad}}[\vX]}=\sum_{\alpha\in\{\rmx,\rmy,\rmz\}}\nr{[\Stot^\alpha,\vX]}^2
		\qquad\text{for all }\vX\in\CB(\Hn),
	\end{equation}
	and $\ker\mathcal C_{\mathrm{ad}}=\Rn$. Every eigenvalue of $\mathcal C_{\mathrm{ad}}$ lies in $\{\ell(\ell+1)\colon\ell\in\BZ_{\ge0}\}$. In particular, every eigenvalue of $\mathcal C_{\mathrm{ad}}$ on the invariant subspace $\Rn^\perp$ is at least $2$, i.e.,
	\begin{equation}\label{eq:collective-casimir-bound}
		\sum_{\alpha\in\{\rmx,\rmy,\rmz\}}\nr{[\Stot^\alpha,\vX]}^2\ \ge\ 2\,\nr{\vX}^2
		\qquad\text{for all }\vX\in\Rn^\perp .
	\end{equation}
\end{restatable}
\begin{proof}[Proof of \Cref{prop:R-complement}]
	The KMS condition~\eqref{eq:f-KMS} at $\nu=0$ reads $f(0)=f^*(0)$, so $f(0)\in\BR$; moreover, it is independent of $n$. For $\vX\in\Rn^\perp$,
	\begin{align*}
		\EC(\vX)
		&\ge\frac{2f(0)^2}{n}\sum_{\alpha\in\{\rmx,\rmy,\rmz\}}\nr{[\Stot^\alpha,\vX]}^2 \tag{\Cref{lem:collective-comparison}}\\
        &\ge\frac{2f(0)^2}{n}\cdot2\,\nr{\vX}^2=\frac{4 f(0)^2}{n}\,\nr{\vX}^2 \tag*{\qedhere}
	\end{align*}
\end{proof}

\subsection{Collective comparison}\label{sec:collective-comparison}
The proof of~\Cref{lem:collective-comparison} combines the sum of squares representation of the Dirichlet form with two aditional elementary observations.

The first is that a unitary change of basis among the jump operators does not change the sum of the squared norms of their commutators with a fixed operator.
\begin{lemma}[Unitary mixing of the jumps]\label{lem:unitary-mixing}
	Let $\vY_1,\dots,\vY_n\in\CB(\Hn)$ and let $\vU=(u_{ai})_{a,i}\in\mathrm U(n)$ be a unitary matrix with rows indexed by $a\in\{0,1,\dots,n-1\}$ and columns by $i\in[n]$. Then
	\begin{equation}\label{eq:parseval}
		\sum_{a=0}^{n-1}\nr{\sum_{i=1}^n u_{ai}\vY_i}^2=\sum_{i=1}^n\nr{\vY_i}^2 .
	\end{equation}
\end{lemma}
\begin{proof}
	Expanding the squared norms by sesquilinearity of the KMS inner product,
	\begin{align*}
		\sum_{a=0}^{n-1}\nr{\sum_{i=1}^n u_{ai}\vY_i}^2
		&=\sum_{a=0}^{n-1}\sum_{i,j=1}^n\overline{u_{ai}}\,u_{aj}\,\ip{\vY_i}{\vY_j}
		=\sum_{i,j=1}^n(\vU^\dagger\vU)_{ij}\,\ip{\vY_i}{\vY_j}\\
		&=\sum_{i,j=1}^n\delta_{ij}\,\ip{\vY_i}{\vY_j}
		=\sum_{i=1}^n\nr{\vY_i}^2 . \tag*{\qedhere}
	\end{align*}
\end{proof}
The second is that the uniform superposition of the single-site jumps is a collective spin operator, on which both the filter and the time evolution act trivially.
\begin{lemma}[The collective jump]\label{lem:collective-jump}
	For every $\alpha\in\{\rmx,\rmy,\rmz\}$,
	\begin{equation}\label{eq:collective-jump}
		\sum_{i=1}^n\vB^{(i,\alpha)}=2f(0)\,\Stot^\alpha ,
		\qquad\text{and}\qquad
		e^{\ri t\vH_n}\Stot^\alpha e^{-\ri t\vH_n}=\Stot^\alpha\quad\text{for all }t\in\BR .
	\end{equation}
\end{lemma}
\begin{proof}
	Since $\vH_n=-\frac1n\sum_\mu(\Stot^\mu)^2$ and $[\Stot^\alpha,\Stot^2]=0$, we have $[\Stot^\alpha,\vH_n]=0$, which gives the second identity. It also implies $\vP_{E'}\Stot^\alpha\vP_E=0$ for $E\neq E'$, i.e., the Bohr decomposition~\eqref{eq:Bohr-decomposition} of $\Stot^\alpha$ has only the zero-frequency component $\Stot^\alpha(0)=\Stot^\alpha$. Hence $h(\ad_{\vH_n})(\Stot^\alpha)=h(0)\Stot^\alpha$ for every function $h$ on the Bohr spectrum. By linearity of $\widetilde f(\beta\ad_{\vH_n})$ and $\sum_i\vsigma^{(i,\alpha)}=2\Stot^\alpha$,
	\[
		\sum_{i=1}^n\vB^{(i,\alpha)}
		=\widetilde f(\beta\ad_{\vH_n})\Bigl(\sum_{i=1}^n\vsigma^{(i,\alpha)}\Bigr)
		=\widetilde f(\beta\ad_{\vH_n})(2\Stot^\alpha)
		=2\widetilde f(0)\,\Stot^\alpha
		=2f(0)\,\Stot^\alpha,
	\]
	where in the last step $\widetilde f(0)=e^{0}f(0)=f(0)$ by~\eqref{eq:fsymm-KMS}.
\end{proof}

\lemmCollective*
\begin{proof}
	Choose a unitary $\vU=(u_{ai})\in\mathrm U(n)$ whose row $a=0$ is the uniform vector, $u_{0i}=1/\sqrt n$ for all $i\in[n]$ (extend the unit vector $(1/\sqrt n,\dots,1/\sqrt n)$ to an orthonormal basis of $\BC^n$), and define the mixed jumps
	\[
		\vC^{(a,\alpha)}\coloneqq\sum_{i=1}^n u_{ai}\,\vB^{(i,\alpha)},\qquad
		\vC^{(a,\alpha)}(t)\coloneqq e^{\ri t\vH_n}\vC^{(a,\alpha)} e^{-\ri t\vH_n}=\sum_{i=1}^n u_{ai}\,\vB^{(i,\alpha)}(t),
		\qquad 0\le a\le n-1 .
	\]
	By \Cref{lem:collective-jump}, the $a=0$ jump is a time-independent collective spin operator:
	\begin{equation}\label{eq:C0}
		\vC_0^\alpha=\frac1{\sqrt n}\sum_{i=1}^n\vB^{(i,\alpha)}=\frac{2f(0)}{\sqrt n}\,\Stot^\alpha,
		\qquad
		\vC_0^\alpha(t)=\vC_0^\alpha\quad\text{for all }t\in\BR .
	\end{equation}
	Fix $t\in\BR$ and $\alpha$, and apply \Cref{lem:unitary-mixing} with $\vY_i\coloneqq[\vB^{(i,\alpha)}(t),\vX]$; by linearity of the commutator $\sum_iu_{ai}\vY_i=[\vC^{(a,\alpha)}(t),\vX]$, hence
	\begin{equation}\label{eq:parseval-commutators}
		\sum_{i=1}^n\nr{[\vB^{(i,\alpha)}(t),\vX]}^2
		=\sum_{a=0}^{n-1}\nr{[\vC^{(a,\alpha)}(t),\vX]}^2
		\ \ge\ \nr{[\vC_0^\alpha(t),\vX]}^2
		=\frac{4f(0)^2}{n}\,\nr{[\Stot^\alpha,\vX]}^2,
	\end{equation}
	The right-hand side is independent of $t$. 
	This now gives for the Dirichlet form\begin{align*}
		\EC(\vX)
		&=\sum_{\alpha\in\{\rmx,\rmy,\rmz\}}\int_\BR\sum_{i=1}^n\nr{[\vB^{(i,\alpha)}(t),\vX]}^2\,\rd\mu_\beta(t) \tag{by~\Cref{lem:gt-dirichlet}}\\
		&\ge\frac{4f(0)^2}{n}\sum_{\alpha\in\{\rmx,\rmy,\rmz\}}\nr{[\Stot^\alpha,\vX]}^2\int_\BR\rd\mu_\beta(t) \tag{by~\eqref{eq:parseval-commutators}}\\
		&=\frac{2f(0)^2}{n}\sum_{\alpha\in\{\rmx,\rmy,\rmz\}}\nr{[\Stot^\alpha,\vX]}^2 . \tag{$\mu_\beta(\BR)=\frac12$~\qedhere}
	\end{align*}
\end{proof}

\subsection{Casimir spectral bound}

\begin{proof}[Proof of Lemma~\ref{lem:collective-casimir}]
The conjugation action $\Delta_{\vg}[\vX]=\vg^{\otimes n} \vX \left(\vg^{\otimes n}\right)^\dagger$ is KMS-unitary because
$[\vg^{\otimes n},\vrho]=0$. Omitting the factor $-\ri$, its infinitesimal generators are
\[
\CS^\alpha[\cdot]=[\Stot^\alpha,\cdot].
\]
Indeed,
$\left.\frac{d}{dt}\right|_{t=0}
e^{-\ri t \Stot^\alpha}\vX e^{\ri t \Stot^\alpha}=-\ri \CS^\alpha[\vX]~\forall \vX \in \CB(\CH_n)$.
Jacobi's identity gives $[\CS^\rmx,\CS^\rmy]=\ri\CS^\rmz$ and the cyclic versions, and $\CS^\alpha$ are both KMS- and Hilbert-Schmidt-self-adjoint. Indeed, since $[\vrho^{1/2},\Stot^\alpha]=0$,
\begin{align*}
    \ip{\vX}{[\Stot^\alpha,\vY]}
	=\Tr\bigl(\vrho^{1/2}(\vX^\dagger\Stot^\alpha
	-\Stot^\alpha \vX^\dagger)\vrho^{1/2}\vY\bigr)
	=\ip{[\Stot^\alpha,\vX]}{\vY}.
\end{align*}
The Casimir operator of the adjoint action is\footnote{In standard physics language the $\CS^\alpha$ and $\CC_{\mathrm{ad}}$ behave exactly as angular momentum components and the total squared angular momentum, but now they are superoperators acting on the ``state space'' $\CB(\CH_n)$. Since these operators are self-adjoint both wrt. the KMS- and Hilbert-Schmidt inner products, whichever is more convenient may be chosen.}
\begin{equation*}
    \CC_{\mathrm{ad}}[\cdot] = \sum_{\alpha=x,y,z}[\Stot^\alpha,[\Stot^\alpha,\cdot]] = \sum_{\alpha=\rmx,\rmy,\rmz} (\CS^\alpha)^2[\cdot],
\end{equation*}
We introduce the form \begin{equation*}
    \ip{\vX}{\CC_{\mathrm{ad}}\vX}=\sum_\alpha\nr{[\Stot^\alpha,\vX]}^2 \geq 0,
\end{equation*} which makes it transparent that $\CC_{\mathrm{ad}}$ is positive-semidefinite, and its kernel consists of exactly those $\vX$ that commute with every
$\Stot^\alpha$. Such operators commute with the corresponding exponentials, hence with
every collective rotation, and therefore \begin{equation}
    \label{eq:Cad-kernel}
    \ker \CC_{\mathrm{ad}}=\Rn.
\end{equation}
We may also write \begin{equation*}
    \CS^\pm := \CS^x\pm i\CS^y, \qquad \CC_{\mathrm{ad}}=\CS^-\CS^+ + \CS^z(\CS^z+\CI), \qquad [\CC_{\mathrm{ad}},\CS^\pm]=[\CC_{\mathrm{ad}},\CS^z]=0
\end{equation*} by the earlier arguments.

In standard physics language the $\CS^\alpha$ and $\Cad$ behave exactly as angular momentum components and the total squared angular momentum, but now they are superoperators acting on the ``state space'' $\CB(\CH_n)$. Since these operators are self-adjoint both wrt. the KMS- and Hilbert-Schmidt inner products, whichever is more convenient may be chosen. This observation suggests that $\Cad$ and $\CS^\rmz$ have a common eigenbasis with eigenvalues $\ell (\ell + 1) \geq 0$ (since $\Cad$ is positive semi-definite) and $M=-\ell, -\ell + 1, \ldots, \ell$, and the operator space $\CB(\CH_n)$ splits into KMS-orthogonal (and, simultaneously, Hilbert-Schmidt orthogonal) sectors labeled by $\ell,M$. Indeed, a Schur-Weyl decomposition analogous to \Cref{eq:schur-weyl} can also be applied to $\CB(\CH_n)$: \begin{equation*}
    \CB(\Hn)\cong
    \bigoplus_{s,t}(V_s\otimes V_t^*)\otimes\operatorname{Hom}(W_t,W_s),
    \qquad
    V_s\otimes V_t^*\cong\bigoplus_{\ell=|s-t|}^{s+t}V_\ell,
\end{equation*} with the second identity being the Clebsch--Gordan rule \cite[Appendix~C.1, Theorem~C.1]{hall2015LieGroupsAlgebrasReps}. Essentially, the operator space $\CB(\Hn)$ splits into spin-$\ell$ blocks resulting from the addition of a spin-$s$ and a spin-$t$, where $s,t \in \CS_n$ from \eqref{eq:schur-weyl}, and on each block $\Cad=\ell(\ell+1)\CI$. Since $s$ and $t$ are integer or half-integer simultaneously, $\ell$ is integer. We may also write \begin{equation*}
    \CS^\pm := \CS^x\pm i\CS^y, \qquad \Cad=\CS^-\CS^+ + \CS^z(\CS^z+\CI), \qquad [\Cad,\CS^\pm]=[\Cad,\CS^z]=0
\end{equation*} by the earlier arguments. The symbols $\CS^\pm$ behave as standard angular momentum ladder operators, stepping the $\CS^\rmz$ eigenvalue by one. For example, if $\CS^\rmz \vX = \ell \mathbf{\vX}$, then $\CS^+ \vX = \mathbf{0}$ and $\CS^\rmz \CS^- \vX = (\ell - 1) \CS^- \vX$, and so on.

By \Cref{eq:Cad-kernel} the $\ell=0$ eigenspace of $\Cad$ is exactly $\Rn$, and $\ell$ are non-negative integers, so any $\vX \in \Rn^\perp$ belongs to a $\ell \geq 1$ block. As discussed above, since $\Cad$ is
KMS-self-adjoint, eigenspaces corresponding to the distinct eigenvalues $\ell(\ell+1)$ are
KMS-orthogonal, thus if $\vX \in \Rn^\perp$, then $\vX=\sum_{\ell\ge1}\vX_\ell$ and
\begin{equation*}
  \ip{\vX}{\Cad\vX}
=\sum_{\ell\ge1}\ell(\ell+1)\nr{\vX_\ell}^2
\ge2\sum_{\ell\ge1}\nr{\vX_\ell}^2=2\nr{\vX}^2.  \qedhere
\end{equation*}

\end{proof}

\section{The critical case \texorpdfstring{$\beta=2$}{beta=2}}\label{sec:critical}

Throughout this section $\beta=2$, $n\ge2$, and \begin{equation*}
    \varrho(s):=\Tr(\vrho\vPi_s)
    =\frac{(2s+1)m_{n,s}e^{2s(s+1)/n}}
    {Z_{n,\beta}}
    =\frac{(2s+1)m_{n,s}e^{2s(s+1)/n}}
    {\displaystyle\sum_{t\in\CS_n}(2t+1)m_{n,t}e^{2t(t+1)/n}},
    \qquad
    \sum_{s\in\CS_n} \varrho(s) = 1.
\end{equation*} Here $\varrho(s)$ is the probability distribution of measurement outcomes $s$ of measuring the total spin $\sum_\alpha (\Stot^\alpha)^2$, not a permutation.
We work with the Davies--Metropolis generator $\LDM$ and its
Dirichlet form $\EDM$, following
\cite{basso2026SpecGapDaviesMeanFieldHeisenberg}.
All superoperators below are understood in
the Heisenberg picture. We use
\[
    \EDM(\vX):=\ip{\vX}{-\LDM[\vX]},
    \qquad
    \EC(\vX):=\ip{\vX}{-\LC[\vX]}.
\]
The corresponding estimates for the general $f$-weighted Davies
generator are obtained by comparison with the Metropolis rates in
Remark~\ref{rem:critical-general-f} at the end of this section.
All constants $c,C>0$ are independent of $n$.
We first prove the critical lower bound for the Davies--Metropolis generator.
We then deduce the corresponding coherent bound by comparing the two
forms on $\Pn$ and using Proposition~\ref{prop:P-complement} on
$\Pn^\perp$.

\criticalBound*

The main step is a Davies--Metropolis form bound on $\Pn$.
A spectral reduction leaves two families: scalar observables and
observables of the form $\Stot^z\sum_s o_s\vPi_s$.
For the second family, the substitution $g_s=so_s$ reduces the norm
estimate to bounds on $\Var_\varrho(g)$ and $|\mathbb E_\varrho g|^2$.
The scalar Poincar\'e inequality controls the variance, while the
estimate $\mathbb E_\varrho s^2=\Theta(n^{3/2})$ controls the mean.

\subsection{Reduction to two families of observables}

Write $\CS_n^+=\{s\in\CS_n:s>0\}$, and define
\begin{equation}\label{eq:critical-test-families}
    \vY_f:=\sum_{s\in\CS_n}f_s\vPi_s,
    \qquad
    \vX_o:=\Stot^z\sum_{s\in\CS_n^+}o_s\vPi_s.
\end{equation}
The coefficients may be complex. When $0\in\CS_n$, set $o_0=0$.
For scalar functions on $\CS_n$, let
\[
    \mathcal E_0(f):=\EDM(\vY_f),
    \qquad
    \Var_\varrho(f):=\mathbb E_\varrho|f-\mathbb E_\varrho f|^2.
\]
Since $\vrho$ is scalar on each spin block,
$\Var_{\vrho}(\vY_f)=\Var_\varrho(f)$.

\begin{fact}[Spectral reduction]\label{fact:critical-davies-ranks}
Let $\lambda>0$. If
\[
    \mathcal E_0(f)\ge\lambda\Var_\varrho(f)
    \quad\text{for every }f,
    \qquad
    \EDM(\vX_o)\ge\lambda\nr{\vX_o}^2
    \quad\text{for every }o,
\]
then
\[
    \EDM(\vX)\ge\lambda\Var_{\vrho}(\vX)
    \qquad\text{for every }\vX\in\Pn.
\]
\end{fact}

This follows from
\cite[Observation~3.16, Lemma~F.5, and Proposition~F.7]
{basso2026SpecGapDaviesMeanFieldHeisenberg}.
Under collective rotations, $\Pn$ is a KMS-orthogonal sum of spaces
of tensor ranks $\ell=0,\ldots,n$, each reducing the Davies--Metropolis generator.
The scalar observables $\vY_f$ form the rank-zero space.
For each positive rank, the restrictions of $-\LDM$ to the magnetic
components have the same spectrum, and their smallest eigenvalue is
nondecreasing with the rank. The proof of \cite[Lemma~F.8, Eqs.~(F.20)--(F.23)]
{basso2026SpecGapDaviesMeanFieldHeisenberg}
applies for every finite $\beta\ge0$: the anticommutator contribution
is rank independent, while the transition terms are
temperature-independent Pauli-twirl coefficients multiplied by
positive rates.
The component with rank one and magnetic label zero is exactly the
space of observables $\vX_o$
\cite[Eq.~(3.43) and Section~F.3]
{basso2026SpecGapDaviesMeanFieldHeisenberg}.
Thus a lower bound on this component controls every positive rank.
These components are orthogonal to the identity, so their variance
equals their squared norm. The form and variance split over the
reducing spaces, which proves the stated reduction.

For the coefficient estimates, set
\[
    \mathcal G_n:=\{s\in\CS_n^+:s-1\in\CS_n^+\},
\]
and define
\begin{equation}\label{eq:critical-B-definition}
    \mathfrak B(o):=
    \sum_{s\in\CS_n^+}\varrho(s)|o_s|^2
    +\sum_{s\in\mathcal G_n}\varrho(s)s^2|o_s-o_{s-1}|^2.
\end{equation}

\begin{lemma}[Coefficient estimates]\label{lem:critical-rank-input}
For every $o$,
\begin{equation}\label{eq:critical-rank-norm}
    \Tr(\vrho\vX_o)=0,
    \qquad
    \nr{\vX_o}^2
    =\frac13\sum_{s\in\CS_n^+}\varrho(s)s(s+1)|o_s|^2
    \le\sum_{s\in\CS_n^+}\varrho(s)s^2|o_s|^2,
\end{equation}
and
\begin{equation}\label{eq:critical-rank-energy}
    \EDM(\vX_o)\ge cn\,\mathfrak B(o).
\end{equation}
\end{lemma}

\begin{proof}
The observable $\vX_o$ commutes with $\vrho$, so its KMS norm satisfies
$\nr{\vX_o}^2=\Tr(\vrho\vX_o^\dagger\vX_o)$.
On a spin-$s$ block, its eigenvalues are $mo_s$, with
$m=-s,\ldots,s$, each repeated $m_{n,s}$ times. The identities
\[
    \sum_{m=-s}^s m=0,
    \qquad
    \frac1{2s+1}\sum_{m=-s}^s m^2=\frac{s(s+1)}3
\]
give the trace and norm formulas. The last inequality follows from
$s(s+1)/3\le s^2$ for $s\ge\frac12$.
This also agrees with \cite[Lemma~F.15]
{basso2026SpecGapDaviesMeanFieldHeisenberg}.

For the energy, retain the nonnegative contributions with $s\ge1$
in the proof of \cite[Lemma~F.16, Eqs.~(F.55)--(F.61)]
{basso2026SpecGapDaviesMeanFieldHeisenberg}. They give
\begin{equation}\label{eq:critical-energy-away-endpoint}
    \EDM(\vX_o)\ge cn\left(
    \sum_{\substack{s\in\CS_n\\s\ge1}}\varrho(s)|o_s|^2
    +\sum_{s\in\mathcal G_n}\varrho(s)(s^2-1)|o_s-o_{s-1}|^2
    \right).
\end{equation}
The prefactor in \cite[Eq.~(F.55)]
{basso2026SpecGapDaviesMeanFieldHeisenberg}
is $e^{-c_1\beta}$ with universal $c_1$, and Eqs.~(F.58)--(F.61)
are temperature independent; hence $c>0$ is uniform in $n$ at $\beta=2$.
Here and below, $\vS_i^\alpha:=\vsigma_i^\alpha/2$.
Our Pauli jumps are twice these spin operators, so the Davies--Metropolis form
is four times the form in that normalization; this factor is absorbed
in $c$.
Every $s\in\mathcal G_n$ satisfies $s\ge\frac32$, and hence
$s^2-1\ge\frac59s^2$.
For even $n$, \eqref{eq:critical-energy-away-endpoint} therefore proves
\eqref{eq:critical-rank-energy}.

For odd $n$, it remains to control $o_{1/2}$.
We first make the local-gradient comparison explicit.
If $\vZ_\omega=[\vS_i^\alpha(\omega),\vX_o]$, then
$[\vH_n,\vZ_\omega]=\omega\vZ_\omega$, because
$[\vH_n,\vX_o]=0$. Consequently,
\[
    \nr{\vZ_\omega}^2
    =e^{-\beta\omega/2}\Tr(\vrho\vZ_\omega^\dagger\vZ_\omega).
\]
The Davies--Metropolis divergence form thus gives, at $\beta=2$,
\[
    \EDM(\vX_o)=2\sum_{i,\alpha,\omega}f_{\mathrm M}^2(2 \omega)
    \Tr\!\left(\vrho
    [\vS_i^\alpha(\omega),\vX_o]^\dagger
    [\vS_i^\alpha(\omega),\vX_o]\right).
\]
By Fact~\ref{fact:pauli-selection-rule}, each relevant frequency satisfies
$|\omega|\le1$, so $f_{\mathrm M}^2(2 \omega)\ge e^{-2}$.
Distinct Bohr subspaces are orthogonal also for the inner product
$\Tr(\vrho\vY^\dagger\vZ)$. Summing their squared norms therefore yields
\begin{equation}\label{eq:critical-local-gradient}
    \EDM(\vX_o)\ge2e^{-2}
    \sum_{i,\alpha}\Tr\!\left(
    \vrho[\vS_i^\alpha,\vX_o]^\dagger
    [\vS_i^\alpha,\vX_o]\right).
\end{equation}

Separate the first spin from the remaining $n-1$ spins, and restrict
to the subspace on which the latter have total spin zero.
This subspace lies in the spin-$\frac12$ block and has dimension
$2m_{n-1,0}$, while the full block has dimension $2m_{n,1/2}$.
Since $\vrho$ is scalar on this block, the subspace has Gibbs weight
$r_n\varrho(\tfrac12)$, where \eqref{eq:multiplicity} gives
\[
    r_n:=\frac{m_{n-1,0}}{m_{n,1/2}}
    =\frac{n+3}{4n}\ge\frac14.
\]
The subspace is invariant under $\vS_1^\alpha$ and $\vX_o$.
On it, $\Stot^\alpha=\vS_1^\alpha$ and
$\vX_o=o_{1/2}\vS_1^z$. The spin-$\frac12$ commutation relations give
\[
    \sum_{\alpha=x,y}
    [\vS_1^\alpha,\vX_o]^\dagger[\vS_1^\alpha,\vX_o]
    =|o_{1/2}|^2\bigl((\vS_1^x)^2+(\vS_1^y)^2\bigr)
    =\frac{|o_{1/2}|^2}{2}\vI.
\]
The projection onto this subspace commutes with $\vrho$.
Keeping its contribution to the positive weighted trace therefore gives
\[
    \sum_\alpha\Tr\!\left(
    \vrho[\vS_1^\alpha,\vX_o]^\dagger[\vS_1^\alpha,\vX_o]\right)
    \ge\frac{r_n}{2}\varrho(\tfrac12)|o_{1/2}|^2.
\]
The full trace sum on the left is the same at every site, since
$\vrho$ and $\vX_o$ are permutation invariant. Thus
\eqref{eq:critical-local-gradient} implies
\[
    \EDM(\vX_o)\ge e^{-2}n r_n\varrho(\tfrac12)|o_{1/2}|^2
    \ge\frac{e^{-2}}4n\varrho(\tfrac12)|o_{1/2}|^2.
\]
Averaging this lower bound with
\eqref{eq:critical-energy-away-endpoint}, and decreasing the constant,
proves \eqref{eq:critical-rank-energy} for odd $n$ as well.
\end{proof}

\subsection{Two scalar estimates}

The first input is the critical Poincar\'e inequality for the
scalar spin chain. We also record an upper bound on its edge weights,
which will allow us to apply this inequality to $g_s=so_s$.

\begin{fact}[Scalar variance bound]\label{lem:critical-scalar}
For every $f\colon\CS_n\to\mathbb C$,
\begin{equation}\label{eq:critical-scalar}
    \Var_\varrho(f)\le Cn^{1/2}\mathcal E_0(f).
\end{equation}
Moreover,
\begin{equation}\label{eq:critical-scalar-edge-form}
    \mathcal E_0(f)=
    \sum_{\substack{s\in\CS_n\\s-1\in\CS_n}}
    a_s|f_s-f_{s-1}|^2,
    \qquad 0<a_s\le3n\varrho(s).
\end{equation}
\end{fact}

The variance bound is \cite[Theorem~3.9, $\beta=2$]
{basso2026SpecGapDaviesMeanFieldHeisenberg}, with the constant adjusted
for our Pauli normalization. The scalar chain is reversible with
stationary law $\pi$ and has positive rates between adjacent spins
\cite[Proposition~3.10]{basso2026SpecGapDaviesMeanFieldHeisenberg}.
Writing $q(s,t)$ for the rate from $s$ to $t$, reversibility gives
$a_s=\varrho(s)q(s,s-1)=\varrho(s-1)q(s-1,s)$.
The two directed terms for each edge combine into the single term
in \eqref{eq:critical-scalar-edge-form}.

For each local Pauli operator $\vA$, the spectral projections of
$\vH_n$ give
\[
    \sum_\omega(\vA_\omega)^\dagger\vA_\omega
    =\sum_s\vPi_s\vA^\dagger\vA\vPi_s=\vI.
\]
Since $\gamma_2\le1$, the total transition intensity from a spin block
due to one such operator is at most one: this follows by taking the
normalized trace on that block of
$\sum_\omega f_{\mathrm M}^2(2 \omega)(\vA_\omega)^\dagger\vA_\omega\le\vI$.
The rate of leaving the block is no larger.
There are $3n$ local Pauli operators, so $q(s,s-1)\le3n$.

The second input determines the scale of the spin under $\pi$.

\begin{lemma}[Critical second moment]\label{lem:critical-moment}
For every $n\ge2$,
\begin{equation}\label{eq:critical-spin-moment}
    \mathbb E_\varrho s^2=\Theta(n^{3/2}),
    \qquad
    \mathbb E_\varrho s(s+1)=\Theta(n^{3/2}).
\end{equation}
\end{lemma}

\begin{proof}
Set $x_s=2s/n$ and
\[
    F(x):=-\frac{x^2}{2}-H_b\!\left(\frac{1+x}{2}\right)+\log2,
    \qquad H_b(p):=-p\log p-(1-p)\log(1-p),
\]
where logarithms are natural and $0\log0:=0$.
The entropy expansion cancels the quadratic term at criticality:
\begin{equation}\label{eq:critical-quartic-free-energy}
    F(x)=\sum_{j\ge2}\frac{x^{2j}}{(2j)(2j-1)}
    \qquad(0\le x\le1).
\end{equation}
At $x=1$, the identity follows by continuity and convergence of the
series. Its first term gives the lower bound below. For $x\le1/2$,
bounding every denominator by $12$ gives
$F(x)\le x^4/[12(1-x^2)]\le x^4/9$. Thus
\begin{equation}\label{eq:critical-free-energy-bounds}
    F(x)\ge\frac{x^4}{12}\quad(0\le x\le1),
    \qquad F(x)\le\frac{x^4}{9}\quad(0\le x\le\tfrac12).
\end{equation}

Define
\[
    w_s:=\frac{(x_s+n^{-1})^2e^{x_s}}
    {(1-x_s^2+4n^{-1})^{1/2}}e^{-nF(x_s)},
    \qquad \mathcal Z_n:=\sum_{s\in\CS_n}w_s.
\]
By the uniform Stirling bound
\cite[Lemma~D.4]{basso2026SpecGapDaviesMeanFieldHeisenberg}, the
unnormalized Gibbs weight
$(2s+1)m_{n,s}e^{2s(s+1)/n}$ lies between constant multiples of
$2^n n^{1/2}w_s$, uniformly in $n\ge2$ and $s\in\CS_n$.
Summing these two bounds and then normalizing gives
\begin{equation}\label{eq:critical-weight-comparison}
    c\frac{w_s}{\mathcal Z_n}\le\varrho(s)
    \le C\frac{w_s}{\mathcal Z_n}.
\end{equation}
It therefore suffices to estimate the two sums
$\sum_s s^kw_s$ for $k\in\{0,2\}$, with $s^0:=1$.
The case $k=0$ gives the normalizing sum, and their ratio gives
the second moment.

We first bound both sums from above on $x_s\le1/2$.
On this interval, the factor
$e^{x_s}(1-x_s^2+4n^{-1})^{-1/2}$ is bounded above and below by
positive constants independent of $n$.
By \eqref{eq:critical-free-energy-bounds}, $nF(x_s)$ lies between
$nx_s^4/12$ and $nx_s^4/9$. The following rescaling turns these
two bounds into functions independent of $n$:
\[
    y_s=n^{1/4}x_s,
    \qquad s=\frac12n^{3/4}y_s,
    \qquad \delta_n=2n^{-3/4}.
\]
The $y_s$ lie on a grid of spacing $\delta_n$, starting at $0$ for
even $n$ and at $\delta_n/2$ for odd $n$.
Since
$x_s+n^{-1}=n^{-1/4}(y_s+n^{-3/4})\le n^{-1/4}(y_s+1)$,
we obtain
\begin{equation}\label{eq:critical-rescaled-summand}
    s^kw_s\le Cn^{(3k-2)/4}h_k(y_s),
    \qquad h_k(y):=y^k(y+1)^2e^{-y^4/12}.
\end{equation}
Here $h_0(0)=1$ and $h_2(0)=0$.
The functions $h_k$ no longer depend on $n$, and both are integrable
and bounded on $[0,\infty)$.

To pass from the discrete sum to an integral, we use their
unimodality. For $y>0$,
\[
    \frac{d}{dy}\log h_k(y)=\frac{k}{y}+\frac2{y+1}-\frac{y^3}{3},
    \qquad
    \frac{d^2}{dy^2}\log h_k(y)
    =-\frac{k}{y^2}-\frac2{(y+1)^2}-y^2<0.
\]
The first derivative is positive near zero and negative for large
$y$, so $h_k$ increases up to a unique maximum and then decreases.
On the increasing side, a grid value times $\delta_n$ is bounded by
the integral over the interval immediately to its right, provided
that interval stays on the increasing side. On the decreasing side,
we instead use the interval immediately to its left.
These intervals have disjoint interiors. The remaining grid points
lie within distance less than $\delta_n$ of the maximum, so there
are at most two of them. Each contributes at most $\sup h_k$.
Consequently, for either grid offset,
\[
    \sum_{x_s\le1/2}h_k(y_s)
    \le\delta_n^{-1}\int_0^\infty h_k(y)\,dy
       +2\sup_{y\ge0}h_k(y).
\]
All comparison intervals lie in $[0,\infty)$. Replacing their union
by this half-line only increases the integral, since $h_k$ is
nonnegative.
Since the integral and the supremum are fixed finite constants,
\eqref{eq:critical-rescaled-summand} yields
\begin{equation}\label{eq:critical-inner-moment-upper}
    \sum_{x_s\le1/2}s^kw_s
    \le Cn^{(3k-2)/4}(n^{3/4}+1)
    \le Cn^{(3k+1)/4}.
\end{equation}
Thus the extra factor $n^{3/4}$ comes precisely from the density
$\delta_n^{-1}$ of the rescaled grid.

For the remaining spins, $x_s>1/2$, the free energy is bounded
away from zero: $F(x_s)\ge1/192$.
Also
$(1-x_s^2+4n^{-1})^{-1/2}\le\sqrt n/2$,
while $e^{x_s}$ and $(x_s+n^{-1})^2$ are uniformly bounded above.
There are at most $n$ spin values and $s\le n/2$, so
\[
    \sum_{x_s>1/2}s^kw_s
    \le Cn^{k+3/2}e^{-n/192}
    \le Cn^{(3k+1)/4}.
\]
The last inequality holds uniformly in $n\ge2$, because an
exponentially decaying factor bounds every fixed polynomial factor.
Together with \eqref{eq:critical-inner-moment-upper}, this proves
the required upper bounds on the full sums.

For matching lower bounds, it is enough to retain the spins for
which $1\le y_s\le2$. For $n\ge256$, their set
\[
    I_n:=\{s\in\CS_n:n^{-1/4}\le x_s\le2n^{-1/4}\}
\]
lies in $x_s\le1/2$. In the spin variable, this is an interval of
length $\frac12n^{3/4}$ on a grid of spacing one. Hence
\[
    |I_n|\ge\frac12n^{3/4}-1\ge\frac14n^{3/4}.
\]
For every $s\in I_n$,
$nF(x_s)\le16/9$, $x_s+n^{-1}\ge n^{-1/4}$, and
$s\ge\frac12n^{3/4}$. The remaining prefactor in $w_s$ is bounded
below on this interval. Therefore $w_s\ge cn^{-1/2}$, and
\[
    \sum_s s^kw_s\ge\sum_{s\in I_n}s^kw_s
    \ge cn^{(3k+1)/4}
    \qquad(k\in\{0,2\}).
\]
We have proved
\[
    \mathcal Z_n=\Theta(n^{1/4}),
    \qquad \sum_s s^2w_s=\Theta(n^{7/4}).
\]
Taking their ratio in \eqref{eq:critical-weight-comparison} gives
$\mathbb E_\varrho s^2=\Theta(n^{3/2})$.
The constants can be adjusted to cover $2\le n<256$: there are only
finitely many such $n$, and each second moment is strictly positive.
Finally, $s^2\le s(s+1)\le3s^2$ for $s\ge1/2$, and both quantities
vanish at $s=0$. This proves the second assertion.
\end{proof}

\begin{corollary}
    For every $n\ge2$, $\ipc{\Stot^z}{\Stot^z}_{\vrho}=\Var_{\vrho}(\Stot^z)=\Theta(n^{\frac32})$.
\end{corollary}
\begin{proof}
     We have $\Var_{\vrho}(\Stot^z)=\frac13\sum_s\varrho(s)s(s+1)=\Theta(n^{\frac32})$ by \autoref{lem:critical-moment}.
\end{proof}

\subsection{Transferring the scalar estimate}

Let $\CP_P$ be the permutation average,
\begin{equation}\label{eq:critical-permutation-projection}
    \CP_P\vX:=\frac1{n!}\sum_{\tau\in S_n}
    \vU_\tau\vX\vU_\tau^\dagger,
    \qquad
    \vX_P:=\CP_P\vX,\qquad
    \vX_\perp:=(\CI-\CP_P)\vX.
\end{equation}
Since permutations commute with $\vrho$, $\CP_P$ is the
KMS-orthogonal projection onto $\Pn$ and preserves the Gibbs
expectation. Both superoperators are permutation-covariant and
KMS-self-adjoint, so
\begin{equation}\label{eq:form-orthogonal-split}
    \mathcal E_*(\vX)
    =\mathcal E_*(\vX_P)+\mathcal E_*(\vX_\perp),
    \qquad *\in\{\mathrm D,\mathrm{coh}\},
\end{equation}
and orthogonality gives
\begin{equation}\label{eq:variance-orthogonal-split}
    \Var_{\vrho}(\vX)
    =\Var_{\vrho}(\vX_P)+\nr{\vX_\perp}^2.
\end{equation}

\begin{proof}[Proof of Proposition~\ref{prop:critical-bound}]
The scalar family satisfies the required Davies--Metropolis bound by
\eqref{eq:critical-scalar}.
For the family $\vX_o$, put $g_s=so_s$ for $s>0$ and $g_0=0$ when
zero is an admissible spin. By \eqref{eq:critical-rank-norm},
\begin{equation}\label{eq:critical-variance-mean}
    \nr{\vX_o}^2\le\mathbb E_\varrho|g|^2
    =\Var_\varrho(g)+|\mathbb E_\varrho g|^2.
\end{equation}
For every $s\in\mathcal G_n$,
\begin{equation}\label{eq:critical-g-gradient}
    |g_s-g_{s-1}|^2
    =|(s-1)(o_s-o_{s-1})+o_s|^2
    \le2s^2|o_s-o_{s-1}|^2+2|o_s|^2.
\end{equation}
For even $n$, there is one additional edge, $0\leftrightarrow1$,
where $|g_1-g_0|^2=|o_1|^2$.
Thus \eqref{eq:critical-scalar-edge-form} implies
\[
    \mathcal E_0(g)\le6n\mathfrak B(o)\le C\EDM(\vX_o),
\]
where the second inequality is \eqref{eq:critical-rank-energy}.
Applying the scalar Poincar\'e inequality gives
\[
    \Var_\varrho(g)\le Cn^{1/2}\mathcal E_0(g)
    \le Cn^{1/2}\EDM(\vX_o).
\]

Although $\vX_o$ is centered, $g$ need not be.
By Cauchy--Schwarz, \eqref{eq:critical-spin-moment}, and the first
term of $\mathfrak B(o)$,
\begin{equation}\label{eq:critical-g-mean}
    |\mathbb E_\varrho g|^2
    \le\mathbb E_\varrho s^2\sum_{s\in\CS_n^+}\varrho(s)|o_s|^2
    \le Cn^{3/2}\mathfrak B(o)
    \le Cn^{1/2}\EDM(\vX_o).
\end{equation}
Combining the variance and mean bounds in
\eqref{eq:critical-variance-mean} yields
$\EDM(\vX_o)\ge cn^{-1/2}\nr{\vX_o}^2$.
Taking the smaller constant for the two families,
Fact~\ref{fact:critical-davies-ranks} now gives
\begin{equation}\label{eq:critical-C-target}
    \EDM(\vX)\ge c_Pn^{-1/2}\Var_{\vrho}(\vX)
    \qquad(\vX\in\Pn),
\end{equation}
where $c_P>0$ is independent of $n$.

The uniform Davies--Metropolis lower bound on $\Pn^\perp$ from
\cite[Theorem~3.5]{basso2026SpecGapDaviesMeanFieldHeisenberg},
together with \eqref{eq:form-orthogonal-split} and
\eqref{eq:variance-orthogonal-split}, extends this estimate to all
observables. The variational characterization of the gap gives
the gap bound in Proposition~\ref{prop:critical-bound} for the
Davies--Metropolis generator.
\end{proof}

Finally, \eqref{eq:davies-comparison} transfers
\eqref{eq:critical-C-target} to the coherent form on $\Pn$.
Combining it with Proposition~\ref{prop:P-complement} on $\Pn^\perp$
and the same orthogonal splittings gives
\begin{equation}\label{eq:critical-coherent-lower}
    \EC(\vX)\ge cn^{-1/2}\Var_{\vrho}(\vX)
    \qquad(\vX\in\CB(\Hn)).
\end{equation}
Consequently, $\gap(-\LC)\ge cn^{-1/2}$.

\begin{rmk}[General Davies weights]\label{rem:critical-general-f}
For every fixed continuous $f\colon\BR\to\BC$ satisfying
\eqref{eq:f-KMS}, $|f|\le1$, and
$\inf_{\nu\le0}|f(\nu)|^2>0$, the corresponding Davies generator satisfies
\[
    \gamma(\LDM^{(f)})=\Theta(n^{-1/2})
    \qquad(\beta=2).
\]
\end{rmk}

\begin{proof}
Write $\EDM^{(f)}$ for the KMS Dirichlet form of $\LDM^{(f)}$.
At $\beta=2$, set $q_f(\omega):=|f(2\omega)|^2$ and
$q_{\mathrm M}(\omega):=\min\{1,e^{-2\omega}\}$.
By \eqref{eq:f-KMS}, $|f|\le1$, and
$c_\star:=\inf_{\nu\le0}|f(\nu)|^2>0$,
\[
    q_f(\omega)=q_{\mathrm M}(\omega)|f(-2|\omega|)|^2,
    \qquad
    c_\star q_{\mathrm M}(\omega)
    \le q_f(\omega)\le q_{\mathrm M}(\omega).
\]
The generators $\LDM^{(f)}-c_\star\LDM$ and
$\LDM-\LDM^{(f)}$ have, respectively, rates
$q_f-c_\star q_{\mathrm M}$ and $q_{\mathrm M}-q_f$.
Both rate functions are nonnegative and satisfy
$q(\omega)=e^{-2\omega}q(-\omega)$.
Thus both differences are KMS-detailed-balanced Lindblad generators
with respect to the same Gibbs state $\vrho$, and their Dirichlet
forms are nonnegative. Consequently,
\[
    c_\star\EDM(\vX)\le\EDM^{(f)}(\vX)\le\EDM(\vX)
    \qquad(\vX\in\CB(\Hn)).
\]
The lower comparison and the Davies--Metropolis bound proved above give
\[
    \EDM^{(f)}(\vX)
    \ge c_\star c\,n^{-1/2}\Var_{\vrho}(\vX)
    \qquad(\vX\in\CB(\Hn)),
\]
so $\gamma(\LDM^{(f)})=\Omega(n^{-1/2})$ by the variational
characterization of the gap.
For the upper bound, Lemma~\ref{lem:critical-moment} gives
$\Var_{\vrho}(\Stot^z)=\Theta(n^{3/2})$, while
$\EDM(\Stot^z)=O(n)$ by
\cite[Lemma~3.22]{basso2026SpecGapDaviesMeanFieldHeisenberg},
with the constant adjusted for our Pauli normalization. Hence
\[
    \gamma(\LDM^{(f)})
    \le\frac{\EDM^{(f)}(\Stot^z)}{\Var_{\vrho}(\Stot^z)}
    \le\frac{\EDM(\Stot^z)}{\Var_{\vrho}(\Stot^z)}
    =O(n^{-1/2}).
\]
All constants are independent of $n$, with the lower-bound constant
allowed to depend on $f$ through $c_\star$.
\end{proof}

\section{The zero-temperature case}\label{sec:zero-temperature}

Let $f\colon\BR\to\BC$ be continuous and satisfy
\begin{equation}\label{eq:zero-temperature-filter}
    f(u)=e^{-u/2}f^*(-u),\qquad |f(u)|\le1,\qquad
    \inf_{u\le0}|f(u)|^2>0.
\end{equation}
Assume that $f(0)=c_0$ and $\lim_{u\to-\infty}f(u)=c_-$, where
$0<c_0\le c_-\le1$ are fixed independently of $n$. Put $\kappa=4c_0^2$.
For $f_\beta(\nu)=f(\beta\nu)$, the KMS condition gives
\[
    f_\infty(\nu):=\lim_{\beta\to\infty}f_\beta(\nu)
       =\begin{cases}
          c_-,&\nu<0,\\
          c_0,&\nu=0,\\
          0,&\nu>0.
        \end{cases}
\]
All limits are taken at fixed $n$, and all generators use the summed
single-site Pauli family $\{\vsigma^{(i,\alpha)}\}_{i,\alpha}$.
By \eqref{eq:Heisenberg-spectrum} and \eqref{eq:multiplicity}, the ground
space is the spin-$J$ representation with multiplicity one, where $J=n/2$.
Thus $\vrho_\infty:=\lim_{\beta\to\infty}\vrho_{n,\beta}=\vPi_J/(n+1)$.

For a jump operator $\vA$, write
\[
    \CL_{\vA}^\dagger[\vX]
       :=\vA\vX\vA^\dagger-\frac12\{\vA^\dagger\vA,\vX\}.
\]
Since increasing $s$ lowers the energy, \Cref{fact:pauli-selection-rule}
gives
\begin{align}
    \vL_\infty^{(i,\alpha)}
       &=c_-\sum_{s<J}\vPi_{s+1}\vsigma^{(i,\alpha)}\vPi_s
          +c_0\sum_s\vPi_s\vsigma^{(i,\alpha)}\vPi_s,
       \label{eq:zero-coherent-jumps}\\
    \CL_{\mathrm C,\infty}^\dagger
       &=\sum_{i,\alpha}\CL_{\vL_\infty^{(i,\alpha)}}^\dagger,\qquad
    \CL_{\mathrm D,\infty}^\dagger
       =\sum_{i,\alpha}\left(
          c_-^2\sum_{\nu<0}\CL_{\vsigma_\nu^{(i,\alpha)}}^\dagger
          +c_0^2\CL_{\vsigma_0^{(i,\alpha)}}^\dagger\right).
       \label{eq:zero-davies-generator}
\end{align}
Spin sums run through $\CS_n$, with inadmissible projections set to zero.
These are norm limits of the finite-temperature generators: the Bohr
spectrum is finite, and the coherent correction vanishes by
\Cref{lem:zero-coherent-term}.
The choices $f(u)=(1+e^{u/2})^{-1}$, $(1+e^u)^{-1/2}$, and
$\min\{1,e^{-u/2}\}$ give $(c_-,c_0)=(1,1/2)$,
$(1,1/\sqrt2)$, and $(1,1)$, respectively; these include coherent
Glauber, Davies--Glauber, and Metropolis sampling.

Fix $\sigma>0$ independently of $n$ and $\beta$, and denote by
$\CL_{\mathrm F,\sigma,\infty}$ the norm limit of
$\CL_{\mathrm F,\sigma,\beta}$ as $\beta\to\infty$.
Its existence is established in \Cref{sec:zero-oft}.

For any of these generators, let
\begin{align*}
    \CT_t&=e^{t\CL^\dagger},\qquad
    d_{\CL}(t):=\sup_{\bm{\eta}\text{ a state}}
       \frac12\|\CT_t\bm{\eta}-\vrho_\infty\|_1,\\
    t_{\mathrm{mix},\CL}(\varepsilon)
       &:=\inf\{t\ge0:d_{\CL}(t)\le\varepsilon\}.
\end{align*}

\begin{theorem}[Zero-temperature gap and mixing]\label{thm:zero-gap}
For $n\ge2$ and
$\CL\in\{\CL_{\mathrm C,\infty},\CL_{\mathrm D,\infty},
         \CL_{\mathrm F,\sigma,\infty}\}$,
the unique stationary state is $\vrho_\infty$, and
\[
    \gap(-\CL)=\Theta(n^{-1}),\qquad
    t_{\mathrm{mix},\CL}(\varepsilon)
       =\Theta\!\left(n\log\frac1\varepsilon\right)
       \quad(0<\varepsilon\le1/4).
\]
For the coherent and Davies generators, the implicit constants depend
only on $c_0$, uniformly over $c_-\in[c_0,1]$.
For the OFT generator, they depend only on $f$ and $\sigma$.
\end{theorem}

We first consider $\CL\in\{\CL_{\mathrm C,\infty},\CL_{\mathrm D,\infty}\}$.
The OFT case is treated in \Cref{sec:zero-oft}.

\subsection{Collective rotations}

Let $\CP_R$ be the Hilbert--Schmidt orthogonal projection onto $\Rn$,
given by Haar averaging:
\begin{equation}\label{eq:zero-twirl}
    \CP_R[\vX]:=\int_{\SU}\Delta_{\vg,\mathrm{id}}[\vX]\,\rmd\vg,
    \qquad \int_{\SU}\rmd\vg=1.
\end{equation}
The collective part of the generator can be extracted from the jumps.
Since $\sum_i\vsigma^{(i,\alpha)}=2\Stot^\alpha$ has only zero Bohr
frequency, $\sum_i\vL_\infty^{(i,\alpha)}=2c_0\Stot^\alpha$ and
$\sum_i\vsigma_0^{(i,\alpha)}=2\Stot^\alpha$. Hence the centered jumps
\[
    \widetilde{\vL}^{(i,\alpha)}
       :=\vL_\infty^{(i,\alpha)}-\frac{2c_0}{n}\Stot^\alpha,\qquad
    \vZ^{(i,\alpha)}:=\vsigma_0^{(i,\alpha)}-\frac2n\Stot^\alpha
\]
sum to zero over $i$. Expanding the dissipators cancels the cross terms
and gives
\begin{equation}\label{eq:zero-collective-split}
    \CL^\dagger=\kappa\CG+\CL_{\mathrm R}^\dagger,\qquad
    \CG[\vX]:=-\frac1{2n}\sum_\alpha
       [\Stot^\alpha,[\Stot^\alpha,\vX]],
\end{equation}
where $\CL_{\mathrm R}^\dagger$ is the Lindblad generator with jumps
$\widetilde{\vL}^{(i,\alpha)}$ in the coherent case, and
$c_-\vsigma_\nu^{(i,\alpha)}$ ($\nu<0$), $c_0\vZ^{(i,\alpha)}$ in the
Davies case.
Each triple transforms under rotations by the same orthogonal matrix
as the Pauli triple. Thus $\CL_{\mathrm R}^\dagger$ commutes with
collective conjugation, with $\CP_R$, and, by differentiation, with
$\CG$. Since $\CG\CP_R=0$, trace-norm contraction yields
\begin{equation}\label{eq:zero-angular-comparison}
    \CT_t=e^{t\CL_{\mathrm R}^\dagger}e^{\kappa t\CG},\qquad
    \frac12\|\CT_t\bm{\eta}-\CP_R[\CT_t\bm{\eta}]\|_1
       \le\frac12\|(e^{\kappa t\CG}-\CP_R)[\bm{\eta}]\|_1.
\end{equation}

On the symmetric ground space all compressed single-site Paulis agree,
so
\begin{equation}\label{eq:zero-ground-jumps}
    \vPi_J\vsigma^{(i,\alpha)}\vPi_J=\frac2n\Stot^\alpha\vPi_J,
    \qquad
    \widetilde{\vL}^{(i,\alpha)}\vPi_J
       =\vZ^{(i,\alpha)}\vPi_J=0,\qquad
    \vsigma_\nu^{(i,\alpha)}\vPi_J=0\quad(\nu<0).
\end{equation}
Every remainder jump annihilates the ground space. The restriction of
$\CL^\dagger$ to ground-supported operators is therefore $\kappa\CG$,
and $\vrho_\infty$ is stationary.

\subsection{Cooling through the spin sectors}

Set $\vK:=\sum_s(J-s)\vPi_s$ and $M=\lfloor n/2\rfloor$, so that
$\vI-\vPi_J\le\vK\le M\vI$.
The transition coefficient computed in
\cite[Eqs.~(C.12), (C.27)]{basso2026SpecGapDaviesMeanFieldHeisenberg}
gives, for $s<J$,
\begin{equation}\label{eq:zero-raising}
    \sum_{i,\alpha}\vPi_s\vsigma^{(i,\alpha)}
       \vPi_{s+1}\vsigma^{(i,\alpha)}\vPi_s
       =a_s\vPi_s,\qquad
    a_s:=\frac{2(2s+3)(J-s)}{2s+1}\ge2(J-s).
\end{equation}
Here the coefficient is four times that in the reference, whose
couplings are $\vsigma^{(i,\alpha)}/2$.
By rotation covariance, $\CL[\vK]\in\Rn$ and is therefore block
diagonal in $s$. For either generator, compressing the observable
dissipators to the $s$-block gives
\[
    \vPi_s\CL[\vK]\vPi_s
       =-c_-^2\sum_{i,\alpha}
          \vPi_s\vsigma^{(i,\alpha)}
          \vPi_{s+1}\vsigma^{(i,\alpha)}\vPi_s.
\]
Thus \eqref{eq:zero-raising}, with $a_J=0$, yields
\begin{equation}\label{eq:zero-drift}
    \CL[\vK]=-c_-^2\sum_s a_s\vPi_s\le-2c_-^2\vK.
\end{equation}
Consequently, for every state $\bm{\eta}$,
\begin{equation}\label{eq:zero-cooling}
    \Tr(\vK\CT_t\bm{\eta})
       \le e^{-2c_-^2t}\Tr(\vK\bm{\eta})
       \le Me^{-2c_-^2t}.
\end{equation}

For a state $\bm{\tau}$, the ground block of $\CP_R[\bm{\tau}]$
is $p\vrho_\infty$, where $p=\Tr(\vPi_J\bm{\tau})$; its remaining mass
$1-p$ has orthogonal support. Hence
\[
    \frac12\|\CP_R[\bm{\tau}]-\vrho_\infty\|_1
       =1-p\le\Tr(\vK\bm{\tau}),
\]
and the triangle inequality gives
\begin{equation}\label{eq:zero-distance-split}
    \frac12\|\bm{\tau}-\vrho_\infty\|_1
       \le\frac12\|\bm{\tau}-\CP_R[\bm{\tau}]\|_1
          +\Tr(\vK\bm{\tau}).
\end{equation}

\subsection{Mixing and the spectral gap}

\begin{proof}[Proof of \Cref{thm:zero-gap} for the coherent and Davies generators]
By \Cref{lem:collective-casimir}, $\CG$ has eigenvalue
$-\ell(\ell+1)/(2n)$ on each spin-$\ell$ component of the operator
space, where $\ell\in\{0,\ldots,n\}$.
Let $\CP_\ell$ project onto all copies of spin $\ell$, and let
$\chi_\ell$ be the character of $V_\ell$. Schur orthogonality gives
\[
    \CP_\ell[\vX]=(2\ell+1)\int_{\SU}
       \overline{\chi_\ell(\vg)}\Delta_{\vg,\mathrm{id}}[\vX]\,\rmd\vg.
\]
Moreover, $\int_{\SU}|\chi_\ell(\vg)|^2\,\rmd\vg=1$ and $\CP_0=\CP_R$.
Thus $\|\CP_\ell[\vX]\|_1\le(2\ell+1)\|\vX\|_1$ by unitary
invariance and Cauchy--Schwarz. For $u\ge n$ and $x=e^{-u/n}$,
\begin{align}
    \frac12\|(e^{u\CG}-\CP_R)[\bm{\eta}]\|_1
       &\le\frac12\sum_{\ell=1}^n(2\ell+1)
          e^{-\ell(\ell+1)u/(2n)}\nonumber\\
       &\le\frac12\sum_{\ell\ge1}(2\ell+1)x^\ell
        =\frac{x(3-x)}{2(1-x)^2}\le4x.
        \label{eq:zero-angular-decay}
\end{align}
The last inequality uses $x\le e^{-1}$.
Combining \eqref{eq:zero-angular-comparison},
\eqref{eq:zero-cooling}, and \eqref{eq:zero-distance-split} yields
\begin{equation}\label{eq:zero-mixing-upper}
    d_{\CL}(t)\le4e^{-\kappa t/n}+Me^{-2c_-^2t}
       \le5e^{-\kappa t/n}\qquad(t\ge n/\kappa),
\end{equation}
where the second inequality follows from
\[
    Me^{-(2c_-^2-\kappa/n)t}
       \le\frac n2 e^{1-n/2}\le1
       \qquad(t\ge n/\kappa),
\]
using $c_-\ge c_0$ and $n\ge2$. Convergence to $\vrho_\infty$
from every state proves uniqueness.

For the lower bound take
$\bm{\eta}_0=|\uparrow\rangle\langle\uparrow|^{\otimes n}$ and
$\vA=\vPi_J\Stot^\rmz\vPi_J/J$. The ground-space restriction is
$\kappa\CG$, which is Hilbert--Schmidt self-adjoint, and the spin
commutation relations give
\[
    \CG[\vA]=-\vA/n,\qquad
    \Tr(\vA\bm{\eta}_0)=1,\qquad
    \Tr(\vA\vrho_\infty)=0,\qquad \op{\vA}=1.
\]
Trace-norm duality therefore gives
\begin{equation}\label{eq:zero-mixing-lower}
    d_{\CL}(t)\ge\frac12
       |\Tr(\vA(\CT_t\bm{\eta}_0-\vrho_\infty))|
       =\frac12e^{-\kappa t/n}.
\end{equation}
Solving \eqref{eq:zero-mixing-upper} and \eqref{eq:zero-mixing-lower}
gives the claimed mixing time for $0<\varepsilon\le1/4$.

Finally, decomposing traceless operators into Hermitian parts and then
positive and negative parts extends \eqref{eq:zero-mixing-upper} to
an exponential norm bound on the traceless subspace. Every nonzero
eigenvalue of $\CL^\dagger$ has real part at most $-\kappa/n$, and
$\CL^\dagger[\vA]=-(\kappa/n)\vA$ attains this bound.
The finite-temperature generators have real spectra by KMS detailed
balance; norm convergence and continuity of the roots of their
characteristic polynomials preserve this property in the limit.
Thus $\gap(-\CL)=\kappa/n$.
\end{proof}

\subsection{Reduction of the OFT generator}
\label{sec:zero-oft}

\begin{proof}[Completion of the proof of \Cref{thm:zero-gap}]
At zero Bohr frequency, the rescaling $u=\beta\omega$ makes the
OFT rate independent of $\beta$. Its Gaussian averaging therefore
survives the zero-temperature limit, despite the vanishing
frequency width $\sigma/\beta$.

Put $h=|f|^2$ and
\[
    b_\sigma:=\mathbb E\,h\!\left(\sigma Z+\frac{\sigma^2}{2}\right),
    \qquad Z\sim N(0,1).
\]
Continuity and $f(0)>0$ imply $b_\sigma>0$.
Expanding \eqref{eq:oft-S-transform} in Bohr frequencies and completing
the square gives the coefficient
\[
    K_{\beta,\sigma}(\nu,\mu)
      =e^{-\beta^2(\nu-\mu)^2/(8\sigma^2)}
       \mathbb E\,h\!\left(
          \frac{\beta(\nu+\mu)}2+\sigma Z+\frac{\sigma^2}{2}
       \right).
\]
By dominated convergence, its off-diagonal entries tend to zero,
while the diagonal entries tend to $c_-^2$, $b_\sigma$, and $0$
for $\nu<0$, $\nu=0$, and $\nu>0$, respectively.
The Bohr spectrum is finite, so convergence holds in operator norm:
\begin{equation}\label{eq:zero-oft-limit}
    \CL_{\mathrm F,\sigma,\beta}^\dagger
    \longrightarrow
    \CL_{\mathrm F,\sigma,\infty}^\dagger
    =
    \sum_{i,\alpha}\left(
       c_-^2\sum_{\nu<0}\CL_{\vsigma_\nu^{(i,\alpha)}}^\dagger
       +b_\sigma\CL_{\vsigma_0^{(i,\alpha)}}^\dagger
    \right).
\end{equation}
This is \eqref{eq:zero-davies-generator} with $c_0^2$ replaced by
$b_\sigma$.

Write $\CL=\CL_{\mathrm F,\sigma,\infty}$.
The collective decomposition \eqref{eq:zero-collective-split}
holds with $\kappa=4b_\sigma$, and \eqref{eq:zero-cooling} is unchanged.
The first inequality in \eqref{eq:zero-mixing-upper} and the argument
for \eqref{eq:zero-mixing-lower} therefore give
\begin{align*}
    d_{\CL}(t)
       &\le4e^{-4b_\sigma t/n}+Me^{-2c_-^2t}
       &&\left(t\ge\frac{n}{4b_\sigma}\right),\\
    d_{\CL}(t)
       &\ge\frac12e^{-4b_\sigma t/n}
       &&(t\ge0).
\end{align*}
Set $m_\sigma:=\min\{c_-^2,b_\sigma\}$.
The upper bound is at most $5e^{-4m_\sigma t/n}$ for
$t\ge n/(4m_\sigma)$, since
\[
    Me^{-(2c_-^2-4m_\sigma/n)t}
       \le\frac n2e^{1-n/2}\le1.
\]
These bounds prove uniqueness and the claimed mixing time.
As in the preceding spectral argument, they also give
\[
    \frac{4m_\sigma}{n}
       \le\gap(-\CL)
       \le\frac{4b_\sigma}{n}.
\]
\end{proof}

Dividing the summed generators by $n$ divides their gaps by $n$
and multiplies their mixing times by $n$.

\subsection*{Acknowledgements and AI disclosure}

Research was carried out in part during a REU (research for undergraduates) project at the Rényi Institute; the study group was supervised by the last three authors, where this problem was stated as a potential research question to attack. LR kept working on this problem after the end of the REU program and developed the key ideas of the proof for the case of the coherent Lindblad generator; we are grateful for insightful discussions with the other members of the study group during the REU program: Anita Takács, Botond Sipos, Eszter Szabó and Martin Merencsics.
LLMs were used to help transfer the results from the coherent Lindbladians to the Operator Fourier Transform based ones, to develop the proof sketch in the zero temperature case, to find the elegant proof of \Cref{lem:spec-laplacian}, to find relevant references, to draft smalls sections of the text and to provide feedback on the writing.
LR and CCs are supported by the European Research Council (ERC) via the Starting grant QuanThermal (101222179). AG is supported by the Lendület ``Momentum'' program of the Hungarian Academy of Sciences under grant agreement no.~LP2025-8/2025. 

\section*{Nomenclature}\label{sec:recap_notation}
\addcontentsline{toc}{section}{Nomenclature}

This section recapitulates notation; the reader may skim through it and return as needed.
We write scalars, functions, and vectors in normal font, operators (matrices) in bold font~$\vO$, superoperators in calligraphic font~$\CT$, and super-superoperators in blackboard bold~$\BS$.
Calligraphic letters are also used for the Hilbert space $\CH_n$, the operator space $\CB(\CH_n)$, and Dirichlet forms $\CE$.
The imaginary unit $\ri$ and spatial index $\rmx,\rmy,\rmz$ are typeset in Roman font.
All Lindbladians are written in the Heisenberg picture; their Hilbert--Schmidt adjoints are the Schrödinger-picture generators.
Asymptotic notation $\bigO{\cdot},\bigOm{\cdot},\bigTh{\cdot}$ refers to $n\to\infty$ with all other parameters ($\beta$, $f$, $\sigma$) fixed.

\begingroup
\allowdisplaybreaks

\newcommand{\nomdescr}[1]{\parbox[t]{0.8\textwidth}{\raggedright #1}}

\paragraph{Abbreviations.}
\begin{alignat*}{2}
	\text{KMS}& &\qquad& \nomdescr{Kubo--Martin--Schwinger; the KMS inner product $\ip{\cdot}{\cdot}$ and KMS detailed balance (self-adjointness with respect to it)}\\
	\text{GNS}& && \nomdescr{Gelfand--Naimark--Segal; GNS detailed balance (self-adjointness with respect to $\Tr(\vrho\vX^\dagger\vY)$), satisfied by the Davies generator}\\
	\text{HS}& && \nomdescr{Hilbert--Schmidt (inner product, adjoint, orthogonal complement)}\\
	\text{OFT}& && \nomdescr{operator Fourier transform, \eqref{eq:OFT}}\\
	\text{irrep}& && \nomdescr{irreducible representation}\\
	\mathrm{D},\ \mathrm{DM}& && \nomdescr{subscripts for Davies and Davies--Metropolis, e.g., $\LD$, $\LDM$}\\
	\mathrm{C}& && \nomdescr{subscript for the ``coherent'' construction, e.g., $\LC$, $\EC$}\\
	\mathrm{F},\sigma& && \nomdescr{subscript for the OFT construction with width $\sigma$, e.g., $\LFs$}\\
	\mathrm{IP}& && \nomdescr{subscript for the interchange process, e.g., $\CL_{\mathrm{IP}}$ in \autoref{sec:high-temp}}
\end{alignat*}

\renewcommand{\nomdescr}[1]{\parbox[t]{0.525\textwidth}{\raggedright #1}}

\paragraph{Model and Gibbs state.}
\begin{alignat*}{2}
	n&\ge 2,\quad [n]:=\{1,\ldots,n\} &\qquad& \nomdescr{number of qubits (spins) and the set of sites}\\
	\CH_n&:=(\BC^2)^{\otimes n},\quad \CB(\CH_n) && \nomdescr{$n$-qubit Hilbert space and the operators acting on it}\\
	\vsigma^{(i,\alpha)}& && \nomdescr{Pauli-$\alpha$ operator on qubit $i\in[n]$, $\alpha\in\{\rmx,\rmy,\rmz\}$}\\
	\Stot^\alpha&:=\frac12\sum_{i=1}^n\vsigma^{(i,\alpha)},\quad \Stot^2:=\sum_{\alpha}(\Stot^\alpha)^2 && \nomdescr{collective spin operators and the quadratic Casimir}\\
	\vH_n&:=-\frac1n\sum_{\alpha\in\{\rmx,\rmy,\rmz\}}(\Stot^\alpha)^2 && \nomdescr{mean-field Heisenberg Hamiltonian, \Cref{eq:mean-field-heisenberg-model}}\\
	\beta&\in[0,\infty] && \nomdescr{inverse temperature
    }\\
	\vrho&=\vrho_{n,\beta}:=\frac{e^{-\beta\vH_n}}{Z_{n,\beta}},\quad Z_{n,\beta}:=\Tr(e^{-\beta\vH_n}) && \nomdescr{Gibbs state and partition function, \Cref{eq:Gibbs-state}}\\
	\vM_z&:=\Stot^z && \nomdescr{total magnetization, test observable for gap upper bounds}
\end{alignat*}

\renewcommand{\nomdescr}[1]{\parbox[t]{0.75\textwidth}{\raggedright #1}}

\paragraph{Spin sectors.}
\begin{alignat*}{2}
	s&\in\CS_n && \nomdescr{total-spin quantum numbers, $\CS_n=\{0,1,\ldots,\frac n2\}$ or $\{\frac12,\frac32,\ldots,\frac n2\}$ for even or odd $n$ (not to be confused with $S_n$), \Cref{eq:schur-weyl}}\\
	\CH_n&\cong\bigoplus_{s\in\CS_n}V_s\otimes W_s &\qquad& \nomdescr{Schur--Weyl decomposition; $V_s$ is the spin-$s$ irrep of $\SU$, $W_s$ an irrep of $S_n$}\\
	m_{n,s}&:=\dim W_s && \nomdescr{multiplicity of spin $s$, \Cref{eq:multiplicity}}\\
	\vPi_s& && \nomdescr{projection onto the spin-$s$ block $V_s\otimes W_s$}\\
	E_s&:=-\frac{s(s+1)}{n} && \nomdescr{eigenvalue of $\vH_n$ on the spin-$s$ block, \Cref{eq:Heisenberg-spectrum}}\\
    \varrho(s)&:=\Tr(\vrho\vPi_s) && \nomdescr{distribution of the total spin $s$ under the Gibbs state}
\end{alignat*}

\renewcommand{\nomdescr}[1]{\parbox[t]{0.6\textwidth}{\raggedright #1}}

\paragraph{Bohr decomposition and operator Fourier transform.}
\begin{alignat*}{2}
	\vH_n&=\sum_{E\in\mathrm{spec}(\vH_n)}E\,\vP_E &\qquad& \nomdescr{spectral decomposition; $\vP_E$ is the projection onto the (highly degenerate) eigenspace of $E$}\\
	\nu,\omega\in B(\vH_n)&:=\mathrm{spec}(\vH_n)-\mathrm{spec}(\vH_n) && \nomdescr{Bohr frequencies; $B(\vH)$ for a general Hamiltonian $\vH$}\\
	\vA_\nu&:=\sum_{E'-E=\nu}\vP_{E'}\vA\vP_E && \nomdescr{component of $\vA$ at Bohr frequency $\nu$, \Cref{eq:Bohr-decomposition}}\\
	\vA(t)&:=e^{\ri t\vH_n}\vA e^{-\ri t\vH_n} && \nomdescr{Heisenberg time evolution of $\vA$}\\
	\ad_{\vX}[\vY]&:=[\vX,\vY] && \nomdescr{commutator superoperator}\\
	h(\beta\,\ad_{\vH_n})[\vA]&:=\sum_{\nu\in B(\vH_n)}h(\beta\nu)\vA_\nu && \nomdescr{reweighting of the Bohr components by a function $h$}\\
	\widehat{\vA}(\omega)&:=\sum_{\nu\in B(\vH_n)}\hat g_\sigma(\omega-\nu)\vA_\nu && \nomdescr{operator Fourier transform (OFT) of $\vA$ at label $\omega\in\BR$, \Cref{eq:OFT}}\\
	g_\sigma,\ \hat g_\sigma& && \nomdescr{Gaussian time window and its Fourier transform, \Cref{eq:OFT}}\\
	\sigma&>0 && \nomdescr{OFT width; $\sigma/\beta$ is the energy resolution (not to be confused with the Pauli operators $\vsigma$)}
\end{alignat*}

\renewcommand{\nomdescr}[1]{\parbox[t]{0.48\textwidth}{\raggedright #1}}

\paragraph{Inner products, norms, and adjoints.}
\begin{alignat*}{2}
	\ipc{\vX}{\vY}&:=\Tr(\vX^\dagger\vY),\quad \|\vX\|_2^2:=\ipc{\vX}{\vX} &\qquad& \nomdescr{Hilbert--Schmidt inner product and norm}\\
	\ip{\vX}{\vY}&:=\Tr(\vrho^{1/2}\vX^\dagger\vrho^{1/2}\vY),\quad \nr{\vX}^2:=\ip{\vX}{\vX} && \nomdescr{KMS inner product and norm}\\
	\|\vX\|_p&:=(\Tr|\vX|^p)^{1/p},\quad \|\vX\|:=\|\vX\|_\infty && \nomdescr{Schatten $p$-norm; operator norm}\\
	\Var_{\vrho}(\vX)&:=\nr{\vX-\Tr(\vrho\vX)\vI}^2 && \nomdescr{variance of $\vX$ in the Gibbs state, \Cref{eq:VarDef}}\\
	\vX^\dagger,\ z^*& && \nomdescr{Hermitian conjugate of an operator; complex conjugate of a number or a function, $f^*(\nu):=f(\nu)^*$}\\
	\CM^\dagger& && \nomdescr{Hilbert--Schmidt adjoint of a superoperator $\CM$}\\
	O^\perp& && \nomdescr{Hilbert--Schmidt orthogonal complement of a subspace $O\subseteq\CB(\CH_n)$; it coincides with the KMS one for $O\in\{\Pn,\Rn\}$}\\
	\vI,\ \CI& && \nomdescr{identity operator; identity superoperator}
\end{alignat*}

\renewcommand{\nomdescr}[1]{\parbox[t]{0.525\textwidth}{\raggedright #1}}

\paragraph{Filters and filtered jumps.}
\begin{alignat*}{2}
	f\colon\BR&\to\BC &\qquad& \nomdescr{continuous filter with $f(\nu)=e^{-\nu/2}f^*(-\nu)$ (KMS condition) and $|f|\le1$, \Cref{eq:f-KMS}}\\
	c_\star&:=\inf_{\nu\le0}|f(\nu)|^2>0 && \nomdescr{normalization constant of the filter, \Cref{eq:f-KMS}}\\
	f_M(\nu)&:=\min\{1,e^{-\nu/2}\} && \nomdescr{Metropolis filter, \Cref{eq:Metropolis}}\\
	\widetilde f(\nu)&:=e^{\nu/4}f(\nu) && \nomdescr{symmetrized filter, $\widetilde f(-\nu)=\widetilde f^*(\nu)$, \Cref{eq:fsymm-KMS}}\\
	\vL^{(i,\alpha)}&:=f(\beta\,\ad_{\vH_n})[\vsigma^{(i,\alpha)}] && \nomdescr{jump operators of $\LC$, \Cref{eq:B-definition}}\\
	\vB^{(i,\alpha)}&:=\vrho^{-1/4}\vL^{(i,\alpha)}\vrho^{1/4}=\widetilde f(\beta\,\ad_{\vH_n})[\vsigma^{(i,\alpha)}] && \nomdescr{self-adjoint symmetrized jumps, \Cref{eq:B-definition}}\\
	\rmd\mu_\beta(t)&:=\frac{\rmd t}{\beta\cosh(2\pi t/\beta)},\quad \mu_0:=\frac12\delta_0 && \nomdescr{measure in the sum-of-squares form of $\EC$, \autoref{lem:form}; $\mu_\beta(\BR)=\frac12$}
\end{alignat*}

\renewcommand{\nomdescr}[1]{\parbox[t]{0.575\textwidth}{\raggedright #1}}

\paragraph{Lindbladians, discriminants, and Dirichlet forms.}
\begin{alignat*}{2}
	\CT[\cdot]&:=\sum_{i,\alpha}\vsigma^{(i,\alpha)}[\cdot]\vsigma^{(i,\alpha)} &\qquad& \nomdescr{``transition part'' quantum channel of single-site Paulis}\\
	\CT'&:=\BS^{(f)}\llbracket\CT\rrbracket && \nomdescr{KMS-detailed-balanced completely positive map obtained \textit{via} one of the super-superoperator $\BS^{(f)}_{\mathrm{D}}$, $\BS^{(f)}_{\mathrm{C}}$, $\BS^{(f)}_{\mathrm{F},\sigma}$ for the three constructions}\\
	\CL[\cdot]&:=\CT'[\cdot]-\frac12\{\CT'[\vI],\cdot\} && \nomdescr{KMS-detailed-balanced Lindbladian (Heisenberg picture), \Cref{eq:KMS-Lindbladian-general-form}; $\CL^\dagger$ is the Schrödinger-picture generator}\\
	&\TD^{(i,\alpha)},\ \LD^{(i,\alpha)},\ \LD,\ \LDM && \nomdescr{Davies transition map \eqref{eq:davies-generator-transition} and Lindbladians \eqref{eq:LindbladTermsDef}; $\LDM$ uses the Metropolis filter $f_M$}\\
	&\TC^{(i,\alpha)},\ \LC^{(i,\alpha)},\ \LC && \nomdescr{``coherent'' transition map \eqref{eq:coherent-generator-transition} and Lindbladians}\\
	&\TFs^{(i,\alpha)},\ \LFs^{(i,\alpha)},\ \LFs && \nomdescr{OFT transition map \eqref{eq:oft-S-transform} and Lindbladians}\\
	\CL&\in\{\LC,\LD,\LFs\} && \nomdescr{generic $\CL$: the statement holds for each construction}\\
	\CD[\cdot]&:=\vrho^{1/4}\CL[\vrho^{-1/4}\cdot\vrho^{-1/4}]\vrho^{1/4} = \CD^\dagger[\cdot] && \nomdescr{discriminant (the Hilbert--Schmidt self-adjointness is the KMS condition), \Cref{eq:DiscriminantDef}; $\DC,\DD,\DFs$ and $\DC^{(i,\alpha)},\DD^{(i,\alpha)},\DFs^{(i,\alpha)}$ analogously}\\
	\CE(\vX)&:=\ip{\vX}{-\CL[\vX]} && \nomdescr{Dirichlet form, \Cref{eq:Dirichlet-def}; $\EC,\ED,\EDM,\EF$ for $\LC,\LD,\LDM,\LFs$}
\end{alignat*}

\renewcommand{\nomdescr}[1]{\parbox[t]{0.75\textwidth}{\raggedright #1}}

\paragraph{Spectral gaps.}
\begin{alignat*}{2}
	\gamma(\CL)&:=\gap(-\CL) &\qquad& \nomdescr{spectral gap of $-\CL$,
    \Cref{eq:gap-rayleigh-quotient}}\\
	\gamma(\CL,O)& && \nomdescr{spectral gap of $-\CL$ restricted to an invariant subspace $O = \CL[O]$
    }
\end{alignat*}

\renewcommand{\nomdescr}[1]{\parbox[t]{0.575\textwidth}{\raggedright #1}}

\paragraph{Symmetries.}
\begin{alignat*}{2}
	\vg&\in\SU,\ \pi\in S_n &\qquad& \nomdescr{collective rotations and site permutations; $\bm{e}$ and $\mathrm{id}$ denote the respective identity elements}\\
	\vU_\pi& && \nomdescr{unitary permuting the sites: 
    site $i$ is moved to site $\pi(i)$
    }\\
	\vU_{(ij)}&=\frac12\Big(\vI+\sum_{\alpha}\vsigma^{(i,\alpha)}\vsigma^{(j,\alpha)}\Big) && \nomdescr{swap of sites $i$ and $j$, \Cref{eq:UijPauliDec}}\\
	\Delta_{\vg,\pi}[\cdot]&:=\vg^{\otimes n}\vU_\pi[\cdot](\vg^{\otimes n}\vU_\pi)^\dagger && \nomdescr{conjugation by the symmetry $(\vg,\pi)$
    }\\
	&\comm{O},\ \comm{\vO} && \nomdescr{commutant of a set $O$ of operators, resp.\ of a single operator $\vO$, \Cref{eq:CommDef}}\\
	\Rn&:=\comm{\{\Stot^x,\Stot^y,\Stot^z\}} && \nomdescr{collective-spin-rotation-invariant algebra, \autoref{def:Rn}}\\
	\Pn&:=\comm{\{\vU_{(ij)}\colon 1\le i<j\le n\}} && \nomdescr{permutation-invariant algebra, \autoref{def:Pn}}\\
	&\CP_R,\ \CP_P && \nomdescr{orthogonal projections onto $\Rn$ and $\Pn$ (Haar and permutation averages), \Cref{eq:zero-twirl} and \Cref{eq:critical-permutation-projection}}\\
	\Cad[\cdot]&:=\sum_{\alpha\in\{x,y,z\}}[\Stot^\alpha,[\Stot^\alpha,\cdot]] && \nomdescr{Casimir superoperator of the collective conjugation action, \autoref{lem:collective-casimir}}
\end{alignat*}

\renewcommand{\nomdescr}[1]{\parbox[t]{0.9\textwidth}{\raggedright #1}}

\paragraph{Miscellaneous.}
\begin{alignat*}{2}
	\Phi& &\qquad& \nomdescr{distribution function of the standard normal distribution}
\end{alignat*}
\endgroup

\addcontentsline{toc}{section}{References}
\bibliographystyle{alphaUrlePrint.bst}
\bibliography{ref,qc_gily}

\addtocontents{toc}{\protect\let\protect\numberline\protect\appendixnumberline}
\appendix

\section{Symmetries of the jumps: coherent terms are zero}\label{apx:coherent-term-zero}

All three Lindbladian constructions of the introduction share the same structure. 
In this appendix,
we highlight the similarity in 
the %
structure
of the Lindbladians
by introducing a parameter $\lambda$ ranging over a measure space $(\Lambda,\mu)$.
This allows defining
filtered jumps $\vL^{(i,\alpha)}(\lambda)$ and the completely positive map as follows:
\begin{equation}\label{eq:unified-CP}
    \CT'[\vX] = \int_\Lambda \sum_{i,\alpha} \vL^{(i,\alpha)}(\lambda)^\dagger\, \vX\, \vL^{(i,\alpha)}(\lambda)\,\rmd\mu(\lambda),
    \qquad
    \vD := \CT'[\vI],
\end{equation}
where
\begin{itemize}
    \item ``coherent'' construction: $\Lambda$ is a single point and $\vL^{(i,\alpha)}=\sum_{\nu\in B(\vH_n)}f(\beta\nu)\vsigma^{(i,\alpha)}_\nu$, so that $\CT'=\sum_{i,\alpha}\TC^{(i,\alpha)}$ as in \Cref{eq:coherent-generator-transition};
    \item Davies construction: $\Lambda=B(\vH_n)$ with counting measure and $\vL^{(i,\alpha)}(\nu)=f(\beta\nu)\vsigma^{(i,\alpha)}_\nu$, so that $\CT'=\sum_{i,\alpha}\TD^{(i,\alpha)}$ as in \Cref{eq:davies-generator-transition};
    \item OFT construction: $\Lambda=\BR$ with Lebesgue measure and $\vL^{(i,\alpha)}(\omega)=f\big(\beta\omega+\tfrac{\sigma^2}{2}\big)\widehat{\vsigma^{(i,\alpha)}}(\omega)$, with $\widehat{\vsigma^{(i,\alpha)}}(\omega)=\sum_{\nu\in B(\vH_n)}\hat g_\sigma(\omega-\nu)\vsigma^{(i,\alpha)}_\nu$ as in \eqref{eq:OFT}, so that $\CT'=\sum_{i,\alpha}\TFs^{(i,\alpha)}$ as in \Cref{eq:oft-S-transform}.
\end{itemize}
As before, $\vA_\nu$ denotes the Bohr frequency component \eqref{eq:Bohr-decomposition}; writing $\vP_E$ for the spectral projection of $\vH_n$ onto the eigenvalue $E$, we have $\vA_\nu=\sum_{E'-E=\nu}\vP_{E'}\vA\vP_E$.
In every case $\vL^{(i,\alpha)}(\lambda)=\sum_{\nu\in B(\vH_n)} c_\lambda(\nu)\,\vsigma^{(i,\alpha)}_\nu$ with coefficients $c_\lambda(\nu)\in\BC$ that do not depend on $a$. In general, the KMS-detailed-balanced Lindbladian built from $\CT'$ reads (in the Heisenberg picture)
\begin{equation}\label{eq:L-with-coherent-term}
    \CL[\cdot] = \CT'[\cdot]-\frac12\{\vD,\cdot\} + \ri[\vC,\cdot].
\end{equation}
Here, by \cite[Lemma 1]{gilyen2024QGeneralGlauberMetropolis}, the coherent term for any of the constructions is \begin{equation*}
    \vC = \sum_{\nu \in B(\Hn)} \frac\ri2 \tanh\left(\frac{\beta\nu}{4}\right) \vD_\nu,
\end{equation*} a linear combination of the Bohr components $\vD_\nu$ with $\nu \neq 0$ (the weight
$\tanh(\beta\nu/4)$
vanishes at $\nu=0$).

Below we explicitly show that $\vC=\mathbf{0}$, and that a similar line of argument proves that the Lindbladian generators $\CL$ commute with the conjugation action $\Delta_{\vg,\pi}$, i.e., $\CL \circ \Delta_{\vg,\pi} =\Delta_{\vg,\pi} \circ \CL$.
Throughout,
let
\begin{equation*}
    \begin{aligned}
        \Delta_{\vg,\pi}[\vX] = \vW\vX\vW^\dagger \qquad \text{where} \qquad \vW = \vg^{\otimes n}\vU_\pi \qquad \text{and} \qquad \vg \in \SU,~\pi \in S_n;
    \end{aligned}
\end{equation*}
the permutation $\pi$ maps index $i$ to $\pi(i)$.

\begin{lemma}[The filtered jumps rotate orthogonally]\label{lem:jump-covariance}
    Let $\vO(\vg)\in\mathrm{SO}(3)$ be defined by $\vg\vsigma^\alpha\vg^\dagger=\sum_\gamma O_{\alpha\gamma}(\vg)\vsigma^\gamma$, where $\vsigma^\alpha$ ($\alpha\in\{\rmx,\rmy,\rmz\}$) are the single-qubit Pauli matrices, and let
    \begin{equation*}
        R_{(i,\alpha),(j,\gamma)} := \delta_{j,\pi(i)}\,O_{\alpha\gamma}(\vg).
    \end{equation*}
    Then $\vR$ is a real orthogonal $3n\times 3n$ matrix, and for every $a$ and $\lambda\in\Lambda$,
    \begin{equation}\label{eq:jump-covariance}
        \vW\vL^{(i,\alpha)}(\lambda)\vW^\dagger = \sum_b R_{ab}\,\vL^{b}(\lambda),
        \qquad
        \vW\vL^{(i,\alpha)}(\lambda)^\dagger\vW^\dagger = \sum_b R_{ab}\,\vL^{b}(\lambda)^\dagger.
    \end{equation}
\end{lemma}
\begin{proof}
    \emph{Transformation of Bohr components.} The transformation law $\vg\vsigma^\alpha\vg^\dagger=\sum_\gamma O_{\alpha\gamma}(\vg)\vsigma^\gamma$ is found in standard quantum mechanics texts, e.g., \cite[Chapter 3]{Sakurai2020ModernQuantumMechanics}; importantly, $\vO(\vg) \in \mathrm{SO}(3)$, an orthogonal matrix.
    Then $\vg^{\otimes n}\Stot^\alpha(\vg^{\otimes n})^\dagger=\sum_\gamma O_{\alpha\gamma}(\vg)\Stot^\gamma$, and $\sum_\alpha(\Stot^\alpha)^2$ is rotation invariant; it is also permutation invariant, as each $\Stot^\alpha$ is a sum over all single-site Paulis. Hence, by \eqref{eq:mean-field-heisenberg-model}, $[\vW,\vH_n]=0$, so $\vW$ commutes with every spectral projection $\vP_E$ of $\vH_n$, and, by \eqref{eq:Bohr-decomposition},
    \begin{equation*}
        \vW\vA_\nu\vW^\dagger=(\vW\vA\vW^\dagger)_\nu\qquad\text{for all }\vA\in\CB(\CH_n),\ \nu\in B(\vH_n).
    \end{equation*}
    
    \emph{Transformation of single-site Paulis.} Since $\vg^{\otimes n}$ acts as $\vg$ on every site and $\vU_\pi$ moves site $i$ to site $\pi(i)$,
    \begin{equation*}
        \vW\vsigma^{(i,\alpha)}\vW^\dagger=\sum_\gamma O_{\alpha\gamma}(\vg)\,\vsigma^{(\pi(i),\gamma)}=\sum_{j,\gamma} R_{(i,\alpha),(j,\gamma)}\,\vsigma^{(j,\gamma)}.
    \end{equation*}
    $\vR$ is a matrix representing permutation of site indices and rotation of spatial directions, hence real orthogonal: \begin{equation*}
        (\vR^{\rmT} \vR)_{(j,\gamma),(k,\delta)} = \delta_{jk} \delta_{\gamma,\delta}
    \end{equation*}
    
    \emph{Combine.} Combining the two identities above, $\vW\vsigma^{(i,\alpha)}_\nu\vW^\dagger=\sum_{j,\gamma} R_{(i,\alpha),(j,\gamma)}\vsigma^{(j,\gamma)}_\nu$. Since $\vL^{(i,\alpha)}(\lambda)=\sum_\nu c_\lambda(\nu)\vsigma^{(i,\alpha)}_\nu$ with coefficients independent of $(i,\alpha)$, linearity gives the first identity in \eqref{eq:jump-covariance}; the second follows by taking adjoints, as $\vR$ is real.
\end{proof}

\begin{lemma}[Coherent terms are zero]\label{lem:zero-coherent-term}
    For the mean-field Heisenberg model $\vH_n$ with single-site Pauli jumps, and for each of the ``coherent'', Davies and OFT constructions, $\vD=\CT'[\vI]$ satisfies $[\vD,\vH_n]=0$. Consequently $\vC=0$ in \eqref{eq:L-with-coherent-term}, and $\CL$ is of the form \eqref{eq:KMS-Lindbladian-general-form}.
\end{lemma}
\begin{proof}
    \emph{$\vD$ is invariant under permutations and global rotations.} By \autoref{lem:jump-covariance}, for every $\lambda$,
    \begin{equation*}
        \vW\Big(\sum_{i,\alpha}\vL^{(i,\alpha)}(\lambda)^\dagger\vL^{(i,\alpha)}(\lambda)\Big)\vW^\dagger
        =\sum_{j,\gamma,k,\delta}\big(\vR^{\rmT}\vR\big)_{(j,\gamma),(k,\delta)}\,\vL^{(j,\gamma)}(\lambda)^\dagger\vL^{(k,\delta)}(\lambda)
        =\sum_{j,\gamma}\vL^{(j,\gamma)}(\lambda)^\dagger\vL^{(j,\gamma)}(\lambda).
    \end{equation*}
    Integrating against the measure $\mu$ in \Cref{eq:unified-CP} gives $\vW\vD\vW^\dagger=\vD$ for all $(\vg,\pi)$.

    \emph{$\vD$ commutes with $\vH_n$.} Since $\vD$ is invariant under all global rotations and permutations, it must commute with the global rotation generators $\Stot^\alpha$ and hence with $\Hn$. In more detail, take $\pi=\mathrm{id}$ and $\vg=\exp\!\big(-\ri t\tfrac12\sum_\alpha n_\alpha\vsigma^\alpha\big)$ for a real unit vector $\vn$, so that $\vg^{\otimes n}=\exp\!\big(-\ri t\sum_\alpha n_\alpha\Stot^\alpha\big)$. Differentiating $\vW\vD\vW^\dagger=\vD$ at $t=0$ gives $[\sum_\alpha n_\alpha\Stot^\alpha,\vD]=0$ for every $\vn$, hence $[\Stot^\alpha,\vD]=0~\forall\alpha\in\{\rmx,\rmy,\rmz\}$, i.e., $\vD\in\Rn$ (\Cref{def:Rn}). By \eqref{eq:RnInCHn}, $[\vD,\vH_n]=0$.

    \emph{$\vC$ vanishes.} Since $[\vD,\vH_n]=0$, we have $\vP_{E'}\vD\vP_E=0$ whenever $E'\neq E$, i.e., $\vD_\nu=0$ for all $\nu\neq0$. As $\vC$ is a linear combination of these components, $\vC=0$.
\end{proof}

\begin{corollary}[Covariance of the Lindbladians]\label{cor:covariance}
    For $\CL\in\{\LC,\LD,\LFs\}$ and every $(\vg,\pi)\in\SU\times S_n$,
    \begin{equation*}
        \CL\circ\Delta_{\vg,\pi}=\Delta_{\vg,\pi}\circ\CL
        \qquad\text{and}\qquad
        \CL^\dagger\circ\Delta_{\vg,\pi}=\Delta_{\vg,\pi}\circ\CL^\dagger.
    \end{equation*}
    This is the extension of \cite[Theorem 3.2, Remark B.3]{basso2026SpecGapDaviesMeanFieldHeisenberg} to all Lindbladians $\CL$ considered in this work.
\end{corollary}
\begin{proof}
    By \autoref{lem:zero-coherent-term}, $\CL[\vX]=\CT'[\vX]-\frac12\{\vD,\vX\}$. By \autoref{lem:jump-covariance}, for every $\lambda$ and $\vX \in \CB(\Hn)$,
    \begin{equation*}
        \begin{aligned}
            \vW\Big(\sum_{i,\alpha}\vL^{(i,\alpha)}(\lambda)^\dagger\vX\vL^{(i,\alpha)}(\lambda)\Big)\vW^\dagger
            &=\sum_{j,\gamma,k,\delta}\big(\vR^{\rmT} \vR\big)_{(j,\gamma),(k,\delta)}\,\vL^{(j,\gamma)}(\lambda)^\dagger\,\Delta_{\vg,\pi}[\vX]\,\vL^{(k,\delta)}(\lambda)\\
            &=\sum_{j,\gamma}\vL^{(j,\gamma)}(\lambda)^\dagger\,\Delta_{\vg,\pi}[\vX]\,\vL^{(j,\gamma)}(\lambda),
        \end{aligned}
    \end{equation*}
    so integrating over $\lambda$ gives $\Delta_{\vg,\pi}\circ\CT'=\CT'\circ\Delta_{\vg,\pi}$. Since $\vW\vD\vW^\dagger=\vD$ (proof of \autoref{lem:zero-coherent-term}), also $\Delta_{\vg,\pi}[\{\vD,\vX\}]=\{\vD,\Delta_{\vg,\pi}[\vX]\}$. This proves the first identity. As $\Delta_{\vg,\pi}$ is unitary for the Hilbert--Schmidt inner product $\ipc{\cdot}{\cdot}$, $\Delta_{\vg,\pi}^\dagger=\Delta_{\vg,\pi}^{-1}$, and taking adjoints of the first identity gives the second.
\end{proof}

\begin{rmk}
    Only the $\SU$ part of the symmetry is needed for \autoref{lem:zero-coherent-term}. For the Davies construction the statement is immediate, since each $(\vsigma^{(i,\alpha)}_\nu)^\dagger\vsigma^{(i,\alpha)}_\nu$ already commutes with $\vH_n$.
\end{rmk}
\providecommand{\vA}{\boldsymbol{A}}
\providecommand{\CC}{\mathcal{C}}
\providecommand{\vL}{\boldsymbol{L}}
\providecommand{\vM}{\boldsymbol{M}}
\providecommand{\vD}{\boldsymbol{D}}
\providecommand{\vO}{\boldsymbol{O}}
\providecommand{\vP}{\boldsymbol{P}}
\providecommand{\CE}{\mathcal{E}}
\providecommand{\Filt}{\mathfrak{F}}
\providecommand{\gap}{\lambda_{\mathrm{gap}}}
\providecommand{\Var}{\mathrm{Var}}

\section{Transferring spectral gaps from coherent to OFT Lindbladians}\label{sec:gap-transfer}

In this appendix we show how a spectral gap lower bound for the coherent Lindbladian transfers to the OFT Lindbladian. The rough idea is to use the integral representation of~\cite[Lemma~VI.1]{bergamaschi2025FastMixingSpinChainsAllTemps} to show that the Dirichlet form of the OFT Lindbladian with Metropolis weight decomposes as an average of Dirichlet forms of coherent Lindbladians over an explicit family of KMS filters. Any lower bound on these coherent forms that holds uniformly over the family, on a subspace invariant under the OFT Lindbladian, therefore passes to the OFT Dirichlet form. Since we do not assume any symmetry of the Hamiltonian here, we keep the coherent (Hamiltonian) parts of the Lindbladians.

\subsection{Setting}

Throughout $0<\beta<\infty$ and $\sigma>0$. For a self-adjoint jump $\vK=\vK^\dagger$ let $\Omega_{\vK}:=\{\nu\in B(\vH)\colon\vK_\nu\neq0\}$; we have
\begin{equation}\label{eq:gt-bohr}
	\vK_\nu^\dagger=\vK_{-\nu},\qquad \vrho^{s}\vK_\nu\vrho^{-s}=e^{-s\beta\nu}\vK_\nu\qquad(s\in\BR).
\end{equation}

\paragraph{Filters.}
A \emph{KMS filter} is a continuous $h\colon\BR\to\BC$ satisfying \eqref{eq:f-KMS}, i.e.\
$h(x)=\e^{-x/2}h^*(-x)$. For $c_\star\in(0,1]$ let
\[
	\Filt(c_\star):=\Big\{h\ \text{KMS filter}\colon\ |h|\le1,\ \inf_{x\le0}|h(x)|^2\ge c_\star\Big\}.
\]
Note that $h$ is a KMS filter if and only if $k(x):=\e^{x/4}h(x)$ is continuous and satisfies
$k(-x)=k^*(x)$.

\paragraph{Lindbladians with coherent term.}
Define the Bohr-frequency reweighting superoperators
\begin{equation}\label{eq:gt-CC}
	\CC_s[\vX]:=\frac{\ri}{2}\sum_{\xi\in B(\vH)}\tanh\!\Big(\frac{\beta\xi}{4}\Big)\vX_\xi,\qquad
	\CC_c[\vX]:=\sum_{\xi\in B(\vH)}\frac{\vX_\xi}{\cosh(\beta\xi/4)} .
\end{equation}
Since $(\vX_\xi)^\dagger=(\vX^\dagger)_{-\xi}$, $\tanh$ is odd and $\cosh$ is even, both maps send Hermitian
operators to Hermitian operators. For a completely positive map $\CT$ the corresponding KMS-detailed balanced Lindbladian in the Heisenberg picture is~\cite{chen2023ExactQGibbsSampler,gilyen2024QGeneralGlauberMetropolis,ding2024GibbsSamplingViaKMS}
\begin{equation}\label{eq:gt-L}
	\CL_{\CT}[\cdot]:=\CT[\cdot]-\tfrac12\{\CT[\vI],\cdot\}+\ri\big[\CC_s[\CT[\vI]],\cdot\big],
\end{equation}
where $\CC_s[\CT[\vI]]$ is the coherent term of the KMS-detailed-balanced Lindbladians as described in
\cite{chen2023ExactQGibbsSampler,gilyen2024QGeneralGlauberMetropolis}. If $\CT[\vI]$ commutes with
$\vH$, then $\CC_s[\CT[\vI]]=0$ and \eqref{eq:gt-L} reduces to \eqref{eq:KMS-Lindbladian-general-form}.

For a self-adjoint jump $\vK$ and a KMS filter $h$ let $\vL_h:=\sum_\nu h(\beta\nu)\vK_\nu$ and define the
\emph{coherent Lindbladian} $\LC^{\vK,h}:=\CL_{\CT}$ with $\CT[\vX]=\vL_h^\dagger\vX\vL_h$, and
$\vD_h:=\vL_h^\dagger\vL_h$. For a family $\{\vK^{(a)}\}_{a\in\mathcal A}$ of self-adjoint jumps let
$\LC^{h}:=\sum_{a}\LC^{\vK^{(a)},h}$; for $\{\vK^{(a)}\}=\{\vsigma^{(i,\alpha)}\}$ this is $\LC$ of the
main text with $f$ replaced by $h$.

The \emph{OFT Lindbladian} is $\LFs^{\vK}:=\CL_{\TFs^{\vK}}$ with $\TFs^{\vK}$ as in
\eqref{eq:oft-S-transform}, and $\LFs:=\sum_a\LFs^{\vK^{(a)}}$. In this section $f$ is the Metropolis
filter \eqref{eq:Metropolis}.
Then $\LFs^{\vK}$ is the Lindbladian of \cite{chen2023ExactQGibbsSampler}, and its
Schr\"odinger-picture adjoint is the Lindbladian $\CL_a$ of
\cite[Eq.~(2.19)]{bergamaschi2025FastMixingSpinChainsAllTemps} with Metropolis weight and energy width $\sigma/\beta$.

\paragraph{Dirichlet form and spectral gap.}
For a Lindbladian $\CL$ let $\CE_{\CL}(\vO):=\ip{\vO}{-\CL[\vO]}$ and
\[
	\gap(\CL):=\inf\Big\{\frac{\CE_{\CL}(\vO)}{\Var(\vO)}\colon \Var(\vO)>0\Big\}.
\]
For KMS-detailed-balanced
$\CL$ this is the spectral gap of $-\CL$. We use $\CE_{\mathrm C}^{h}:=\CE_{\LC^{h}}$ and
$\CE_{\mathrm F}:=\CE_{\LFs}$.

\subsection{Coherent Lindbladians and their Dirichlet forms}

In \eqref{eq:coherent-lindblad} and \Cref{lem:form}, the single-site Pauli
jumps are summed for the mean-field Hamiltonian, and the total coherent
correction vanishes. For a general Hamiltonian and a single self-adjoint
jump, $\vD_h$ need not commute with $\vH$; the formulas below retain the
correction in \eqref{eq:gt-L}.

\begin{lemma}[cf.\ {\cite[Section~4.2.1]{gilyen2024QGeneralGlauberMetropolis}}]\label{lem:gt-coherent}
    Let $\vK=\vK^\dagger$, let $h$ be a KMS filter and $k(x):=\e^{x/4}h(x)$.
    \begin{enumerate}[label=(\roman*)]
        \item\label{it:Symm}
        $\tilde\vL_h:=\vrho^{-1/4}\vL_h\vrho^{1/4}
        =\sum_\nu k(\beta\nu)\vK_\nu$ is Hermitian.
        \item\label{it:DisK}
        The discriminant \eqref{eq:DiscriminantDef} of $\LC^{\vK,h}$ is
        \[
            \DC^{\vK,h}[\vX]
            =\tilde\vL_h\vX\tilde\vL_h
             -\tfrac12\big\{\CC_c[\vD_h],\vX\big\}.
        \]
        Hence $\DC^{\vK,h}$ is Hilbert--Schmidt self-adjoint, and
        $\LC^{\vK,h}$ is KMS detailed balanced.
        \item\label{it:Bohrf}
        $\LC^{\vK,h}$ depends on $h$ only through $h|_{\beta\Omega_{\vK}}$,
        and $\LC^{\vK,th}=t^2\LC^{\vK,h}$ for $t\in\BR$.
    \end{enumerate}
\end{lemma}
\begin{proof}
    Part~\ref{it:Symm} follows from \eqref{eq:gt-bohr} and
    $k(-x)=k^*(x)$. Under the similarity transformation
    $\vX=\vrho^{1/4}\vO\vrho^{1/4}$, the gain term in
    \eqref{eq:gt-L} becomes $\tilde\vL_h\vX\tilde\vL_h$.
    For a Bohr component $(\vD_h)_\xi$, setting $z=\beta\xi/4$,
    the left and right multiplication coefficients become
    $-\tfrac12\e^{\mp z}(1\pm\tanh z)=-1/(2\cosh z)$,
    proving \ref{it:DisK}. Part~\ref{it:Bohrf} follows from the
    definition of $\vL_h$ and $\vL_{th}=t\vL_h$.
\end{proof}

\begin{lemma}[Dirichlet form of coherent Lindbladians]\label{lem:gt-dirichlet}
    Under the assumptions and with the notation of \Cref{lem:gt-coherent}, let
    $\tilde\vL_h(t):=\e^{\ri\vH t}\tilde\vL_h\e^{-\ri\vH t}$ and
    $\eta(t):=1/[\beta\cosh(2\pi t/\beta)]$, so that
    $\eta(t)\,\rmd t=\rmd\mu_\beta(t)$ as in \eqref{eq:fourier}.
    Then for every $\vO$,
    \[
    \begin{aligned}
        \CE_{\LC^{\vK,h}}(\vO)
        &=\int_\BR\eta(t)\,\nr{[\tilde\vL_h(t),\vO]}^2\,\rmd t\\
        &=\sum_{\nu,\nu'\in\Omega_{\vK}}
          \frac{k^*(\beta\nu)\,k(\beta\nu')}
               {2\cosh(\beta(\nu-\nu')/4)}
          \,\ip{[\vK_\nu,\vO]}{[\vK_{\nu'},\vO]}.
    \end{aligned}
    \]
    In particular, $\CE_{\LC^{\vK,h}}\ge0$.
\end{lemma}
\begin{proof}
    Write $k_\nu:=k(\beta\nu)$ and
    $w_{\nu\nu'}:=k_\nu^*k_{\nu'}/(2\cosh(\beta(\nu-\nu')/4))$;
    all double sums below run over $\Omega_{\vK}^2$.
    Expanding $\tilde\vL_h(t)=\sum_\nu\e^{\ri\nu t}k_\nu\vK_\nu$
    and using \eqref{eq:fourier} identifies the integral with the double sum.
    To identify this sum with the Dirichlet form, let
    $\vX:=\vrho^{1/4}\vO\vrho^{1/4}$. By \eqref{eq:gt-bohr},
    \[
        \vrho^{1/4}[\vK_\nu,\vO]\vrho^{1/4}
        =\e^{-\beta\nu/4}\vK_\nu\vX
         -\e^{\beta\nu/4}\vX\vK_\nu,
    \]
    hence
    \begin{align*}
        \sum_{\nu,\nu'}w_{\nu\nu'}
        \ip{[\vK_\nu,\vO]}{[\vK_{\nu'},\vO]}
        ={}&\Tr(\vX^\dagger\vM_1\vX)
           +\Tr(\vX^\dagger\vX\vM_2)\\
        &-\sum_{\nu,\nu'}w_{\nu\nu'}\Big(
          \e^{\beta(\nu'-\nu)/4}
          \Tr(\vX^\dagger\vK_\nu^\dagger\vX\vK_{\nu'})
          +\e^{\beta(\nu-\nu')/4}
          \Tr(\vX^\dagger\vK_{\nu'}\vX\vK_\nu^\dagger)\Big),
    \end{align*}
    where
    \[
        \vM_1:=\sum_{\nu,\nu'}w_{\nu\nu'}
          \e^{-\beta(\nu+\nu')/4}\vK_\nu^\dagger\vK_{\nu'},
        \qquad
        \vM_2:=\sum_{\nu,\nu'}w_{\nu\nu'}
          \e^{\beta(\nu+\nu')/4}\vK_{\nu'}\vK_\nu^\dagger.
    \]
    Since $k(-x)=k^*(x)$, we have $w_{-\nu',-\nu}=w_{\nu\nu'}$.
    The relabelling $(\nu,\nu')\mapsto(-\nu',-\nu)$ in $\vM_2$
    therefore gives $\vM_2=\vM_1$. Applying it to the second mixed term
    combines the two mixed terms into
    \[
        \sum_{\nu,\nu'}w_{\nu\nu'}\,
        2\cosh\!\big(\beta(\nu-\nu')/4\big)
        \Tr(\vX^\dagger\vK_\nu^\dagger\vX\vK_{\nu'})
        =\Tr(\vX^\dagger\tilde\vL_h\vX\tilde\vL_h).
    \]
    Moreover,
    $\vD_h=\sum_{\nu,\nu'}\e^{-\beta(\nu+\nu')/4}
    k_\nu^*k_{\nu'}\vK_\nu^\dagger\vK_{\nu'}$,
    and $\vK_\nu^\dagger\vK_{\nu'}$ has Bohr frequency $\nu'-\nu$,
    so $\vM_1=\tfrac12\CC_c[\vD_h]$.
    The expanded expression is thus
    $\ipc{\vX}{-\DC^{\vK,h}[\vX]}
    =\CE_{\LC^{\vK,h}}(\vO)$ by
    \Cref{lem:gt-coherent}\ref{it:DisK}.
\end{proof}
The corresponding OFT identity is given in
\cite[Lemma~C.2]{rouze2024EffThermalizationGibbsSamp}.

\subsection{Gap transfer}

\begin{theorem}[Linear gap transfer]\label{thm:gapTransfer}
	Let $f_M$ be the Metropolis filter \eqref{eq:Metropolis}, $\sigma>0$, and let $\{\vK^{(a)}\}_{a\in A}$ be
	self-adjoint jumps are such that $\vK^{(a)}_\nu=0$ for all $\nu> \Lambda/\beta$ and $a\in A$. Let $\LFs$ be the OFT Lindbladian~\eqref{eq:oft-S-transform} built from these jumps with filter $f_M$ and width parameter $\sigma$, and for a KMS filter $h$ let $\LC^{h}$ be the coherent Lindbladian built from the same jumps with filter $h$. With $\hat g_\sigma$ the Gaussian of \eqref{eq:OFT}, define for $\omega\in\BR$ the functions
	\begin{equation}\label{eq:gt-homega}
		k_{\omega,\pm}(x):=\tfrac{1}{\sqrt2}\Big(\e^{\pm\ri\pi/4}\,\hat g_\sigma\big(\omega-\tfrac x\beta\big)+\e^{\mp\ri\pi/4}\,\hat g_\sigma\big(\omega+\tfrac x\beta\big)\Big),
		\qquad
		h_{\omega,\pm}(x):=\e^{-x/4}\,k_{\omega,\pm}(x),
	\end{equation}
	and the weight
	\begin{equation}\label{eq:gt-w}
		w(\omega):=\e^{-\sigma^2/8}\,|f_M(\beta\omega)|^2\,\e^{\beta\omega/2}=\e^{-\sigma^2/8}\,\e^{-\beta|\omega|/2}.
	\end{equation}
	Each $h_{\omega,\pm}$ is a KMS filter and $w$ is even. Furthermore, let $\Phi$ denote the standard normal distribution function and let $V\subseteq\CB(\Hn)$ be a subspace invariant under $\LFs$, and let $\omega_{\mathrm{trn}}\in(0,\infty]$. If for some $\kappa>0$
	\begin{equation}\label{eq:unifCohBound}
		\CE^{h}_{\mathrm C}(\vO) \geq \kappa\,\min_{x\in[-\Lambda,0]}|h(x)|^2\,\Var(\vO)
		\qquad\text{for all }\vO\in V\text{ and all }h\in\{h_{\omega,\pm}\colon|\omega|\le\omega_{\mathrm{trn}}\},
	\end{equation}
    \begin{equation}\label{eq:gt-concl}
		\gamma(\LFs,V)\ \ge\ \kappa\,c_{\mathrm{eff}}(\omega_{\mathrm{trn}}),
		\qquad
		c_{\mathrm{eff}}(\omega_{\mathrm{trn}}):=\int_{|\omega|\le\omega_{\mathrm{trn}}}w(\omega)\,\min_{x\in[-\Lambda,0]}|h_{\omega,+}(x)|^2\,\rmd\omega ,
	\end{equation}
	and
	\begin{equation}\label{eq:gt-ceff}
		c_{\mathrm{eff}}(\infty) \geq 1-\Phi\Big(\max\Big\{\frac{\sigma}{2},\frac{\Lambda}{2\sigma}\Big\}\Big),
		\qquad
		c_{\mathrm{eff}}(\omega_{\mathrm{trn}}) \geq \frac{\omega_{\mathrm{trn}}\,\beta}{\sigma\sqrt{2\pi}}\,
		\exp\!\Big(-\frac{(\max\{\Lambda,\beta\omega_{\mathrm{trn}}\}+\sigma^2/2)^2}{2\sigma^2}\Big)\qquad(\omega_{\mathrm{trn}}<\infty).
	\end{equation}
\end{theorem}
\begin{proof}
	\emph{Step 1 (Dirichlet form of the OFT Lindbladian).} For each jump $\vK=\vK^{(a)}$,
	\cite[Lemma~VI.1]{bergamaschi2025FastMixingSpinChainsAllTemps} 
    applied with energy width $\sigma/\beta$ gives
	\[
	\CE_{\LFs^{\vK}}(\vO)=\iint_{\BR^2}\eta(t)\,w(\omega)\,\nr{[\widehat{\vK}(\omega,t),\vO]}^2\,\rmd t\,\rmd\omega,
	\qquad w(\omega):=\e^{-\sigma^2/8}\e^{-\beta|\omega|/2},
	\]
	where $\widehat{\vK}(\omega,t):=\e^{\ri\vH t}\widehat{\vK}(\omega)\e^{-\ri\vH t}$ and $\eta$ is as in \Cref{lem:gt-dirichlet}.

	\emph{Step 2 (symmetrization).} Since $\vK=\vK^\dagger$ and $\hat g_\sigma$ is real and even, \eqref{eq:OFT} gives
	$\widehat{\vK}(-\omega,t)=\widehat{\vK}(\omega,t)^\dagger$, and $w$ is even. For any operator $\vY$ and
	$\vA_\pm:=\frac{1}{\sqrt2}\big(\e^{\pm\ri\pi/4}\vY+\e^{\mp\ri\pi/4}\vY^\dagger\big)$ the cross terms cancel and
	$\nr{[\vY,\vO]}^2+\nr{[\vY^\dagger,\vO]}^2=\sum_\pm\nr{[\vA_\pm,\vO]}^2$. Applying this to $\vY=\widehat{\vK}(\omega,t)$
	and using that $w$ is even,
	\[
	\CE_{\LFs^{\vK}}(\vO)=\frac12\int_\BR w(\omega)\sum_{\pm}\int_\BR\eta(t)\,\nr{[\vA_{\omega,\pm}(t),\vO]}^2\,\rmd t\,\rmd\omega,
	\qquad
	\vA_{\omega,\pm}:=\sum_\nu k_{\omega,\pm}(\beta\nu)\vK_\nu ,
	\]
	with $k_{\omega,\pm}$ as in \eqref{eq:gt-homega} and $\vA_{\omega,\pm}(t):=\e^{\ri\vH t}\vA_{\omega,\pm}\e^{-\ri\vH t}$.

	\emph{Step 3 (each piece is a coherent Dirichlet form).} We have $k_{\omega,\pm}(-x)=k_{\omega,\pm}^*(x)$, so
	$h_{\omega,\pm}(x):=\e^{-x/4}k_{\omega,\pm}(x)$ is a KMS filter and $\vA_{\omega,\pm}=\tilde\vL_{h_{\omega,\pm}}$
	(Lemma~\ref{lem:gt-coherent}\ref{it:Symm}). By Lemma~\ref{lem:gt-dirichlet}, the inner $t$-integral in Step~2 equals
	$\CE_{\LC^{\vK,h_{\omega,\pm}}}(\vO)$. Since $h_{\omega,\pm}$ does not depend on the jump, summing over $a$ gives
	\begin{equation}\label{eq:gt-mixture-EF}
        \CE_{\mathrm F}(\vO)=\frac12\int_\BR w(\omega)\sum_\pm\CE^{h_{\omega,\pm}}_{\mathrm C}(\vO)\,\rmd\omega .
    \end{equation}

	\emph{Step 4 (conclusion).}
    A direct calculation of $|k_{\omega,\pm}|^2$ gives
    \begin{equation}
        \label{eq:gt-k2}
        \labs{k_{\omega,\pm}}^2 = \frac{1}{2}\left(\hat{g}_\sigma\left(\omega - \frac{x}{\beta}\right)^2 + \hat{g}_\sigma\left(\omega + \frac{x}{\beta}\right)^2\right),
    \end{equation}
    so
	$|h_{\omega,\pm}(x)|^2=\frac{\e^{-x/2}}{2}\big(\hat g_\sigma(\omega-\frac x\beta)^2+\hat g_\sigma(\omega+\frac x\beta)^2\big)$ for both signs.
	Using $\hat g_\sigma(\omega-\frac x\beta)=\hat g_\sigma(0)\e^{-(\beta\omega-x)^2/(4\sigma^2)}$ and completing the square,
	\begin{equation}\label{eq:gt-wh2}
	w(\omega)\,|h_{\omega,\pm}(x)|^2=\frac{\hat g_\sigma(0)^2}{2}\Big[\min\{1,\e^{-\beta\omega}\}\,\e^{-\frac{(x-\beta\omega+\sigma^2/2)^2}{2\sigma^2}}
	+\min\{1,\e^{\beta\omega}\}\,\e^{-\frac{(x+\beta\omega+\sigma^2/2)^2}{2\sigma^2}}\Big].
	\end{equation}
    Every term in \eqref{eq:gt-mixture-EF} is nonnegative, since $w\ge0$ and $\CE^{h}_{\mathrm C}\ge0$ for every KMS
	filter $h$ (\Cref{lem:gt-dirichlet}). Hence for $\vO\in V$ we may drop the contributions with $|\omega|>\omega_{\mathrm{trn}}$ and
	insert \eqref{eq:unifCohBound} for the remaining $\omega$
	\begin{align*}
	    \CE_{\mathrm F}(\vO)\ &\ge\ \frac12\int_{|\omega|\le\omega_{\mathrm{trn}}}w(\omega)\sum_\pm\CE^{h_{\omega,\pm}}_{\mathrm C}(\vO)\,\rmd\omega \\
        &\ge\ \kappa\,\Var(\vO)\int_{|\omega|\le\omega_{\mathrm{trn}}}w(\omega)\min_{x\in[-\Lambda,0]}|h_{\omega,+}(x)|^2\,\rmd\omega \tag{By~\eqref{eq:unifCohBound} and since $|h_{\omega,+}(x)| = |h_{\omega,-}(x)|$}\\
        &=\kappa\,c_{\mathrm{eff}}(\omega_{\mathrm{trn}})\,\Var(\vO).
	\end{align*}
    It remains to prove \eqref{eq:gt-ceff}; recall $\hat g_\sigma(0)^2=\beta/(\sigma\sqrt{2\pi})$.

	\emph{Case $\omega_{\mathrm{trn}}=\infty$.} For $\omega\leq0$ keep only the first term of \eqref{eq:gt-wh2}, a Gaussian in $x$
	centred at $y:=\beta\omega-\sigma^2/2\le-\sigma^2/2$, whose minimum over $x\in[-\Lambda,0]$ is attained at the endpoint
	farther from $y$; for $\omega\ge0$ keep only the second term, which contributes the same after $\omega\mapsto-\omega$. Thus
	\[
	c_{\mathrm{eff}}(\infty)\ \ge\ \frac{1}{\sigma\sqrt{2\pi}}\int_{-\infty}^{-\sigma^2/2}\e^{-\max\{|y|,|y+\Lambda|\}^2/(2\sigma^2)}\,\rmd y
	\ \ge\ 1-\Phi\Big(\max\Big\{\frac\sigma2,\frac{\Lambda}{2\sigma}\Big\}\Big),
	\]
	where for $\Lambda\le\sigma^2$ the maximum equals $|y|$ on the whole range of integration, and for $\Lambda\ge\sigma^2$ we
	restrict the integral to $y\le-\Lambda/2$.

	\emph{Case $\omega_{\mathrm{trn}}<\infty$.} For $\omega\in[-\omega_{\mathrm{trn}},0]$ keep only the first term of~\eqref{eq:gt-wh2}, whose
	prefactor is $\min\{1,\e^{-\beta\omega}\}=1$. For $x\in[-\Lambda,0]$ we have $x-\beta\omega\in[-\Lambda,\beta\omega_{\mathrm{trn}}]$,
	hence $|x-\beta\omega+\sigma^2/2|\leq\max\{\Lambda,\beta\omega_{\mathrm{trn}}\}+\sigma^2/2$, and therefore
	\[
    	w(\omega)\min_{x\in[-\Lambda,0]}|h_{\omega,+}(x)|^2\ \geq\ \frac{\hat g_\sigma(0)^2}{2}\,
    	\exp\!\Big(-\frac{(\max\{\Lambda,\beta\omega_{\mathrm{trn}}\}+\sigma^2/2)^2}{2\sigma^2}\Big)\qquad(\omega\in[-\omega_{\mathrm{trn}},0]).
	\]
	For $\omega\in[0,\omega_{\mathrm{trn}}]$ the second term of \eqref{eq:gt-wh2} gives the same bound. Integrating over
	$\omega\in[-\omega_{\mathrm{trn}},\omega_{\mathrm{trn}}]$ and applying the derived bounds gives~\eqref{eq:gt-ceff}.
\end{proof}

\begin{corollary}[Sector-wise transfer of the coherent bounds]\label{cor:gt-sectors}
	Let $\beta>0$, $\sigma>0$, and let $\LFs$ be the OFT Lindbladian with Metropolis filter and the single-site Pauli jumps, so that $\Lambda=\beta$. Then for all $n\ge2$:
	\begin{align}
		\gamma(\LFs,\Rn^\perp)\ &\ge\ \frac{4}{n}\Big[1-\Phi\Big(\max\Big\{\frac\sigma2,\frac{\beta}{2\sigma}\Big\}\Big)\Big],\label{eq:gt-Rperp}\\
		\gamma(\LFs,\Pn^\perp)\ &\ge\ \frac{\beta}{54\,\sigma\sqrt{2\pi}}\;\exp\!\Big(-\beta-\frac{2\beta^2}{\sigma^2}-\frac{(\beta+\sigma^2/2)^2}{2\sigma^2}\Big),\label{eq:gt-Pperp}\\
		\gamma(\LFs,V)\ &\ge\ \Big[1-\Phi\Big(\max\Big\{\frac\sigma2,\frac{\beta}{2\sigma}\Big\}\Big)\Big]\,\gamma(\LC,V)\qquad(V\in\{\Pn,\Rn\}).\label{eq:gt-sym}
	\end{align}
\end{corollary}
\begin{proof}
	In each case we verify the hypothesis~\eqref{eq:unifCohBound} of~\Cref{thm:gapTransfer} on the stated sector $V$, with a suitable $\kappa$ and truncation $\omega_{\mathrm{trn}}$, and read off the conclusion from~\eqref{eq:gt-ceff}. The sectors are invariant under $\LFs$ and all $\LC$ by \Cref{cor:reductions}, and on $\Pn^\perp,\Rn^\perp$ we have $\Var(\vO)=\nr{\vO}^2$ since $\vI\in\Pn\cap\Rn$.

	\emph{The sector $\Rn^\perp$ ($\kappa=4/n$, $\omega_{\mathrm{trn}}=\infty$):} By \Cref{prop:R-complement}, applied with an arbitrary KMS filter $h$, $\EC^{h}(\vO)\ge\frac{4h(0)^2}{n}\nr{\vO}^2$ for $\vO\in\Rn^\perp$. Since $0\in[-\beta,0]$, $h(0)^2\ge\min_{x\in[-\beta,0]}|h(x)|^2$, so \eqref{eq:unifCohBound} holds for every KMS filter with $\kappa=4/n$. The first bound in \eqref{eq:gt-ceff} gives \eqref{eq:gt-Rperp}.

	\emph{The sector $\Pn^\perp$ ($\kappa=\e^{-\beta-2\beta^2/\sigma^2}/54$, $\omega_{\mathrm{trn}}=1$):} By \Cref{thm:high-temp}, applied with an arbitrary KMS filter $h$, $\EC^{h}(\vO)\ge\frac{m_h^4}{54N_h^2}\nr{\vO}^2$ for $\vO\in\Pn^\perp$, with $m_h=\inf_{|x|\le\beta}|\tilde h(x)|$, $N_h=\sup_{|x|\le\beta}|h(x)|$. Let
    \begin{equation*}
        R_h\coloneqq\frac{\sup_{|x|\le\beta}|\tilde h(x)|^2}{\inf_{|x|\le\beta}|\tilde h(x)|^2},
    \end{equation*}
    since $|h(x)|^2=\e^{-x/2}|\tilde h(x)|^2$ and $|\tilde h(-x)|=|\tilde h(x)|$ we can further show that
	\[
		N_h^2\le \e^{\beta/2}\sup_{|x|\le\beta}|\tilde h(x)|^2=\e^{\beta/2}R_h\,m_h^2,
		\qquad
		m_h^2=\inf_{x\in[-\beta,0]}\e^{x/2}|h(x)|^2\ge \e^{-\beta/2}\min_{x\in[-\beta,0]}|h(x)|^2,
	\]
	hence
    \begin{equation}\label{eq:kappa-required-bound}
        \frac{m_h^4}{54N_h^2}\ge\frac{\e^{-\beta}}{54R_h}\min_{x\in[-\beta,0]}|h(x)|^2.
    \end{equation}
    For the filters $h_{\omega,\pm}$ of~\eqref{eq:gt-homega}, the symmetrized filter is $k_{\omega,\pm}$, and by \eqref{eq:gt-k2},
	\begin{equation*}
	    |k_{\omega,\pm}(x)|^2=\frac{\hat g_\sigma(0)^2}{2}\Big(\e^{-\frac{(\beta\omega-x)^2}{2\sigma^2}}+\e^{-\frac{(\beta\omega+x)^2}{2\sigma^2}}\Big). \tag{Since $\hat g_\sigma(u)^2=\hat g_\sigma(0)^2\e^{-\beta^2u^2/(2\sigma^2)}$}
	\end{equation*}
	For $|\omega|\le1$ and $|x|\le\beta$ we have $|\beta\omega\mp x|\le2\beta$, so both exponents lie in $[0,2\beta^2/\sigma^2]$; hence
	$\sup_{|x|\le\beta}|k_{\omega,\pm}(x)|^2\le\hat g_\sigma(0)^2$ and $\inf_{|x|\le\beta}|k_{\omega,\pm}(x)|^2\ge\hat g_\sigma(0)^2\e^{-2\beta^2/\sigma^2}$, i.e.\
	$R_{h_{\omega,\pm}}\le \e^{2\beta^2/\sigma^2}$. Consequently from~\eqref{eq:kappa-required-bound} it follows that~\eqref{eq:unifCohBound} holds on $\{h_{\omega,\pm}\colon|\omega|\le1\}$ with
    \begin{equation*}
        \kappa=\frac{\e^{-\beta-2\beta^2\sigma^2}}{54}.
    \end{equation*}
    The second bound in~\eqref{eq:gt-ceff} with $\Lambda=\beta$ and $\omega_{\mathrm{trn}}=1$ gives
    \begin{equation*}
        c_{\mathrm{eff}}(1)\ge\frac{\beta}{\sigma\sqrt{2\pi}}\,\e^{-(\beta+\sigma^2/2)^2/(2\sigma^2)},
    \end{equation*}
    and the product $\kappa\,c_{\mathrm{eff}}(1)$ is \eqref{eq:gt-Pperp}.

    \emph{The sectors $\Pn,\Rn$ ($\kappa=\gamma(\LC,V)$, $\omega_{\mathrm{trn}}=\infty$):}
	Let $h$ be a KMS filter and $c_h:=\min_{x\in[-\beta,0]}|h(x)|^2$, and let $f$ be the filter of the main text, a KMS filter with
	$|f|\le1$. For any KMS filter $g$ the condition \eqref{eq:f-KMS} gives $|g(x)|^2=e^{-x}|g(-x)|^2$, hence
	$|\widetilde g(x)|^2=e^{x/2}|g(x)|^2=e^{-|x|/2}\,|g(-|x|)|^2$ for all $x\in\BR$. For $|\nu|\le1$ the point $-\beta|\nu|$ lies in
	$[-\beta,0]$, so applying this to $g=h$ and to $g=f$ (with $|f|\le1$) yields
	\begin{equation}\label{eq:gt-weight-comparison}
		|\widetilde h(\beta\nu)|^2=e^{-\beta|\nu|/2}\,|h(-\beta|\nu|)|^2\ \ge\ c_h\,e^{-\beta|\nu|/2}
		\ \ge\ c_h\,e^{-\beta|\nu|/2}\,|f(-\beta|\nu|)|^2=c_h\,|\widetilde f(\beta\nu)|^2\qquad(|\nu|\le1).
	\end{equation}
	Now let $\vX\in\comm{\vH_n}$. \Cref{lem:bohr-diagonalization} holds for every KMS filter, so
	\[
		\EC^{h}(\vX)=\frac12\sum_{i,\alpha}\sum_{\nu}|\widetilde h(\beta\nu)|^2\,\nr{[\vsigma^{(i,\alpha)}_\nu,\vX]}^2,
		\qquad
		\EC(\vX)=\frac12\sum_{i,\alpha}\sum_{\nu}|\widetilde f(\beta\nu)|^2\,\nr{[\vsigma^{(i,\alpha)}_\nu,\vX]}^2,
	\]
	where by \Cref{fact:pauli-selection-rule} only $|\nu|\le1$ contributes. The summands $\nr{[\vsigma^{(i,\alpha)}_\nu,\vX]}^2\ge0$
	are the same in both sums and do not depend on the filter, so by~\eqref{eq:gt-weight-comparison},
	\[
		\EC^{h}(\vX)=\frac12\sum_{i,\alpha,\nu}|\widetilde h(\beta\nu)|^2\nr{[\vsigma^{(i,\alpha)}_\nu,\vX]}^2
		\ \ge\ c_h\cdot\frac12\sum_{i,\alpha,\nu}|\widetilde f(\beta\nu)|^2\nr{[\vsigma^{(i,\alpha)}_\nu,\vX]}^2
		=c_h\,\EC(\vX)\qquad(\vX\in\comm{\vH_n}).
	\]
	Finally, $V\subseteq\comm{\vH_n}$ by \Cref{eq:RnInCHn,eq:PnInCHn}, and the variational characterization of the sector gap
	(as in \eqref{eq:gap-rayleigh-quotient}, restricted to the invariant sector $V$) reads $\EC(\vO)\ge\gamma(\LC,V)\Var(\vO)$ for all
	$\vO\in V$. Hence for every $\vO\in V$
	\[
		\EC^{h}(\vO)\ \ge\ c_h\,\EC(\vO)\ \ge\ c_h\,\gamma(\LC,V)\,\Var(\vO),
	\]
	which is \eqref{eq:unifCohBound} with $\kappa=\gamma(\LC,V)$ and $\omega_{\mathrm{trn}}=\infty$. The first bound in \eqref{eq:gt-ceff} gives
	\eqref{eq:gt-sym}.
\end{proof}

\end{document}